\documentclass[preprint]{elsarticle-local}
\usepackage{fullpage}
\usepackage[disable,colorinlistoftodos,bordercolor=orange,backgroundcolor=orange!20,linecolor=orange,textsize=scriptsize]{todonotes}

\usepackage[utf8]{inputenc}
\usepackage[T1]{fontenc}

\usepackage{amsfonts,amssymb,mathrsfs,amsmath,amsthm,amscd,epsfig,amstext,color,graphics,verbatim}
\usepackage[colorlinks=true,breaklinks=true,bookmarks=true,urlcolor=blue,
     citecolor=blue,linkcolor=blue,bookmarksopen=false,draft=false]{hyperref}

\usepackage[capitalise,noabbrev]{cleveref}

\crefname{assumption}{Assumption}{Assumptions}
\Crefname{assumption}{Assumption}{Assumptions}

\renewcommand{\ge}{\geqslant}
\renewcommand{\le}{\leqslant}
\renewcommand{\geq}{\geqslant}
\renewcommand{\leq}{\leqslant}

\usepackage{accents}

\usepackage{tikz}

\usetikzlibrary{arrows}
\usetikzlibrary{patterns}
\usetikzlibrary{calc}
\usetikzlibrary{shapes}
\usetikzlibrary{positioning}

\tikzset{grid/.style={gray!30,very thin}}
\tikzset{axis/.style={gray!50,->,>=stealth'}}
\tikzset{convex/.style={draw=none,fill=lightgray,fill opacity=0.7}}
\tikzset{convexborder/.style={very thick}}
\tikzset{point/.style={blue!50}}
\tikzset{hs/.style={fill opacity=0.3,fill=orange,draw=none}}
\tikzset{hsborder/.style={orange,ultra thick,dashdotted}}

\newcommand{\overbar}[1]{\mkern 1.5mu\overline{\mkern-1.5mu#1\mkern-1.5mu}\mkern 1.5mu}
\newcommand{\Z}{\mathbb{Z}}
\newcommand{\N}{\mathbb{N}} 
\newcommand{\R}{\mathbb{R}}    
\newcommand{\Q}{\mathbb{Q}}

\newcommand{\trop}{\bar{\R}}
\newcommand{\dgraph}{\mathscr{G}}

\newcommand{\gameval}{\chi}
\newcommand{\set}{\mathscr{S}}
\newcommand{\setp}{\mathscr{S}'}
\newcommand{\setI}{\mathscr{S}_1}
\newcommand{\setII}{\mathscr{S}_2}
\newcommand{\dominion}{\mathscr{D}}
\newcommand{\dominionp}{\mathscr{D}'}
\newcommand{\dominionI}{\mathscr{D}_1}
\newcommand{\dominionII}{\mathscr{D}_2}
\newcommand{\domMaxVal}{{\mathscr{D}_{\max}}}

\newcommand{\edges}{\mathscr{E}}
\newcommand{\vertexset}{\mathscr{V}}
\newcommand{\vertexsetD}{{\mathscr{V}_D}}
\newcommand{\vertexsetT}{{\mathscr{V}_T}}
\newcommand{\vertexsetP}{{\mathscr{V}_P}}
\newcommand{\vertexsetMin}{{\mathscr{V}_{\mathrm{Min}}}}
\newcommand{\vertexsetMax}{{\mathscr{V}_{\mathrm{Max}}}}
\newcommand{\vertexsetNat}{{\mathscr{V}_{\mathrm{Nat}}}}
\newcommand{\vertexsubset}[2][]{\mathscr{V}_{\mathrm{#1}}^{#2}}
\newcommand{\vertexsubsetit}[2][]{\mathscr{V}_{#1}^{#2}}
\newcommand{\setAlgo}{\mathscr{S}}
\newcommand{\setpAlgo}{\mathscr{S}'}
\newcommand{\dominionAlgo}{\mathscr{D}}
\newcommand{\shapley}{F}

\newcommand{\bias}{u}

\newcommand{\mybot}{\mathbf{b}}
\newcommand{\mytop}{\mathbf{t}}

\newcommand{\Extend}{\textsc{Extend}}
\newcommand{\TopClass}{\textsc{TopClass}}
\newcommand{\dunion}{\uplus}

\newcommand{\supnorm}[1]{\|#1\|_{\infty}}

\newcommand{\abs}[1]{|{#1}|}
\newcommand{\card}[1]{|{#1}|}
\newcommand{\pert}{\lambda}
\newcommand{\zero}{{-\infty}}
\newcommand{\comd}{M}
\newcommand{\rstates}{s}
\newcommand{\Payoff}{W}
\newcommand{\transpose}{\intercal}
\newcommand{\ind}{\mathbf{1}}
\newcommand{\polyh}{\mathcal{W}}
\newcommand{\conv}{\mathrm{conv}}
\newcommand{\cone}{\mathrm{cone}}

\usepackage{mathtools}
\usepackage{mathrsfs}

\usepackage{phdalgo2}
\crefformat{equation}{\textup{#2(#1)#3}}
\crefrangeformat{equation}{\textup{#3(#1)#4--#5(#2)#6}}
\crefmultiformat{equation}{\textup{#2(#1)#3}}{ and \textup{#2(#1)#3}}
{, \textup{#2(#1)#3}}{, and \textup{#2(#1)#3}}
\crefrangemultiformat{equation}{\textup{#3(#1)#4--#5(#2)#6}}%
{ and \textup{#3(#1)#4--#5(#2)#6}}{, \textup{#3(#1)#4--#5(#2)#6}}{, and \textup{#3(#1)#4--#5(#2)#6}}

\Crefformat{equation}{#2Equation~\textup{(#1)}#3}
\Crefrangeformat{equation}{Equations~\textup{#3(#1)#4--#5(#2)#6}}
\Crefmultiformat{equation}{Equations~\textup{#2(#1)#3}}{ and \textup{#2(#1)#3}}
{, \textup{#2(#1)#3}}{, and \textup{#2(#1)#3}}
\Crefrangemultiformat{equation}{Equations~\textup{#3(#1)#4--#5(#2)#6}}%
{ and \textup{#3(#1)#4--#5(#2)#6}}{, \textup{#3(#1)#4--#5(#2)#6}}{, and \textup{#3(#1)#4--#5(#2)#6}}

\newcommand{\ie}{i.e.}

\newcommand*\maxstate[1]{\tikz[baseline=(char.base)]{
            \node[shape=circle,draw,inner sep=1.1pt,minimum size=0.5cm] (char) {#1};}}
\newcommand*\minstate[1]{\tikz[baseline=(char.base)]{
            \node[shape=rectangle,draw,inner sep=2.1pt,minimum size=0.5cm] (char) {#1};}}

\newcommand{\ec}{\operatorname{erg}}
\newcommand{\specrad}{\rho}
\newcommand{\minstates}{n}

\newcommand{\lcw}{\underline{\mathrm{cw}}}
\newcommand{\ucw}{\overline{\mathrm{cw}}}

\newcommand{\unitvector}{\mathbf{1}}
\newcommand{\zerovector}{\mathbf{0}}

\newcommand{\defi}{\coloneqq}
\newcommand{\oneto}[1]{[#1]}

\newcommand{\HNorm}[1]{{\| #1 \|}_{\rm{H}}}

\newcommand{\states}{n}

\newcommand{\transition}{P}
\newcommand{\nstates}{n}

\newcommand{\invmes}{\pi}

\newcommand{\enc}[1]{\langle #1 \rangle}
\newcommand{\Enc}{\phi}

\newcommand{\apshapley}{\tilde{\shapley}} \newcommand{\mshapley}{T}
\newcommand{\Weights}{W} \newcommand{\eweight}{m}
\newcommand{\mgameval}{V^\infty}
\newcommand{\recshapley}{\hat{\shapley}}

\newcommand{\revproj}{\mathbf{i}}
\newcommand{\proj}{\mathbf{p}}
\DeclareMathOperator{\separ}{sep}

\newtheorem{theorem}{Theorem}
\newtheorem{proposition}{Proposition}
\newtheorem{corollary}{Corollary}
\newtheorem{lemma}{Lemma}
\newtheorem{assumption}{Assumption}
\newtheorem{definition}{Definition}
\theoremstyle{remark}
\newtheorem{remark}{Remark}
\newtheorem{example}{Example}

\begin{document}

\title{Universal Complexity Bounds Based on Value Iteration for Stochastic Mean Payoff Games and Entropy Games}

\author[1]{Xavier Allamigeon}
\ead{xavier.allamigeon@inria.fr}

\author[1]{Stéphane Gaubert\corref{cor1}}
\ead{stephane.gaubert@inria.fr}

\author[2]{Ricardo D.~Katz}
\ead{katz@cifasis-conicet.gov.ar}

\author[3]{Mateusz Skomra}
\ead{mateusz.skomra@laas.fr}

\cortext[cor1]{Corresponding author}

\affiliation[1]{organization={INRIA and CMAP, \'Ecole polytechnique, IP Paris, CNRS},
city = {Palaiseau},
country = {France}}

\affiliation[2]{organization={CIFASIS-CONICET},
addressline = {Bv. 27 de Febrero 210 bis},
postcode = {2000},
city = {Rosario},
country = {Argentina}}

\affiliation[3]{organization={LAAS-CNRS, Universit\'e de Toulouse, CNRS},
city = {Toulouse},
country = {France}}

\begin{abstract}
  We develop value iteration-based algorithms to solve in a unified
  manner different classes of combinatorial zero-sum games with
  mean-payoff type rewards.  These algorithms rely on an oracle,
  evaluating the dynamic programming operator up to a given precision.
  We show that the number of calls to the oracle needed to determine
  exact optimal (positional) strategies is, up to a factor polynomial
  in the dimension, of order $R/\operatorname{sep}$, where the
  ``separation'' $\operatorname{sep}$ is defined as the minimal
  difference between distinct values arising from strategies, and $R$
  is a metric estimate, involving the norm of approximate sub and
  super-eigenvectors of the dynamic programming operator.  We
  illustrate this method by two applications. The first one is a new
  proof, leading to improved complexity estimates, of a theorem of
  Boros, Elbassioni, Gurvich and Makino, showing that turn-based mean-payoff games with a fixed number of random positions can be solved
  in pseudo-polynomial time.  The second one concerns entropy games, a
  model introduced by Asarin, Cervelle, Degorre, Dima, Horn and
  Kozyakin.  The {\em rank} of an entropy game is defined as the
  maximal rank among all the ambiguity matrices determined by
  strategies of the two players.  We show that entropy games with a
  fixed rank, in their original formulation, can be solved in
  polynomial time, and that an extension of entropy games
  incorporating weights can be solved in pseudo-polynomial time under
  the same fixed rank condition.
  \end{abstract}


\maketitle

\section{Introduction}
\subsection{Motivation}
Deterministic and turn-based stochastic Mean-Payoff games
are fundamental classes of games
with an unsettled complexity.
They belong
to the complexity class NP $\cap$ coNP~\cite{condon,karzanovlebedev,zwick_paterson}
but they are not known
to be polynomial-time solvable. Various algorithms have
been developed and analyzed.  The pumping
algorithm is a pseudo-polynomial iterative scheme
introduced by Gurvich, Karzanov and Khachiyan~\cite{gurvich}
to solve the optimality equation of deterministic mean-payoff games.
Zwick and Paterson~\cite{zwick_paterson}
derived peudo-polynomial bounds for the same games
by analyzing value iteration.  Friedmann showed
that policy iteration,  originally introduced by Hoffman and Karp in the
setting of zero-sum games~\cite{hoffman_karp}, and albeit being experimentally fast on typical instances, is generally exponential~\cite{friedmanncex}.  \todo{add other algorithms}
We refer to the survey of~\cite{andersson_miltersen} and to the recent
collective synthesis~\cite{collectivegraphs} for more
information and additional references.

Entropy games have been introduced by Asarin et al.~\cite{asarin_entropy}. 
They are combinatorial games, in which one player, called Tribune,
wants to maximize a topological entropy, whereas its opponent, called Despot,
wishes to minimize it. This topological entropy quantifies the freedom
of a half-player, called People. 
Although the formalization of entropy games is recent,
specific classes or variants of entropy games appeared earlier in
several fields, including the control of branching processes, population dynamics and growth maximization~\cite{Sladky1976,rothblumwhittle,rothblum,zijmjota}, 
risk sensitive control~\cite{Howard-Matheson,anantharam}, mathematical finance~\cite{Akian2001}, or matrix multiplication games~\cite{asarin_entropy}. 
Asarin et al.\ showed that entropy games also belong to the class NP $\cap$ coNP. Akian et al.\ showed
in~\cite{entropygamejournal} that entropy games reduce to ordinary stochastic mean-payoff games with infinite action spaces (actions consist of probability
measures and the payments are given by relative entropies), and deduced
that the subclass of entropy games in which Despot
has a fixed number of significant positions (positions with a non-trivial choice)
can be solved in polynomial time. The complexity of entropy games
without restrictions on the number of (significant) Despot positions
is an open problem.

\subsection{Main Results}
We develop value iteration-based algorithms
to solve in a unified manner different classes of combinatorial zero-sum games with
mean-payoff type rewards. These algorithms rely on an oracle,
evaluating approximately the dynamic programming operator of the game.
Our main results include universal estimates, providing explicit bounds
for the error of approximation of the value, as a function of two
characteristic quantities, of a metric nature. The first one
is the {\em separation} $\separ$, defined as the minimal difference
between distinct values induced by (positional) strategies. 
The second one, $R$, is defined in terms of the norm of approximate
sub and super-optimality certificates. These certificates
are vectors, defined as sub or super-solutions of non-linear eigenproblems.
For games such that the mean payoff is independent of the initial state,
we show that we can decide the winner in $O(R/\separ)$ iterations (\Cref{Th:NitsBound,Th:NitsBound:approx}). Morever,
an $\epsilon$-approximation of the mean payoff can be
computed in $O(R/\epsilon)$ iterations (\Cref{pr:approx_constant_value}).
Then, for concrete classes of games with finite action spaces, we obtain (exact) optimal strategies in a number of calls to the oracle bounded by the ratio $R/\separ$, up to a factor polynomial in the number
of states.\todo[color=red!30]{RK : [previously quoted : see \cref{Th:NitsBound,topclass_oracle,cor-perturb}] Are these the right results to refer to here? SG: no, thanks, I changed the discussion and changed the references, check.}
More precisely, we provide the two following applications of this principle.

The first application is a new
proof of an essential part of the theorem of Boros, Elbassioni,
   Gurvich and Makino~\cite{boros_gurvich_makino}, showing that turn-based stochastic mean-payoff games with
   a fixed number of random positions can be solved in pseudo-polynomial time.
   The original proof relies on a deep analysis of a generalization
   to the stochastic case of the ``pumping algorithm'' of~\cite{gurvich}. Our analysis of value iteration leads to improved complexity estimates,
   see~\Cref{cor-perturb} and~\Cref{finding_top_class_complexity0}.\todo[color=red!20]{SG: added specific references}
Indeed, we bound the characteristic numbers $R$ and $\separ$ in a tight way, by exploiting bit-complexity estimates for the solutions
of Fokker--Planck and Poisson-type equations of discrete Markov chains,
see especially~\Cref{th-bias} showing the existence of a ``short'' bias.\todo[color=red!20]{SG: also added this ref}

   The second application concerns entropy games. Let us recall that in such a game, the value of a pair of (positional)
   strategies of the two players is given by the Perron root of a certain principal
   submatrix of a nonnegative matrix, which we call the {\em ambiguity matrix}, as it measures
   the number of nondeterministic choices of People.  We show that entropy
   games with a fixed rank, and in particular, entropy games with a fixed
   number of People's states, can be solved in pseudo-polynomial time;  see \cref{cor-fp}. These results concern the extended model of entropy games introduced
   in~\cite{entropygamejournal}, taking into account weights. Then, entropy
   games in the sense of~\cite{asarin_entropy}
   (implying a unary encoding of weights)
   that have a fixed rank can be solved in polynomial time. These
   results rely on separation bounds for algebraic numbers arising
   as the eigenvalues of integer matrices with a fixed rank.

   \subsection{Related Work}
   Value iteration is a fundamental technique in game theory, we refer the reader to the survey of Chatterjee and Henzinger for a general perspective~\cite{Chatterjee2008}.
   The idea of applying value iteration to analyze the complexity of deterministic
   mean-payoff games goes back to the classical
   work of  Zwick and Paterson~\cite{zwick_paterson}.
   In some sense, the present approach extends this
   idea to more general classes of games. 
When specialized to stochastic
mean-payoff games with perfect information, our bounds should be compared with
the ones of Boros, Elbassioni, Gurvich, and Makino~\cite{boros_gurvich_makino,boros_gurvich_makino_convex}.
The authors of~\cite{boros_gurvich_makino} generalize the 
``pumping'' algorithm, developed for deterministic games by
Gurvich, Karzanov, and Khachiyan~\cite{gurvich}, to the case of stochastic games. The resulting algorithm is
also pseudopolynomial
if the number of random positions is fixed,
see \cref{rk-compar-gurvich} for a detailed comparison.
The algorithm of Ibsen-Jensen and Miltersen~\cite{ibsen-jensen_miltersen}
yields a stronger bound in the case of simple stochastic games,
still assuming that the number of random positions
is fixed. A different approach, based on an analysis
of strategy iteration, was developed by Gimbert and Horn~\cite{gimbert_horn}
and more recently by Auger, Badin de Montjoye and Strozecki~\cite{auger}.
Eisentraut, Kelmendi, K\v{r}et\'\i nsk\'y and Wininger
developed a value iteration algorithm with guaranteed anytime
upper and lower bounds,
to solve simple stochastic games~\cite{Eisentraut2022}.
The value iteration algorithm for concurrent mean-payoff games, under
an ergodicity condition, has been studied by Chatterjee and Ibsen-Jensen~\cite{chatterjee_ibsen-jensen}.
Theorem~18 there gives a $O(|\log\epsilon|/\epsilon)$ bound
for the number of iterations needed to get an $\epsilon$-approximation
of the mean payoff. When specialized to this case, \cref{pr:approx_constant_value} below
improves this bound by a factor of $|\log \epsilon|$, see \cref{rk-concurrent}.

Another approach to compute the mean payoff is based on {\em relative value iteration} combined with Krasnoselskii--Mann damping~\cite{baillonbruck}, see~\cite{stott2020}. Corollary~14 of \cite{stott2020}
entails that as soon as the ergodic eigenproblem is solvable,
we can get an $\epsilon$-approximation of the mean payoff
as well as an approximate optimality certificate
in $O(1/\epsilon^2)$ iterations.
In contrast, the present \cref{pr:approx_constant_value} does not require
this eigenproblem to be solvable (we only need the mean payoff to be independent
of the initial state), while getting bound of $O(1/\epsilon)$.
Under a restrictive
assumption, requiring every pair of policies of the two players
to yield a unichain transition matrix, this bound
has been recently refined in~\cite{akian2023} to $O(|\log\epsilon|)$.\todo[color=red!30]{SG: reference to \cite{akian2023} added following suggestion of RK} 
See \cref{rk-comparewithkm,rk-concurrent} for a more detailed
comparison with relative Krasnoselskii--Mann iteration.
We build on the operator approach for zero-sum games, see~\cite{bewley_kohlberg,Ney03,RS01}. Our study of entropy games is inspired
by the works of Asarin et al.~\cite{asarin_entropy}
and Akian et al.~\cite{entropygamejournal}.
We rely on the existence of optimal
positional strategies for entropy games, established
in~\cite{entropygamejournal} by an o-minimal geometry approach~\cite{bolte2013},
and also on results of non-linear Perron--Frobenius theory,
especially the Collatz--Wielandt variational formulation
of the escape rate of an order preserving and additively
homogeneous mapping~\cite{nussbaum,gaubert_gunawardena,polyhedra_equiv_mean_payoff,agn}.

The present work, providing complexity bounds based on value iteration,
grew out from an effort to understand the surprising
speed of value iteration on random stochastic games examples arising
from tropical geometry~\cite{issac2016jsc}, by investigating suitable notions of condition numbers~\cite{mtns2018}.
An initial version of some of the present results (concerning turn-based
stochastic games) appeared in the PhD thesis of one of the authors~\cite{skomra_phd}, moroever, an announcement of the present results appeared in the conference paper~\cite{icalp2022}.%
\subsection{Organization of the Paper}
In \cref{sec-preliminaries} we recall the definitions and basic properties of turn-based stochastic mean-payoff games and entropy games, and also key notions in the ``operator approach'' of zero-sum games, including the Collatz--Wielandt optimality certificates.

The universal complexity bounds based on value iteration
are presented in \cref{sec-complexity}. First, we deal
with games whose value is independent of the initial
state, and then, we extend these results to determine
the set of initial states with a maximal value.

The applications to turn-based stochastic mean-payoff
games and to entropy games are provided in \cref{sec:appl_smpg} and \cref{sec-entropy}. %

\section{Preliminaries on Dynamic Programming Operators and Games}\label{sec-preliminaries}

\subsection{Introducing Shapley Operators: The Example of Stochastic Turn-Based Zero-Sum Games}\label{subsec-introducing}
Shapley operators are the two-player version of  Bellman operators (a.k.a.~dynamic programming or one-day operators) which are classically used to study Markov decision processes. In this section we introduce the simplest example of Shapley operator,
arising from stochastic turn-based zero-sum games.

A \emph{stochastic turn-based zero-sum game} is a game played on a digraph $(\vertexset,\edges)$ in which the set of vertices
$\vertexset$ has a non-trivial partition: $\vertexset=\vertexsetMin \dunion \vertexsetMax \dunion \vertexsetNat$. There are two players, called {\em Min} and {\em Max}, and a half-player, {\em Nature}. The sets $\vertexsetMin$, $\vertexsetMax$ and  $\vertexsetNat$ represent the sets of states in which Min, Max, and Nature respectively play. The set of edges $\edges$ represents the allowed moves. We assume $\edges\subset
\vertexsetMin \times \vertexsetMax \cup \vertexsetMax \times \vertexsetNat \cup \vertexsetNat \times \vertexsetMin$, meaning that Min, Max, and Nature
alternate their moves. More precisely, a turn consists of three successive moves: when the current state is $j \in \vertexsetMin$, Min selects
and edge $(j,i)$ in $\edges$ and the next state
is $i \in \vertexsetMax$. Then, Max selects an edge $(i,k)$ in $\edges$
and the next state is $k \in \vertexsetNat$. Next, Nature chooses an edge $(k,j')\in \edges$ and the next state is $j' \in \vertexsetMin$. This
process can be repeated, alternating moves of Min, Max, and Nature.

We make the following assumption.
\begin{assumption}\label{assum-fin}\label{as:entropy_nondeg}
 Each player
 has at least one available action in each state in which he has\todo{MS: Can we change this to ``they have to play'' for more gender neutrality? SG. I am happy with both writings.} to play,
i.e., for all $j \in \vertexsetMin, i\in\vertexsetMax$, and $k\in \vertexsetNat$,
the sets $\{i' \colon (j,i')\in \edges\}$, $\{k' \colon (i,k')\in \edges\}$ and $\{j' \colon (k,j')\in \edges\}$ are non-empty.
\end{assumption}
Furthermore, every state $k \in \vertexsetNat$ controlled by Nature is equipped with a probability distribution on its outgoing edges, i.e., we are given a vector $(P_{kj})_{j \in \vertexsetMin}$ with rational entries such that $P_{kj} \ge 0$ for all $j$ and $\sum_{(k,j) \in \edges} P_{kj} = 1$. We suppose that Nature makes its decisions according to this probability distribution, i.e., it chooses an edge $(k,j)$ with probability $P_{kj}$. Moreover, 
an integer $A_{ij}$ is associated with each edge $(j,i)$ in $\edges\cap(\vertexsetMin \times \vertexsetMax)$, and an integer $B_{ik}$ is associated with each edge $(i,k)$ in $\edges\cap(\vertexsetMax \times \vertexsetNat)$. These integers encode the payoffs of the game in the following way: if the current state of the game is $j \in \vertexsetMin$ and Min selects the edge $(j,i)$, then Min pays to Max the amount $-A_{ij}$. Similarly, if the current state of the game is $i \in \vertexsetMax$ and Max selects the edge $(i,k)$, then Max receives from Min the payment $B_{ik}$.

We first consider the {\em game in horizon $N$}, in which each of the two players Min and Max
makes $N$ moves, starting from a known initial state, which by convention we require to be controlled
by Min. In this setting, a \emph{history} of the game consists of the sequence of states visited
up to a given stage. A (deterministic) \emph{strategy} of a player is a function which assigns to a history of
the game a decision of this player. A pair of strategies $(\sigma,\tau)$ of players
Min and Max induces a probability measure on the set of finite sequences of states. Then, the
expected reward of Max, starting from the initial position $j_0 \in \vertexsetMin$, is defined
by
\(
R_{j_0}(\sigma,\tau) \coloneqq \mathbb{E}_{\sigma \tau} \Bigl(\sum_{p = 0}^{N-1} (-A_{i_p j_p} + B_{i_p k_p}) \Bigr)%
\),
in which the expectation $\mathbb{E}_{\sigma,\tau}$ refers to the probability measure induced by $(\sigma,\tau)$, and
$j_0, i_0, k_0, j_1, i_1, k_1, \dots$ is the random sequence of states visited
when applying this pair of strategies. The objective of Max is to maximize this reward, while Min wants to minimize it. The game in horizon $N$ starting from state $j_0 \in \vertexsetMin$ is known
to have a {\em value} $v^N_{j_0}$ and optimal strategies $\sigma^*$ and $\tau^*$,
meaning that
\(
R_{j_0}(\sigma^*, \tau) \leq v^N_{j_0} \coloneqq R_{j_0}(\sigma^*, \tau^*)  \leq R_{j_0}(\sigma, \tau^*) 
\)
for all strategies $\sigma$ of Min and $\tau$ of Max.
The {\em value vector} $v^N \coloneqq (v^N_j)_{j\in \vertexsetMin}$ keeps track
of the values of all initial states. A classical dynamic programming
argument, see e.g.~\cite[Th.~IV.3.2]{sorin_repeated_games}, shows that
\(
v^0 = \zerovector,  v^N=\shapley(v^{N-1}) \),
where $\zerovector$  denotes the vector that has all entries equal to $0$ and the {\em Shapley operator} $\shapley$ is the map from $\R^\vertexsetMin$ to $\R^\vertexsetMin$ defined by
\begin{align}
\shapley_j(x) \coloneqq \min_{(j,i) \in \edges} \Bigl(-A_{ij}+ \max_{(i,k) \in \edges}\bigl(B_{ik}+ \sum_{(k,l) \in \edges} P_{kl} x_l\bigr)\Bigr), \text{ for all } j\in \vertexsetMin \, .\label{e-elemshapley}
\end{align}
\Cref{assum-fin} guarantees that $\shapley$ is well defined.

One can also consider the stochastic {\em mean-payoff} game with initial state $j_0$, in which the payment $g_{j_0}(\sigma,\tau)$ received by Max becomes the limiting average of the sum of instantaneous payments, i.e.,
\begin{align}
g_{j_0}(\sigma, \tau) \defi \liminf_{N \to +\infty} \mathbb{E}_{\sigma \tau} \Bigl(\frac{1}{N} \sum_{p = 0}^{N-1} (-A_{i_p j_p} + B_{i_p k_p}) \Bigr) \, .\label{e-def-mp}
\end{align}
We say that a strategy is {\em positional} if the decision of the player depends
only of the current state.
A result of Liggett and Lippman~\cite{liggett_lippman} entails that the mean-payoff game with initial state $j_0$ has a value
$\gameval_{j_0}$, and that there exists a pair of
positional strategies $(\sigma^*, \tau^*)$
optimal for all choices of $j_0$, meaning
that 
\(
g_{j_0}(\sigma^*, \tau) \leq \gameval_{j_0} \coloneqq g_{j_0}(\sigma^*, \tau^*) \leq g_{j_0}(\sigma, \tau^*) \),
for every initial state ${j_0} \in \vertexsetMin$ and
for every pair of non-necessarily positional strategies $(\sigma, \tau)$ of players Min and Max. 
A result of Mertens and Neyman~\cite{mertens_neyman} entails in particular that the value of the mean-payoff game
coincides with the limit of the normalized value of the games in horizon $N$, i.e.,
\[
\gameval  = \lim_{N\to \infty} \frac{v^N}{N} = \lim_{N\to \infty}  \frac{\shapley^N(\zerovector)}{N} \enspace,
\]
where $\shapley^N
= \shapley\circ \dots\circ \shapley$ denotes the $N$th iterate of $\shapley$.

\begin{example}\label{ex-smpg}
  An example of stochastic mean-payoff game is shown in \cref{fig-ex-smp}.
  Here $\vertexsetMin = \{\minstate{1},\minstate{2},\minstate{3}\}$,
  and $\vertexsetMax = \{\maxstate{1},\maxstate{2},\maxstate{3}\}$. The
  states in which nature plays are represented by small diamond
  symbols: in each of these states, any of the two outgoing arcs
  is selected with probability $1/2$. The payments are represented on the arcs.
  The operator $F$ is given by\todo{SG: check, check in particular sign convention for payments}
  \begin{align*}
    F_1(x) &= - 1+ \max\big(2+ \frac{x_2+x_3}{2}, -3 + \frac{x_1+x_3}{2}\big),\\
    F_2(x) &= \min\Big( 5 + \max\big(2+ \frac{x_2+x_3}{2}, -3 + \frac{x_1+x_3}{2}\big),
    1+ 3 + \frac{x_1+x_2}{2}\Big)\enspace ,\\
    F_3(x) &= - 4+\frac{x_2+x_3}{2} \enspace .
  \end{align*}
  We shall compute the mean-payoff vector in \cref{ex-smpg-cont}.
  \end{example}
  
  \begin{figure}[htbp]
     \centering
     \begin{tikzpicture}[scale=0.7,>=stealth',max/.style={draw,rectangle,minimum size=0.5cm},min/.style={draw,circle,minimum size=0.5cm},av/.style={draw, diamond, inner sep = 0pt,minimum size = 0.2cm}]

\node[min] (min3) at (-5, 0.8) {$3$};
\node[min] (min1) at (5, 0.4) {$1$};
\node[min] (min2) at (1, 0.1) {$2$};

\node[max] (max1) at (0, -1.2) {$1$};
\node[max] (max2) at (-1, 2) {$2$};
\node[max] (max3) at (-1, 0) {$3$};

\node[av] (av13) at (-2.5,-0.5){};
\node[av] (av23) at (-2.5, 1){};
\node[av] (av12) at (3, 0.4){};

\node[av] (av22) at (0,0.9){};

\draw[->] (max3) to node[below=-0.2ex, font=\small]{$4$} (min2);
\draw[->] (max2) to[out = 170, in = 40] node[above left=-1ex, font=\small]{$-5$} (min3);
\draw[->] (max2) to[out = 25, in = 140] node[above, font=\small]{$-1$} (min1);

\draw[->] (min3) to node[above right=-0.7ex and -1.8ex, font=\small]{$-3$} (av13);
\draw[->] (min3) to node[above right, font=\small]{$2$} (av23);
\draw[->] (min1) to node[below, font=\small]{$3$} (av12);

\draw[->] (max1) to[out = 180, in = -90] node[above, font=\small]{$1$} (min3);

\draw[->] (av13) to (max1);
\draw[->] (av13) to (max3);
\draw[->] (av23) to (max2);
\draw[->] (av23) to (max3);
\draw[->] (av12) to[out=-120, in = 0] (max1);
\draw[->] (av12) to[out=130, in = 0] (max2);
\draw[->] (min2) to node[above right, font=\small]{$0$} (av22);
\draw[->] (av22) to (max2);
\draw[->] (av22) to (max3);

  \end{tikzpicture}
  \caption{A stochastic mean-payoff game.}\label{fig-ex-smp}
  \end{figure}
  
\begin{remark}\label{rk-key}
  In our model, players Min, Max, and Nature play successively,
  so that a turn decomposes in three stages,
  resulting in a Shapley operator of the form~\cref{e-elemshapley}.
  Alternative models, like the one of~\cite{boros_gurvich_makino},
  in which a turn consists of a single move, reduce
  to our model by adding linearly many dummy states, and rescaling the mean payoff by a factor $3$.
  We also refer to~\cite{andersson_miltersen} for more information on equivalent formulations of stochastic mean-payoff games.
\end{remark}

\subsection{The Operator Approach to Zero-Sum Games}
\label{sec:perron_frobenius}

We shall develop a general approach, which applies to various classes of zero-sum games
with a mean-payoff type payment. To do so, it is convenient to introduce an abstract version
of Shapley operators,  following the ``operator approach'' of stochastic games~\cite{RS01,Ney03}.
This will allow us to apply notions from nonlinear Perron--Frobenius theory, especially
sub and super eigenvectors, and Collatz-Wielandt numbers, which play a key
role in our analysis.
We set $[n]:=\{1,\dots,n\}$. We shall use the sup-norm $\| x \|_\infty \defi \max_{i\in \oneto{n}} |x_i|$, and also the {\em Hilbert's seminorm}~\cite{gaubert_gunawardena}, which is defined by $\HNorm{x} \defi \mytop(x) - \mybot(x)$, where $\mytop (x) \defi \max_{i\in \oneto{n}} x_i$ (read ``top'') and $\mybot (x) \defi \min_{i\in \oneto{n}} x_i$ (read ``bottom''). %
 We endow $\R$ with the standard order $\leq$, which is extended to vectors entrywise. 

A self-map $\shapley$ of $\R^n$ is said to be {\em order-preserving} when
\begin{align}
x\leq y \implies \shapley(x) \leq \shapley(y) \; \text{for all} \; x,y\in \R^n \, ,\label{e-order}
\end{align}
and {\em additively homogeneous} when 
\begin{align}
\shapley(\lambda  + x) =\lambda + \shapley(x) \; \text{for all} \; \lambda \in \R \; \text{and} \; x\in \R^n \, ,\label{e-hom}
\end{align}
where, for any $z \in \R^n$, $\lambda + z$ stands for the vector with entries $\lambda + z_i$. 

\begin{definition}\label{def-abstract-shapley}
  A self-map $\shapley$ of $\R^n$ is an \emph{(abstract) Shapley operator}
  if it is order-preserving and additively homogeneous.
\end{definition}
A basic example is provided by the Shapley operator of a
turn-based stochastic mean-payoff game~\cref{e-elemshapley}. Here, the
additive homogeneity
axiom captures the absence of discount. 
We shall see
in the next section a different example, arising from
entropy games.  

We point out that any order-preserving and additively homogeneous self-map $\shapley$ of $\R^n$ is nonexpansive in the sup-norm, meaning that
\[
\|\shapley(x) - \shapley(y)\|_\infty \leq \|x - y\|_\infty \; \text{for all} \; x,y\in \R^n \; .
\]
Using the nonexpansiveness property, we get that the existence and the value of the limit $\lim_{N \rightarrow \infty} (\shapley^N(x)/N)$ are independent of the choice of $x\in \R^n$. We call this limit the {\em escape rate} of $\shapley$, and denote it by $\chi(\shapley)$.
When $\shapley$ is the Shapley operator
of a turn-based stochastic mean-payoff game, fixing $x=\zerovector$, we see that $\shapley^N(x)$ coincides with the value vector in horizon $N$, and so
$\chi_j(\shapley)$ yields
the mean payoff when the initial state is $j$, consistently with our notation $\chi_j$ in \cref{subsec-introducing}. 

The escape rate is known to exist under some ``rigidity'' assumptions.
The case of semialgebraic maps is treated in~\cite{Ney03},
whereas the generalization to o-minimal structures (see~\cite{Dries98} for background),
which is needed in the application to entropy games, is established in~\cite{bolte2013}.
\begin{theorem}[\cite{Ney03} and~\cite{bolte2013}]\label{th:escape_omin}
  Suppose that the function $\shapley \colon \R^n\to \R^n$ is nonexpansive in any norm
  and that it is semialgebraic, or, more generally, defined
  in an o-minimal structure. Then, the escape rate $\chi(\shapley)$
  does exist. 
\end{theorem}
This applies in particular to Shapley operators of turn-based mean-payoff
games, since in this case the operator $\shapley$, given by~\cref{e-elemshapley},
is piecewise affine (meaning that its domain can be covered by finitely many polyhedra such that $F$ restricted to any of them is affine), and a fortiori semialgebraic.
 In the case of entropy games,
we shall see in the next section
that the relevant Shapley operator is defined by a finite expression involving the maps $\log$, $\exp$, as well as the arithmetic operations, and
so that it is definable in a richer structure, which is still o-minimal.
We emphasize that no knowledge of o-minimal techniques is needed to follow
the present paper, it suffices to admit that the escape rate
does exist for all the classes of maps considered here,
and this follows from \cref{th:escape_omin}.

When the map $\shapley$ is piecewise affine, a result finer than \cref{th:escape_omin} holds:

\begin{theorem}[\cite{kohlberg}]\label{Theo:Kohlberg}
A piecewise affine self-map $\shapley$ of $\R^n$ that is nonexpansive in any norm admits an {\em invariant half-line}, meaning that there exist $z, w\in \R^n$ such that 
\[
\shapley(z+\beta w) = z + (\beta +1) w
\]  
for any $\beta \in \R$ large enough. In particular, the escape rate $\chi(\shapley)$ exists, and is given by $w$.
\end{theorem}
This entails that $F^N(z+\beta w) = z+(\beta +N) w$, and so,
by nonexpansiveness of $F$, for all $x\in \R^n$,
$\shapley^N(x)=N\chi(\shapley) + O(1)$ as $N\to \infty$.
This expansion is more precise than \cref{th:escape_omin}, which only states
that $\shapley^N(x)= N\chi(\shapley)+o(N)$.

For a general order-preserving and additively homogeneous self-map of $\R^n$,
and in particular, for the Shapley operators of the entropy games considered
below, an invariant half-line may not exist.
However, we can still recover information about the sequences $(\shapley^N(x)/N)_N$ through non-linear
spectral theory methods. Assuming that $\shapley$ is an order-preserving and additively homogeneous self-map of $\R^n$, the {\em upper Collatz--Wielandt number} of $\shapley$ is defined by:
\begin{equation}\label{DefUCW}
\ucw (\shapley) \defi \inf \{ \mu \in \R \colon \exists z \in \R^n , \shapley(z)\leq \mu + z \} \, ,
\end{equation}
and the {\em lower Collatz--Wielandt number} of $\shapley$ by:
\begin{equation}\label{DefLCW}
\lcw (\shapley) \defi \sup \{ \mu \in \R \colon \exists z \in \R^n , \shapley(z)\geq \mu + z \} \, .
\end{equation}
It follows from Fekete's subadditive lemma that the
limits $\lim_{N\to \infty} \mytop (\shapley^N(\zerovector)/N)$
and $\lim_{N\to \infty} \mybot (\shapley^N(\zerovector)/N)$, which
may be thought of as upper and lower regularizations
of the escape rate, always exist, see~\cite{gaubert_gunawardena}.
In the examples of interest to us, the escape rate
$\chi(\shapley)$ does exist, it represents the mean-payoff vector,
and then $\lim_{N\to \infty} \mytop (\shapley^N(\zerovector)/N)= \mytop (\chi(\shapley))=\max_j \chi_j(\shapley)$
is the maximum of the mean payoff among all the initial states. Similarly,
$\lim_{N\to \infty} \mybot (\shapley^N(\zerovector)/N)= \mybot (\chi(\shapley))$ is the minimum
of these mean payoffs. 

The interest of the vectors $z$ arising in the definition
of Collatz-Wielandt numbers is to provide
{\em approximate optimality certificates},
allowing us to bound mean payoffs from above
and from below.
Indeed, if $F(z) \leq \mu +z$, using the order-preserving property and additively homogeneity
of $F$, we get that $F^N(z) \leq N\mu + z$ for all $N \in \mathbb{N}$, and, by nonexpansiveness
of $F$,
\[ \lim_{N \rightarrow \infty} \mytop(F^N(\zerovector)/N) =\lim_{N \rightarrow \infty} \mytop(F^N(z)/N) \leq \mu\enspace .
\]Similarly,
if $F(z)\geq \mu +z$, we deduce that
\[ \lim_{N \rightarrow \infty} \mybot(F^N(\zerovector)/N) \geq \mu\enspace.
\]
The following
result of~\cite{gaubert_gunawardena}, which can also be obtained as a corollary of a minimax result of Nussbaum~\cite{nussbaum},
see~\cite{polyhedra_equiv_mean_payoff}, shows that these bounds
are optimal. 

\begin{theorem}[{\cite[Prop.~2.1]{gaubert_gunawardena}}, {\cite[Lemma~2.8 and Rk.~2.10]{polyhedra_equiv_mean_payoff}}]\label{Theo:UCW}
  Let $\shapley$ be an order-preserving and additively homogeneous self-map of $\R^n$. Then,
  \[ \lim_{N \rightarrow \infty} \mytop (\shapley^N(x)/N) = \ucw (\shapley)\text{ and } \lim_{N \rightarrow \infty} \mybot (\shapley^N(x)/N) = \lcw (\shapley)
\enspace,\]
for any $x\in \R^n$.
\end{theorem}
Thus, when $\shapley$ is the Shapley operator of a game, the quantities $\ucw(\shapley)$ and $\lcw(\shapley)$ respectively correspond to the greatest and smallest mean payoff among all the initial states.

A simpler situation arises when 
there is a vector $v \in \R^n$ and a scalar $\lambda\in\R$ such that
\begin{align}
\shapley(v) = \lambda + v  \, .
\label{e-ergodic}
\end{align}
The scalar $\lambda$, which is
unique, is known as the {\em ergodic constant}, 
and~\cref{e-ergodic} is referred to as the {\em ergodic equation}.
Then, $\lcw(F)=\ucw(F)=\lambda$. The vector $v$ is known
as a {\em bias} or {\em potential}. It will be convenient
to have a specific notation for the ergodic constant $\lambda$
when the ergodic equation is solvable, then, we set
$\ec(F):=\lambda$.
The existence of a solution $(\lambda,v)$ of~\cref{e-ergodic}
is guaranteed by certain ``ergodicity''
assumptions~\cite{hochartdominion}. 

A tool used to study the ergodic equation~\cref{e-ergodic} is the recession operator of the Shapley operator.  
Recall that if $F$ is a self-map of $\R^n$, the {\em recession operator} of $F$
is defined by
\begin{align}
  \hat{\shapley}(x)\coloneqq \lim_{\beta\to\infty} \beta^{-1}F(\beta x),\qquad x\in \R^n \enspace .
  \label{e-def-recession}
\end{align}
All the Shapley operators of the concrete games considered
in this paper admit recession functions (see Lemmas~\ref{le:recession_smpg} and~\ref{le:recession_entropy} below). Among the properties of the recession operator, we shall make use of the following existence condition. We denote by $\unitvector_n$ (or by $\unitvector$ if the dimension is clear)
the unit vector of $\R^n$.
\begin{theorem}[Coro.\ of Theorems~9 and~13 of~\cite{gaubert_gunawardena}]\label{th-exists-rec}
  Let $\shapley$ be an abstract Shapley operator, such that the recession operator $\hat{\shapley}$
  exists. Suppose that $\hat{\shapley}$ has only trivial fixed points, meaning that $\hat{\shapley}(x)
  =x$ implies $x\in \R\unitvector$. Then, the ergodic equation~\eqref{e-ergodic} is solvable.
\end{theorem}
Besides, we shall repeatedly use the following property, noted by several authors.
\begin{proposition}[see e.g.~\cite{RS01} or Prop.~3.1~of~\cite{ergodicity_conditions}]\label{prop-rec}
  Let $\shapley$ be an abstract Shapley operator. If the recession operator $\hat{\shapley}$ and the escape rate $\chi(\shapley)$ exist, then $\hat{\shapley} (\chi(\shapley)) = \chi(\shapley)$.
\end{proposition}

Another remarkable situation arises when the Shapley operator $\shapley$ is piecewise affine.
Then, it follows form Kohlberg's theorem (\cref{Theo:Kohlberg}) that the ergodic equation~\cref{e-ergodic}
is solvable if and only if the mean payoff is independent of the initial state.
\begin{example}\label{ex-smpg-cont}
  Consider again the stochastic mean-payoff game of \cref{ex-smpg}.
  We can check that $F(v)=\lambda +v$ where $v=(-5.5,0,-10.5)$
  and $\lambda=1.25$, which entails that $1.25$
  is the mean payoff of all initial states.
  \end{example}
Denote $\trop \coloneqq \R \cup \{-\infty\}$. Properties~\cref{e-order} and \cref{e-hom} also make
sense for self-maps of $\trop^n$, by requiring them to
hold for all $x,y\in \trop^n$ and $\lambda\in \trop$. Any order-preserving and additively homogeneous self-map $\shapley$ of $\R^n$
  admits a unique continuous extension $\bar \shapley$ to $\trop^n$,
  obtained by setting, for $x\in \trop^n$,
\begin{align}
\bar \shapley(x) \coloneqq \inf \{\shapley(y) \colon y\in \R^n, y\geq x \} \enspace .\label{e-extended}
\end{align}
Moreover, $\bar \shapley$ is still order-preserving
and additively homogeneous, see~\cite{burbanks_nussbaum_sparrow}
for details. Hence, in the sequel, we assume that any order-preserving and additively homogeneous self-map $\shapley$ of $\R^n$ is canonically extended to $\trop^n$, and we
will not distinguish between $\shapley$ and $\bar{\shapley}$.

\subsection{Entropy Games}
\label{sec-prelim-entropy}
Entropy games were introduced in~\cite{asarin_entropy}. We follow the presentation
of~\cite{entropygamejournal} since it extends the original model, see \cref{rk-compare}
for a comparison.

Similarly to stochastic turn-based zero-sum games, an {\em entropy game} is played on a digraph $(\vertexset,\edges)$ in which the set of vertices $\vertexset$ has a non-trivial partition: $\vertexset=\vertexsetMin \dunion \vertexsetMax \dunion \vertexsetNat$. As in the case of stochastic turn-based games, players Min, Max, and Nature control the states in $\vertexsetMin$, $\vertexsetMax$ and $\vertexsetNat$ respectively, and they alternate their moves, i.e., $\edges\subset
\vertexsetMin \times \vertexsetMax \cup \vertexsetMax \times \vertexsetNat \cup \vertexsetNat \times \vertexsetMin$. We also suppose that the underlying graph satisfies \cref{assum-fin}. In the context of entropy games, player Min is called \emph{Despot}, player Max is called \emph{Tribune}, and Nature is called \emph{People}. For this reason, we denote $\vertexsetD \coloneqq \vertexsetMin$, $\vertexsetT \coloneqq \vertexsetMin$, and $\vertexsetP \coloneqq \vertexsetNat$. The name ``Tribune'' coined in~\cite{asarin_entropy}, refers to the magistrates interceding on behalf of the plebeians in ancient Rome.

The first difference between stochastic turn-based games and entropy games lies in the behavior of Nature: while in stochastic games Nature makes its decisions according to some fixed probability distribution, in entropy games People is a \emph{nondeterministic} player, i.e., nothing is assumed about the behavior of People. The second difference lies in the definition of the payoffs received by Tribune. We suppose that every edge $(p,d)\in \edges$ with $p\in \vertexsetP$ and $d\in \vertexsetD$
is equipped with a {\em multiplicity} $\eweight_{pd}$ which is a (positive)
natural number. The {\em weight} of a path is defined as the product of the multiplicities of the arcs arising on this path. For instance, the path
$(d_0,t_0,p_0,d_1,t_1,p_1,d_2,t_2)$ where
$d_i\in \vertexsetD$, $t_i\in \vertexsetT$ and $p_i\in \vertexsetP$, 
makes $2$ and $1/3$ turn, and its weight is $m_{p_0d_1}m_{p_1d_2}$. A \emph{game in horizon $N$} is then defined as follows: if $(\sigma, \tau)$ is a pair of strategies of Despot and Tribune, then we denote by $R^N_d(\sigma,\tau)$ the sum of the weights of paths with initial state $d$ that make $N$ turns and that are consistent with the choice of $(\sigma,\tau)$. Tribune wants to maximize this quantity, while Despot wants to minimize it. As for stochastic turn-based games, a dynamic programming argument given in~\cite{entropygamejournal} shows that the value $V^N \in \R_{>0}^{\vertexsetD}$ of this game does exist, and that it satisfies the recurrence 
\begin{align}
  V^0= \unitvector, \qquad V^N=\mshapley(V^{N-1}) \, ,
  \label{e-dp-entropygame}
\end{align}
where
the operator $\mshapley \colon \R_{>0}^{\vertexsetD} \to \R_{>0}^{\vertexsetD}$ is defined by
\begin{equation}\label{eq:entropy_operator}
\mshapley_d(x) \coloneqq \min_{(d,t) \in \edges} \max_{(t,p) \in \edges} \sum_{(p,l) \in \edges} \eweight_{pl} x_{l}, \text{ for all } d\in \vertexsetD \, .
\end{equation}
To define a game that lasts for an infinite number of turns, we consider the limit
\[
\mgameval_d(\sigma,\tau)\coloneqq \limsup_{N\to+\infty} (R^N_d(\sigma,\tau))^{1/N} \, ,
\]
which may be thought of as a measure of the freedom of People. The logarithm of this limit is known as a {\em topological entropy}
in symbolic dynamics~\cite{keane}.
The following result shows that the value of the entropy
game $V_d^\infty$ does exist and that it coincides with the limit of the renormalized value $(V_d^N)^{1/N}=[T^N(\unitvector)]_d^{1/N}$ of the finite horizon entropy game, so that the situation is similar to the case of stochastic turn-based games, albeit the renormalization now involves a $N$th geometric mean owing to the multiplicative nature of the payment. 
\todo{two notations ofr the unit vector $\unitvector$ and $e_n$ later. rename $\unitvector_n$ later.}
\begin{theorem}[\cite{entropygamejournal}]\label{th-valexists}
  The entropy game with initial state $d$ has a value $\mgameval_d$.
  Moreover, there are (positional) strategies $\sigma^*$ and $\tau^*$ of Despot and Tribune, such that,
  for all $d\in \vertexsetD$,
 \(
\mgameval_d(\sigma^*,\tau)
\leq \mgameval_d= \mgameval_d(\sigma^*,\tau^*)
\leq \mgameval_d(\sigma,\tau^*),
\)
for all strategies $\sigma$ and $\tau$ of the two players.
In addition, the value vector $\mgameval\coloneqq (\mgameval_d)_{d\in \vertexsetD}$
coincides with the vector
  \( \lim_{N \to \infty}\bigl(\mshapley^{N}(\unitvector)\bigr)^{1/N} \in \R_{>0}^{\vertexsetD}, %
  \)
in which the operation $\cdot ^{1/N}$ is understood entrywise.%
\end{theorem}
Entropy games can be cast in the general operator setting of \cref{sec:perron_frobenius}, by introducing the conjugate operator $\shapley \colon \R^\vertexsetD \to \R^\vertexsetD$,
\begin{align}
  \shapley \coloneqq \log \circ T \circ \exp
  \label{e-def-conj}
\end{align}
in which $\exp \colon \R^\vertexsetD \mapsto \R_{>0}^\vertexsetD$ is the map
which applies the exponential entrywise, and $\log \coloneqq \exp^{-1}$.
Since the maps $\log$ and $\exp$ are order-preserving, and since
the weights $m_{pl}$ appearing in the expression of $T(x)$ in~\cref{eq:entropy_operator}
are nonnegative, the operator $F$ is order-preserving. Moreover, using the morphism
property of the maps $\log$ and $\exp$ with respect to multiplication and addition,
we see that $F$ is also additively homogeneous, hence, it is an abstract
Shapley operator in the sense of \cref{def-abstract-shapley}.
Moreover, it is definable in the real exponential field,
which was shown to be an o-minimal structure by Wilkie~\cite{wilkie},
and this is precisely how \cref{th-valexists}
is derived in~\cite{entropygamejournal} from \cref{th:escape_omin}.
Actually, entropy games are studied in~\cite{entropygamejournal} in a more general setting, allowing history dependent strategies and showing that positional
strategies are optimal. It is also shown there that the game
has a uniform value in the sense of Mertens and Neyman~\cite{mertens_neyman}.

When the (positional) strategies $\sigma,\tau$ are fixed, the
value can be characterized by a classical result of Perron--Frobenius
theory.
\begin{definition}\label{def-ambiguity}
Given a pair of strategies
$(\sigma,\tau)$ of Despot and Tribune,
we define the {\em ambiguity matrix} $M^{\sigma,\tau}\in \R_{\geq 0}^{\vertexsetD \times \vertexsetD}$,
with entries $(M^{\sigma,\tau})_{k,l}= \eweight_{\tau(\sigma(k)),l}$ if $\bigl(\tau(\sigma(k)),l\bigr) \in \edges$ and $(M^{\sigma,\tau})_{k,l}=0$ otherwise.
\end{definition}
In other words, $M^{\sigma,\tau}$ is the ``weighted transition matrix''
of the subgraph $\mathscr{G}^{\sigma,\tau}$ obtained by keeping only
the arcs $\vertexsetD \to \vertexsetT$ and $\vertexsetT\to \vertexsetP$ determined by the two strategies.\todo[color=blue!30]{MS: Not exactly, since $\mathscr{G}^{\sigma,\tau}$ contains states that are not in $\vertexsetD$ -- this is a transition matrix of a ``reduced'' graph in which we only have states in $\vertexsetD$. But maybe this is not important. SG. Yes, this is an abuse of terminology. I moved the sentence outside the definition
and put quotation marks.} 
The digraph $\mathscr{G}^{\sigma,\tau}$ can generally be decomposed in strongly
connected components $\mathscr{C}_1, \dots,\mathscr{C}_s$, and each of these
components, $\mathscr{C}_i$,
determines a principal submatrix of $M^{\sigma,\tau}$, denoted
by $M^{\sigma,\tau}[\mathscr{C}_i]$, obtained by keeping only the
rows and columns in $\mathscr{C}_i\cap \vertexsetD$.
We denote by $\specrad(\cdot)$ the spectral radius of a matrix,
which is also known as the {\em Perron root} when the matrix
is nonnegative and irreducible, see~\cite{berman} for background.
\begin{proposition}[\cite{rothblumwhittle}, {\cite[Th.~5.1]{zijmjota}}]\label{prop-rothblum}
  The value of the subgame with initial state $d$, induced by a pair
  of strategies $\sigma,\tau$, coincides with
  \[
  \max \{\specrad(M^{\sigma,\tau}[\mathscr{C}_i] ) \colon \text{there is a dipath }
  d\to \mathscr{C}_i \text{ in }\mathscr{G}^{\sigma,\tau} \} \enspace .
  \]
\end{proposition}

 \begin{figure}
   \centering
     \begin{tikzpicture}[scale=0.7,>=stealth',max/.style={draw,rectangle,minimum size=0.5cm},min/.style={draw,circle,minimum size=0.5cm},av/.style={draw, diamond, inner sep = 0pt,minimum size = 0.2cm}]

\node[min] (min3) at (-5, 0.8) {$3$};
\node[min] (min1) at (5, 0.4) {$1$};
\node[min] (min2) at (1, 0.1) {$2$};

\node[max] (max1) at (0, -1.2) {$1$};
\node[max] (max2) at (-1, 2) {$2$};
\node[max] (max3) at (-1, 0) {$3$};

\node[av] (av13) at (-2.5,-0.5){};
\node[av] (av23) at (-2.5, 1){};
\node[av] (av12) at (3, 0.4){};

\node[av] (av22) at (0,0.9){};

\draw[->] (max3) to node[below=-0.2ex, font=\small]{} (min2);
\draw[->] (max2) to[out = 170, in = 40] node[above left=-1ex, font=\small]{} (min3);
\draw[->] (max2) to[out = 25, in = 140] node[above, font=\small]{} (min1);

\draw[->] (min3) to node[above right=-0.7ex and -1.8ex, font=\small]{} (av13);
\draw[->] (min3) to node[above right, font=\small]{} (av23);
\draw[->] (min1) to node[below, font=\small]{} (av12);

\draw[->] (max1) to[out = 180, in = -90] node[above, font=\small]{} (min3);

\draw[->] (av13) to node[above right=-0.7ex,font=\small]{$2$} (max1);
\draw[->] (av13) to node[above right=-0.3ex,font=\small]{$1$} (max3);
\draw[->] (av23) to node[below right=-1ex,font=\small]{$3$} (max2);
\draw[->] (av23) to node[above right=-0.9ex,font=\small]{$1$} (max3);
\draw[->] (av12) to[out=-120, in = 0] node[below right=0.0ex,font=\small]{$2$} (max1);
\draw[->] (av12) to[out=130, in = 0] node[above right=-0.7ex,font=\small]{$5$} (max2);
\draw[->] (min2) to (av22);
\draw[->] (av22) to node[above right=0.0ex,font=\small]{$4$} (max2);
\draw[->] (av22) to node[above left=-0.9ex,font=\small]{$2$} (max3);

  \end{tikzpicture}
  \caption{An entropy game. Despot's states are represented by squares. Tribune's states are represented by circles. People's state are represented by small diamonds. The multiplicies are indicated on the arcs.}
  \label{fig-entropy-new}
 \end{figure}

 \todo{SG: \Cref{fig-entropy-new} and \Cref{ex-entropy} added}
 \begin{example}\label{ex-entropy}
   An example of entropy game is given in \cref{fig-entropy-new}. The graph is similar
   to the one of the stochastic game of \cref{fig-ex-smp}, but now, the Min states (squares) are
   interpreted as Despot's states, the Max states (circles) are interpreted
   as Tribune's states, and the former ``nature'' states (small diamonds)
   become People's states. The multiplicities are shown on the graph.
   The associated operator $T$ is given by
   \begin{align*}
   T_1(x)&=\max(3x_2+x_3, 2x_1+x_3)\enspace,\\
   T_2(x) &=\min(\max(3x_2+x_3, 2x_1+x_3),2x_1+5x_2)\enspace,\\
   T_3(x)&=4x_2+2x_3 \enspace.
   \end{align*}
   Suppose that Despot plays according to the strategy $\sigma$,
   which makes the move $\minstate{2}\to\maxstate{3}$
   (the other moves  $\minstate{1}\to \maxstate{3}$, 
   $\minstate{3}\to\maxstate{2}$ are trivial), and that Tribune plays
   according to the strategy $\tau$, which makes a move from $\maxstate{3}$ to the upper diamond state
(and the other choices are trivial). Then, the
   associated ambiguity matrix is
   \[
   M^{\sigma,\tau} = \left(\begin{array}{ccc}  0 & 3 & 1\\ 0 & 3 & 1 \\ 0 & 4 & 2
     \end{array}\right) \enspace  .
   \]
   By \cref{prop-rothblum}, under this pair of strategies, and for all initial states,
   the value of the induced subgame is $\mu = (5+\sqrt{17})/2\approx 4.56$, the Perron root
   of $M^{\sigma,\tau} $. It can be checked that $T(V)=\mu V$, where $V=(a,a,1)$
   and $a=1/(\mu-3)\approx 0.64$, which entails that $\mu$ is indeed the value of the original game.
     \end{example}

\begin{remark}\label{rk-compare}
  In the original model of Asarin et al.~\cite{asarin_entropy},
  an entropy game is specified by finite sets of states of Depot and Tribune,
  $D$ and $T$, respectively, by a finite alphabet $\Sigma$ representing
  actions, and by a transition relation $\Delta \subset T \times \Sigma \times D \cup   D \times \Sigma \times T$. A turn consists
  of four successive moves by Despot, People, Tribune,
  and People: in state $d\in D$, Despot selects an action $a\in \Sigma$,
  then, People moves to a state $t\in P$ such that $(d,a,t)\in\Delta$.
  Then, Tribune selects an action $b\in\Sigma$, and People moves
  to a state $d'\in D$ such that $(t,b,d')\in \Delta$.
  This reduces to the model of~\cite{entropygamejournal} by introducing
  dummy states, identifying a turn in the game of~\cite{asarin_entropy}
  to a succession of two turns in the game of~\cite{entropygamejournal}.
  Another difference is that the payment,
  in~\cite{asarin_entropy}, corresponds to $\max_{d\in D}\limsup_{N\to\infty}(R_d^N)^{1/N}$,
  and this is equivalent to letting Tribune choose the initial state
  before playing the game. Then, the value of the game in~\cite{asarin_entropy} coincides with the maximum of the values of the initial states, $\max_d \mgameval_d$, see \cite[Prop.~11]{entropygamejournal}.
  Finally in~\cite{asarin_entropy}, the arcs have multiplicity one,
  whereas we allow integer multiplicies (coded in binary), as in~\cite{entropygamejournal}.
\end{remark}

\section{Bounding the Complexity of Value Iteration}%
\label{sec-complexity}
In this section, $\shapley$ is an (abstract) Shapley operator, i.e.,
an order-preserving and additively homogeneous self-map of $\R^n$.

\subsection{A Universal Complexity Bound for Value Iteration}\label{sec:basic_val_iter}
The most straightforward idea to solve
a mean-payoff game is probably value iteration:
we infer whether or not the mean-payoff
game is winning for Player Max (meaning that the value is nonnegative) by solving the finite horizon game,
for a large enough horizon. This is formalized in \cref{algo-vi}.

\begin{figure}[t]
\begin{footnotesize}
\begin{algorithmic}[1]
\Procedure {ValueIteration}{$\shapley$}
\\ \Lcomment{$\shapley$ a Shapley operator from $\R^n$ to $\R^n$ }
\State $u \coloneqq \zerovector \in \R^n$ %
\Repeat 
\label{step-while}
\label{state-iter} $u \coloneqq \shapley(u)$  
\Lcomment{At iteration $\ell$, $u=\shapley^\ell(\zerovector)$ is the value vector of the game in finite horizon $\ell$}
\Until{$\mytop (u)\leq 0$ or $\mybot (u)\geq 0$}
\If{$\mytop (u)\leq 0$}\label{step-termmax}
\Return ``$ \ucw(\shapley)\leq 0$'' \Lcomment{Player Min wins for all initial states}
\Else \label{step-termmin}
\ \Return ``$\lcw(\shapley)\geq 0$'' \Lcomment{Player Max wins for all initial states}
\EndIf
\EndProcedure
\end{algorithmic}
\end{footnotesize}
\caption{Basic value iteration algorithm.} \label{algo-vi}
\end{figure}

When the non-linear eigenproblem $\shapley(w)=\ec(\shapley) + w$
is solvable, we shall use the following metric estimate, which
represents the minimal Hilbert's seminorm of a bias vector
\[
R(\shapley) \defi \inf \left\{ \HNorm{w} \colon w \in \R^n , \; \shapley(w)=\ec(\shapley) + w \right\} \, .  
\]
In general, however, this non-linear eigenproblem may not be solvable.
Then, we consider, for $\lambda\in\R$, the sets
\[
S_\lambda(F) \defi \{v\in \R^n\colon \lambda +v \leq F(v)\}
\; \makebox{ and }\;
S^\lambda(F) \defi \{v\in \R^n\colon \lambda +v \geq F(v)\} \enspace .
\]

\begin{theorem}\label{Th:NitsBound}
  Procedure \textsc{ValueIteration} (\cref{algo-vi})
  is correct as soon as $\lcw(F)>0$ or $\ucw(F)<0$,
  and it terminates in a number of iterations $N_{\textrm{vi}}$
  bounded
  by
  \begin{equation}\label{NitsBound}
   \inf\bigg\{ \frac{\HNorm{v}}{\lambda}\colon \lambda>0 \, , \; v\in S_\lambda(F) \cup S^{-\lambda}(F)
\bigg\} \enspace . 
  \end{equation}
  In particular, if $F$ has a bias vector
  and $\ec(F)\neq 0$,
we have
$ 
N_{\textrm{vi}} \leq \frac{R(F)}{|\ec(F)|}$.
\end{theorem}
To prove \cref{Th:NitsBound} and some other results, we need the following lemma.

\begin{lemma}[{\cite[Theorem~8]{gaubert_gunawardena}}]\label{top_and_bot}
  We have $\mybot(\shapley^{\ell}(\zerovector)) \le {\ell}\lcw(\shapley) \le {\ell}\ucw(\shapley) \le \mytop(\shapley^{\ell}(\zerovector))$
for ${\ell} \in \N$.
\end{lemma}

\begin{proof}[Proof of \cref{Th:NitsBound}.]
Since $v+ \mybot (w-v)\leq w \leq v + \mytop (w-v)$ for any $v,w \in \R^n$, and $F^\ell$ is order-preserving and additively homogeneous for any $\ell \in \N$, note that 
\[F^\ell(v)- \mytop (v) = F^\ell(v) + \mybot (-v) \leq   F^\ell(\zerovector) 
\leq F^\ell(v) + \mytop (-v) = F^\ell(v)- \mybot (v) \; ,
\]
for all $v \in \R^n$ and $\ell \in \N$.

To prove the theorem, in the first place suppose that $\lcw(F)>0$. Then, by the definition~\cref{DefLCW} of $\lcw(F)$ we know that there exist $\lambda > 0$ and $v\in \R^n$ such that $v\in S_\lambda(F)$. Moreover, observe that in this case for any $\lambda > 0$ and 
$v\in S_\lambda(F) \cup S^{-\lambda}(F)$ we necessarily have $v\in S_\lambda(F)$, because otherwise we would have $v \in S^{-\lambda}(F)$, which implies that $\ucw(F)<0$ and so that $\lcw(F)<0$, contradicting our assumption. Therefore, for any $\lambda > 0$ and $v\in \R^n$ such that $v\in S_\lambda(F) \cup S^{-\lambda}(F)$, we have $\lambda + v \leq F(v)$ and so
\[
\ell \lambda - \HNorm{v} =\ell \lambda + \mybot (v) - \mytop (v) = \mybot(\ell \lambda + v - \mytop (v))  \leq \mybot(F^\ell(v)- \mytop (v))\leq \mybot( F^\ell(\zerovector)) 
\]
for all $\ell \in \N$. Thus, we conclude that $0 \leq \mybot( F^\ell(\zerovector))$ if $\ell \in \N$ is  greater than~\cref{NitsBound}.
This shows that Procedure \textsc{ValueIteration} terminates in a number of iterations $N_{\textrm{vi}}$ bounded by~\cref{NitsBound} when $\lcw(F)>0$. The proof in the case $\ucw(F)<0$ is analogous.

Finally, observe that the correctness of Procedure \textsc{ValueIteration} readily follows from \cref{top_and_bot}. 
\end{proof}

A special case of \cref{Th:NitsBound} in which the existence of a bias vector is assumed appeared (without proof)
in~\cite{mtns2018}.
\begin{remark}
 The infimum in~\cref{NitsBound} is generally not attained. Consider for instance
  $\shapley: \R^2\to \R^2$ given by $F(x)=(\log(\exp(x_1)+\exp(x_2)),x_2)-\alpha$, where $\alpha >0$.
  Then, since $\shapley_2(x)=x_2-\alpha<x_2$, we have $S_{\lambda}(F)= \emptyset$ for $\lambda>0$. Besides, since $x-\lambda \geq \shapley(x)$ if and only if $x_1 -\lambda \geq \log(\exp(x_1)+\exp(x_2)) -\alpha$ and $x_2-\lambda \geq x_2-\alpha$, it follows that $S^{-\lambda}(F)\neq \emptyset$ if and only if $\lambda < \alpha$. Now let $v \in S^{-\lambda}(F)$ for some $\lambda < \alpha$. Without loss of generality, we may assume $\mybot (v) =0$. Then, we have $v_1 -\lambda \geq \log(\exp(v_1)+\exp(v_2)) - \alpha \geq \log 2 - \alpha$ and so $\frac{\HNorm{v}}{\lambda} \geq 1 + \frac{\log 2 - \alpha}{\lambda}$. We conclude that the infimum in~\cref{NitsBound} is equal to $\frac{\log 2}{\alpha}$ but it is not attained.  
\end{remark}
\subsection{Value Iteration in Finite Precision Arithmetic}\label{sec:finite_prec_val_iter}
\begin{figure}[t]
\begin{footnotesize}
  \begin{algorithmic}[1]
\Procedure{FPValueIteration}{$\tilde{\shapley}$}
\State $u \coloneqq \zerovector \in \R^n$, $\ell\coloneqq 0\in \N$, $\epsilon\in \R_{>0}$ 
\Repeat
\label{state-iternew} $u \coloneqq \tilde{\shapley}(u)$; $\ell\coloneqq \ell+1$  
\Lcomment{We suppose that the operator $\shapley$ is evaluated in approximate arithmetic, so that $\tilde{\shapley}(u)$ is at most at distance $\epsilon$ in the sup-norm from its true value $\shapley(u)$.}
\Until{$\ell\epsilon+\mytop (u)\leq 0$ or $-\ell\epsilon+\mybot (u)\geq 0$}

\If{$\ell\epsilon+ \mytop (u)\leq  0$}\label{step-termmax2new}
\Return  ``$ \ucw(F)\leq 0$'' \Lcomment{Player Min wins for all initial states}
\EndIf
\If{$-\ell\epsilon+ \mybot (u)\geq  0$}\label{step-termmax3new}
\Return  ``$ \lcw(F)\geq 0$'' \Lcomment{Player Max wins for all initial states}
\EndIf
\EndProcedure
\end{algorithmic}
\end{footnotesize}
\caption{Value iteration in finite precision arithmetic.} %
\label{algo-vi2}
\end{figure}

The algorithm in \cref{algo-vi} can be adapted to work
in finite precision arithmetic. Consider the variant
of the main body of this algorithm, given in \cref{algo-vi2}. 
Now we assume that each evaluation of the Shapley operator $F$
is performed with an error of at most $\epsilon>0$ in the sup-norm. In what follows, we denote by $\apshapley \colon \R^{n} \to \R^{n}$ the operator which approximates $F$, as in Procedure \textsc{FPValueIteration}, so it satisfies:
\begin{equation}\label{DefApproxOpMain}
\supnorm{\apshapley(x) - \shapley(x)} \le \epsilon \enspace \makebox{for all } x \in \R^n.
\end{equation}

The following result is established by exploiting nonexpansiveness properties of Shapley operators. 

\begin{theorem}\label{Th:NitsBound:approx}
  Procedure \textsc{FPValueIteration} (\cref{algo-vi2})
  is correct as soon as $\lcw(F)>2\epsilon$ or $\ucw(F)<-2\epsilon$,
  and it terminates in a number of iterations $N^\epsilon_{\textrm{vi}}$
  bounded
  by
  \begin{equation}\label{FPNitsBound}
    \inf \bigg\{ \frac{\HNorm{v}}{\lambda-2\epsilon}\colon \lambda>2\epsilon \, , \; v\in S_\lambda(F) \cup S^{-\lambda}(F)
\bigg\} \enspace .
  \end{equation}
  In particular, if $F$ has a bias vector and $|\ec(F)|>2\epsilon$,
we have
$ 
N^\epsilon_{\textrm{vi}} \leq \frac{R(F)}{|\ec(F)|-2\epsilon}$.

\end{theorem}

The proof relies on the next lemma.

\begin{lemma}\label{le:approx_value_iter}
Denote $u^{\ell} \coloneqq \shapley^{\ell}(\zerovector)$ and $\tilde{u}^{\ell} \coloneqq \apshapley^{\ell}(\zerovector)$ for $\ell \in \N$. Then, we have $\supnorm{\tilde{u}^{\ell} - u^{\ell}} \le \ell \epsilon$, $\abs{\mytop(\tilde{u}^{\ell}) - \mytop(u^{\ell})} \le \ell \epsilon$, and $\abs{\mybot(\tilde{u}^{\ell}) - \mybot(u^{\ell})} \le \ell \epsilon$ for any $\ell \in \N$.
\end{lemma}
\begin{proof}
We prove the first claim by induction on $\ell$. We have $\supnorm{\tilde{u}^{1} - u^{1}} \le \epsilon$ by~\cref{DefApproxOpMain}. Furthermore, since $\shapley$ is nonexpansive we get
\begin{align*}
\supnorm{\tilde{u}^{\ell} - u^{\ell}} &= \supnorm{\apshapley(\tilde{u}^{\ell - 1}) - \shapley(u^{\ell - 1})} \\
&\le \supnorm{\apshapley(\tilde{u}^{\ell - 1}) - \shapley(\tilde{u}^{\ell - 1})} + \supnorm{\shapley(\tilde{u}^{\ell - 1}) - \shapley(u^{\ell - 1})} \\
&\le \epsilon + \supnorm{\tilde{u}^{\ell - 1} - u^{\ell - 1}} \le \ell \epsilon \, .
\end{align*}
To prove the other claims, fix $\ell \in \N$ and let $k \in [n]$ be such that $\mytop(u^{\ell}) = u^{\ell}_{k}$. Then, we have $\mytop(u^{\ell}) = u^{\ell}_{k} \le \tilde{u}^{\ell}_k + \ell \epsilon \le \mytop(\tilde{u}^{\ell}) + \ell \epsilon$. Similarly, if $k' \in [n]$ is such that $\mytop(\tilde{u}^{\ell}) = \tilde{u}^{\ell}_{k'}$, then $\mytop(\tilde{u}^{\ell}) = \tilde{u}^{\ell}_{k'} \le u^{\ell}_{k'} + \ell \epsilon \le \mytop(u^{\ell}) + \ell \epsilon$. Thus, we get $\abs{\mytop(\tilde{u}^{\ell}) - \mytop(u^{\ell})} \le \ell \epsilon$. The proof of the fact that $\abs{\mybot(\tilde{u}^{\ell}) - \mybot(u^{\ell})} \le \ell \epsilon$ is analogous. 
\end{proof}

\begin{proof}[Proof of \cref{Th:NitsBound:approx}.]
To prove that Procedure \textsc{FPValueIteration} returns the correct answer, suppose that it stops at the $\ell$th iteration and, in the first place, that the condition $-\ell\epsilon+\mybot (u) =-\ell\epsilon+ \mybot (\apshapley^\ell(\zerovector)) \geq 0$ is satisfied. Then, we have $\mybot( F^{\ell}(\zerovector)) \geq 0$ by \cref{le:approx_value_iter}, and so we conclude that $\lcw(F) \geq 0$ by \cref{top_and_bot}. Therefore, player Max wins for all initial states if the condition $-\ell\epsilon+ \mybot (\apshapley^\ell(\zerovector)) \geq 0$ is satisfied. If this condition is not satisfied, we necessarily have $\ell\epsilon+ \mytop (\apshapley^\ell(\zerovector)) \leq 0$ due to the stopping condition. Then, we can prove that $\ucw(F) \leq 0$ using  symmetric arguments to the ones applied in the case in which the condition $-\ell\epsilon+ \mybot (\apshapley^\ell(\zerovector)) \geq 0$ is satisfied, and so player Min wins for all initial states. 

In order to bound the number of iterations required by \textsc{FPValueIteration}, suppose that $\lcw(\shapley) > 2\epsilon$. Then, by definition there exist $\lambda > 2\epsilon$ and $v \in S_\lambda(\shapley)$. Moreover, if $\lambda$ and $v$ are such that $\lambda > 2\epsilon$ and $v \in S_\lambda(\shapley) \cup S^{-\lambda}(\shapley)$, then we have $v \in S_\lambda(\shapley)$, as $v \in S^{-\lambda}(\shapley)$ would imply that $\ucw(\shapley) < -2\epsilon$. Therefore $\lambda + v \le \shapley(v)$ and, by \cref{le:approx_value_iter},
\[
\ell \lambda - \HNorm{v} = \mybot(\ell \lambda + v - \mytop (v))  \leq \mybot(\shapley^\ell(v)- \mytop (v))\leq \mybot( \shapley^\ell(\zerovector)) \le \mybot(\apshapley^{\ell}(\zerovector)) + \ell \varepsilon \, .
\]
Hence, if $\ell \ge \frac{\HNorm{v}}{\lambda - 2\epsilon}$, we have $-\ell \epsilon +\mybot(\apshapley^{\ell}(\zerovector)) \ge 0$. Therefore, \textsc{FPValueIteration} terminates in a number of iterations bounded by \cref{FPNitsBound}. The proof in the case where $\ucw(F)<-2\epsilon$ is analogous. 
\end{proof}

\begin{figure}[t]
\begin{footnotesize}
  \begin{algorithmic}[1]
    \Procedure{ApproximateConstantMeanPayoff}{$\shapley$}
    \State $u,x,y \coloneqq \zerovector \in \R^n$, $\ell\coloneqq 0\in \N$,
    $\delta \in \R_{>0}$
\Lcomment{The number $\delta$ is the desired precision of approximation.}
\Repeat
\label{state-iternew2} $u \coloneqq \tilde{\shapley}(u)$; $\ell\coloneqq \ell+1$  
\Lcomment{The operator $\shapley$ is evaluated in approximate arithmetic, so that $\tilde{\shapley}(u)$ is at most at distance $\epsilon \coloneqq \delta/8$ in the sup-norm from its true value $\shapley(u)$.}
\Until{$\mytop (u)- \mybot (u)\leq (3/4)\delta \ell$}\\
$\kappa \coloneqq \mybot (u)/\ell$; $\lambda \coloneqq \mytop (u)/\ell$\\
$u \coloneqq \zerovector$
\For{$i = 1,2,\dots,\ell-1$} 
$u \coloneqq \tilde{\shapley}(u)$; $x \coloneqq \max\{x,-i \kappa + u \}$; $y \coloneqq \min\{y,-i \lambda + u \}$
\EndFor\\
\Return ``$[\lcw(F),\ucw(F)]$ is included in the interval $[\kappa - \delta/8, \lambda + \delta/8]$, which is of width at most $\delta$. Furthermore, we have $\kappa - \delta/8 + x \le \shapley(x)$ and $\lambda + \delta/8 + y \ge \shapley(y)$.''
\Lcomment{All initial states have a value in $[\kappa - \delta/8, \lambda + \delta/8]$.}
\EndProcedure
\end{algorithmic}
\end{footnotesize}
\caption{Approximating the value of a mean-payoff game when it is independent of the initial state, and computing approximate optimality certificates, working in finite precision arithmetic.} %
\label{algo-vi3}
\end{figure}

When \todo[color=red!30]{RK: In light of the modifications made in \cref{pr:approx_constant_value}, I slightly modified this paragraph (mainly the first part).} %
Procedure~\textsc{ApproximateConstantMeanPayoff} of \cref{algo-vi3} halts, it returns sub and super-eigenvectors $x$ and $y$ satisfying $\kappa - \delta/8 + x \le \shapley(x)$ and $\lambda + \delta/8 + y \ge \shapley(y)$, 
for some $\kappa, \lambda \in \R$ such that the interval $[\kappa - \delta/8, \lambda + \delta/8]$ is of width at most some desired precision $\delta$.  As discussed above, by \cref{Theo:UCW}
this entails that the mean payoff of every initial state
is included in the interval $[\kappa - \delta/8, \lambda + \delta/8]$. The construction of these sub and sup-eigenvectors, by taking infima and suprema of normalized orbits of $F$, is inspired by~\cite[Proof of Lemma~2]{gaubert_gunawardena}.

\begin{lemma}[cf.\ {\cite[Lemma~2]{gaubert_gunawardena}}]\label{certify_bias}
Suppose that $\mybot (\apshapley^{\ell}(\zerovector)) \ge \pert \ell$ for some $\pert \in \R$ and $\ell \in \N$. If we define 
\[
\hat{\bias} \coloneqq \zerovector \lor \Bigl(- \pert + \apshapley(\zerovector) \Bigr) \lor \dots \lor \Bigl(  - (\ell-1)\pert + \apshapley^{\ell-1}(\zerovector) \Bigr) \, ,
\]
then $\pert - \epsilon + \hat{\bias} \le \shapley(\hat{\bias})$. Analogously, if $\mytop (\apshapley^{\ell}(\zerovector)) \le \pert \ell$ and we define
\[
\overbar{\bias} \coloneqq \zerovector \land \Bigl(- \pert + \apshapley(\zerovector) \Bigr) \land \dots \land \Bigl(  - (\ell-1)\pert + \apshapley^{\ell-1}(\zerovector) \Bigr) \, ,
\]
then $\pert + \epsilon + \overbar{\bias} \ge \shapley(\overbar{\bias})$.
\end{lemma}

\begin{proof}
Since $\shapley$ is order-preserving, we have $\shapley(x \lor y) \ge \shapley(x) \lor \shapley(y)$ for every $x,y \in \R^n$. Therefore, we get
\begin{align*}
&\shapley(\hat{\bias}) \ge \shapley(\zerovector) \lor \Bigl(- \pert + \shapley\bigl(\apshapley(\zerovector)\bigr) \Bigr) \lor \dots \lor \Bigl(  - (\ell-1)\pert + \shapley\bigl(\apshapley^{\ell-1}(\zerovector)\bigr) \Bigr) \\
&\ge -\epsilon + \apshapley(\zerovector) \lor \Bigl(- \pert - \epsilon + \apshapley^2(\zerovector) \Bigr) \lor \dots \lor \Bigl(  - (\ell-1)\pert - \epsilon + \apshapley^{\ell}(\zerovector) \Bigr) \\
&\ge -\epsilon + \apshapley(\zerovector) \lor \Bigl(- \pert - \epsilon + \apshapley^2(\zerovector) \Bigr) \lor \dots \lor \Bigl(  - (\ell-2)\pert - \epsilon + \apshapley^{\ell-1}(\zerovector) \Bigr) \lor (\pert - \epsilon ) \\
&= \pert - \epsilon + \hat{\bias} \, .
\end{align*}
The proof of the fact that $\pert + \epsilon + \overbar{\bias} \ge \shapley(\overbar{\bias})$ if $\mytop (\apshapley^{\ell}(\zerovector)) \le \pert \ell$ is analogous. 
\end{proof}

\begin{theorem}\label{pr:approx_constant_value}
When it halts, Procedure~\textsc{ApproximateConstantMeanPayoff} (\cref{algo-vi3}) returns an interval  containing
$[\lcw(F),\ucw(F)]$ of width at most the desired precision of approximation $\delta$. Moreover, the procedure does halt when $\ucw(F)=\lcw(F)$,
and the number of iterations of the first loop is bounded
by  $\lceil 8R/\delta \rceil$ where 
 $R \coloneqq \max\{\HNorm{v}, \HNorm{w}\}$ for any vectors $v,w \in \R^{n}$ that satisfy $\rho - \delta/8 + v \le \shapley(v)$ and $\rho + \delta/8 + w \ge \shapley(w)$, denoting $\rho$ the common value of $\ucw(F)$ and $\lcw(F)$.
\end{theorem}

\begin{proof}
Let $\epsilon \coloneqq \delta/8$ as in \cref{algo-vi3} and $\tilde{u}^{\ell} \coloneqq \tilde{\shapley}^{\ell}(\zerovector)$ for $\ell \in \N$ as before.
 
Assume first that Procedure~\textsc{ApproximateConstantMeanPayoff}  halts, \ie, that $\mytop (\tilde{u}^{\ell})- \mybot (\tilde{u}^{\ell})\leq (3/4)\delta \ell$ for some $\ell$. Combining \cref{top_and_bot,le:approx_value_iter} we get 
\[
 \mybot(\tilde{u}^\ell) - \ell \epsilon \le \mybot\bigl(\shapley^{\ell}(\zerovector)\bigr) \le {\ell}\lcw(\shapley) \le {\ell}\ucw(\shapley) \le \mytop\bigl(\shapley^{\ell}(\zerovector)\bigr) \le \mytop(\tilde{u}^\ell) + \ell \epsilon \; .
\]
It follows that the interval $[\mybot (\tilde{u}^\ell)/\ell - \delta/8, \mytop (\tilde{u}^\ell)/\ell + \delta/8]$ returned by the procedure contains $[\lcw(F),\ucw(F)]$ and that it has width at most $\delta$ since $\mytop (\tilde{u}^{\ell})- \mybot (\tilde{u}^{\ell})\leq (3/4)\delta \ell$.

Now, assume $\ucw(F)=\lcw(F)$. Let $\rho \coloneqq\ucw(F)$ and $R \coloneqq \max\{\HNorm{v}, \HNorm{w}\}$, where $v,w \in \R^{n}$ are any two vectors that satisfy $\rho - \epsilon + v \le \shapley(v)$ and $\rho + \epsilon + w \ge \shapley(w)$ (note that such vectors exist due to the fact that $\ucw(F)=\lcw(F)=\rho$ and $\epsilon > 0$). %
Since $v \le \mytop(v) + \zerovector$, for every $\ell \in \N$ we get $\ell \rho - \ell \epsilon + v \le \shapley^{\ell}(v) \le \mytop(v) + \shapley^{\ell}(\zerovector)$. Thus, $\shapley^{\ell}(\zerovector) \ge \ell \rho - \ell \epsilon + v - \mytop(v) \ge \ell \rho - \ell \epsilon - \HNorm{v}$. Analogously, we get $\shapley^{\ell}(\zerovector) \le \ell \rho + \ell \epsilon + \HNorm{w}$. Then,  \cref{le:approx_value_iter} %
shows that $\mytop(\tilde{u}^\ell) \le \ell \rho + \ell \epsilon + R + \ell \epsilon = \ell \rho + \ell \delta/4 + R$, and analogously that $\mybot(\tilde{u}^\ell) \ge \ell \rho - \ell \delta/4 - R$. Hence $\mytop(\tilde{u}^\ell)  - \mybot(\tilde{u}^\ell) \le \ell \delta/2 +2R$. In particular, for every $\ell \ge 8R/\delta$ we have $\mytop (\tilde{u}^\ell) - \mybot (\tilde{u}^\ell) \le (3/4)\delta \ell$ and so the stopping condition of the first loop is achieved within the first $\lceil 8R/\delta \rceil$ iterations. 

Finally, observe that if $x$ and $y$ are the vectors defined in the procedure, then \cref{certify_bias} implies $\kappa - \delta/8 + x = \kappa - \epsilon + x \le \shapley(x)$ and $\lambda + \delta/8 + y \ge \shapley(y)$. 
\end{proof}

\begin{remark}\label{rk-comparewithkm}
  An alternative approach to compute the ergodic constant $\rho$ of the game with Shapley operator $F$ is to apply relative value iteration with Krasno\-selskii--Mann damping, see~\cite{stott2020}, which consists in computing the sequences
  $y^k=F(x^{k-1})-\mytop(F(x^{k-1}))\mathbf{1}$, and $x^k = \alpha y^k + (1-\alpha) x^{k-1}$, for any $0<\alpha<1$. 
We have $\mybot(F(x^k)-x^k)
  \leq \rho \leq\mytop(F(x^k)-x^k)$, and it is shown
  there,
  as a consequence of a general result on Krasnoselskii--Mann iteration
  in Banach spaces~\cite{baillonbruck},  that as soon as $F$ admits a bias vector,
  $\|x^k - F(x^{k})\|_H \leq O(1/\sqrt{k})$,
  see~\cite[Coro.~13]{stott2020}.
In this way,
  we obtain an approximation of $\rho$ with a precision
  $\epsilon$ in $O(1/\epsilon^2)$ iterations. In contrast, \cref{pr:approx_constant_value} provides an approximation of the same quality
  in only $O(1/\epsilon)$ iterations. However, note that relative
  value iteration with Krasnoselskii--Mann damping produces a vector $x:=x^k$ such that $
  \rho  + x -O(1/\sqrt{k}) \leq F(x)\leq \rho  + x +O(1/\sqrt{k})$
  when it stops after $k$ iterations, whereas  Procedure~\textsc{ApproximateConstantMeanPayoff} returns {\em two} vectors $x$ and $y$ such that
  $\rho+ x -O(1/k) \leq F(x)$
  and $F(y)\leq \rho+y +O(1/k)$.
  In particular, this procedure returns a {\em pair} $(x,y)$ of approximate optimality certificates.
  Hence, by ``relaxing'' the constraint that ``$x=y$''
  in the optimality certificates,   
  we passed from an iteration complexity of $O(1/\epsilon^2)$
  to $O(1/\epsilon)$ to get an $\epsilon$-approximation.
  Note also that Procedure~\textsc{ApproximateConstantMeanPayoff}
  differs from Krasnoselskii--Mann damping in the fact that it replaces
  a linear averaging of $x^{k-1}$ and $y^k$ by a non-linear ``averaging''
  operation,
  taking a supremum or infimum of a normalized orbit, as in \cref{certify_bias}.
  Alternatively, we could rely on the recent result~\cite{Contreras2022},
  showing that a Halpern-type iteration reaches
  a $O(1/\epsilon)$ iteration complexity bound,
  this would also lead to an $O(1/\epsilon)$ bound in our
  application, but only for games for which a bias vector does exist,
  whereas \Cref{pr:approx_constant_value} only requires $\ucw(F)=\lcw(F)$.
\todo[color=red!30]{SG: added remark/ref to the result of contreras and cominetti}
  \end{remark}

\subsection{Deciding Whether the Value Is Independent of the Initial State}
\label{sec:abstract_cst}

In this section, we will show how the value iteration algorithm can be adapted to decide whether or not a given game has constant value.  %
Our analysis is based on an abstract notion of dominion. 

As previously, we suppose that $\shapley \colon \R^{n} \to \R^{n}$ is an order-preserving and additively homogeneous operator.
Recall that thanks to~\cref{e-extended}, $\shapley$ is canonically extended
to define a self-map of $\trop^n$. 
Furthermore, given a nonempty set $\set \subset [n]$, we define the operator $\shapley^{\set} \colon \trop^{\set} \to \trop^\set$ as $\shapley^{\set} \coloneqq \proj^{\set} \circ {\shapley} \circ \revproj^{\set}$, where $\proj^{\set} \colon \trop^n \to \trop^{\set}$ is the projection on the coordinates in $\set$ which is defined as usual by $\proj^{\set}_j(x) = x_j$ for $j \in \set$, and $\revproj^{\set} \colon \trop^{\set} \to \trop^n$ is defined by $\revproj^{\set}_j(x) = x_j$ if $j \in \set$ and $\revproj^{\set}_j(x) = \zero$ otherwise. 
The next two lemmas provide elementary properties of the operators $\shapley^{\set}$.
\begin{lemma}
  If $\set \subset [n]$ is any set, then the operator $\shapley^{\set}$ is continuous, order-preserving, and additively homogeneous. 
\end{lemma}     
\begin{proof}
By definition, $\shapley^{\set}$ is a composition of continuous, order-preserving, and additively homogeneous maps. 
\end{proof}

\begin{lemma}\label{le:smaller_ucw}
  If $\setI$ and $\setII$ are two nonempty subsets of $[n]$ such that $\setI \subset \setII$, then $\bigl(\shapley^{\setI}\bigr)^{\ell}_j(\zerovector) \le \bigl(\shapley^{\setII}\bigr)^{\ell}_j(\zerovector)$ for all $\ell \in \N$ and $j \in \setI$.
\end{lemma}
\begin{proof}
Observe that if $x \in \trop^{\setI}$ and $y \in \trop^{\setII}$ are such that $x_j \le y_j$ for all $j \in \setI$, then $\shapley^{\setI}_j(x) \le \shapley^{\setII}_j(y)$ for all $j \in \setI$. Indeed, the facts that $x_j \le y_j$ for all $j \in \setI$ and that $\setI \subset \setII$ imply $\revproj^{\setI}(x) \le \revproj^{\setII}(y)$, and so for any $j \in \setI$ we have $\shapley^{\setI}_j(x) = {\shapley}_j(\revproj^{\setI}(x)) \le {\shapley}_j(\revproj^{\setII}(y)) = \shapley^{\setII}_j(y)$ because ${\shapley}$ is order-preserving. 
 Hence, by setting $(x,y) \coloneqq (\zerovector,\zerovector)$ we get the first claim for $\ell = 1$. Then, this claim follows by induction setting $x \coloneqq (\shapley^{\setI})^{\ell - 1}(\zerovector)$ and $y \coloneqq (\shapley^{\setII})^{\ell - 1}(\zerovector)$. 
\end{proof}

\begin{definition}
A \emph{dominion (of Player Max)} is a nonempty set $\dominion \subset [n]$ such that $\shapley^{\dominion}$ preserves $\R^{\dominion}$, i.e., such that  $\shapley^{\dominion}(x) \in \R^{\dominion}$ for all $x \in \R^{\dominion}$.
\end{definition}
As discussed in \cite{issac2016jsc,hochartdominion}, for stochastic mean-payoff games (with finite action spaces), a dominion of a player can be interpreted as a set of states such that the player can force the game to stay in this set if the initial state belongs to it. This terminology differs from the one of~\cite{jurdzinski_subexponential}, in which a dominion is required in addition to consist only of initial states that are winning for this player. 

\begin{lemma}\label{domin_at_zero}
A set $\dominion \subset [n]$ is a dominion if and only if $\shapley^\dominion(\zerovector) \in \R^{\dominion}$.
\end{lemma}
\begin{proof}
If $\dominion \subset [n]$ is a dominion, then $\shapley^\dominion(\zerovector) \in \R^{\dominion}$. Conversely, given any $x \in \R^{\dominion}$, we have $\shapley^{\dominion}(x) \ge \mybot(x) + \shapley^{\dominion}(\zerovector)$ because $x \ge \mybot(x) + \zerovector$ and $\shapley^{\dominion}$ is order-preserving and additively homogeneous. Thus, if $\shapley^\dominion(\zerovector) \in \R^{\dominion}$, we conclude that $\shapley^{\dominion}(x) \in \R^{\dominion}$ for all $x \in \R^{\dominion}$. 
\end{proof}

The procedures that we discuss in this section require an additional assumption on the structure of the Shapley operator $\shapley$.

\begin{assumption}\label{as:good_operators}
We assume that the limit $\gameval^{\dominion} \coloneqq \lim_{\ell \to \infty}\frac{(\shapley^{\dominion})^{\ell}(\zerovector)}{\ell} \in \R^\dominion$ exists for every dominion $\dominion \subset [n]$. Furthermore, we assume that the set $\domMaxVal \coloneqq \{j \in [n] \colon \gameval^{[n]}_j = \ucw(\shapley)\}$ is a dominion and that it satisfies $\lcw(\shapley^{\domMaxVal}) = \ucw(\shapley^{\domMaxVal}) = \ucw(\shapley)$.
\end{assumption}

\begin{remark}\label{rem:o-minimal-dominion}
We note that the first part of \cref{as:good_operators} holds automatically when the Shapley operator $\shapley \colon \R^n\to \R^n$ is definable in an o-minimal structure. Indeed, in this case the relation~\cref{e-extended} implies that $\shapley^{\dominion}$ is definable in the same structure for every dominion $\dominion$, so $\gameval^{\dominion}$ exists by \cref{th:escape_omin}. We will see
that the second part of the assumption applies to the games considered in this paper.
\end{remark}
  \Cref{as:good_operators} will allow us to make an induction on the number states, by a reduction
  to a simpler game with a reduced state space $\dominion$. In particular,
  the assumption that the limit
  $\gameval^{\dominion}= \lim_{\ell \to \infty}\frac{(\shapley^{\dominion})^{\ell}(\zerovector)}{\ell}$
  exists will allow us to apply value iteration to the Shapley operator of the reduced
  game, $\shapley^{\dominion}$.

\begin{lemma}\label{le:smaller_val_domin}
If $\dominionI$ and $\dominionII$ are two  dominions such that $\dominionI \subset \dominionII$, then  $\gameval^{\dominionI}_j \le \gameval^{\dominionII}_{j}$ for all $j \in \dominionI$.
\end{lemma}
\begin{proof}
Since $\gameval^{\dominionI} = \lim_{\ell \to \infty}\frac{(\shapley^{\dominionI})^{\ell}(\zerovector)}{\ell}$ and $\gameval^{\dominionII} = \lim_{\ell \to \infty}\frac{(\shapley^{\dominionII})^{\ell}(\zerovector)}{\ell}$, the lemma readily follows from \cref{le:smaller_ucw}. 
\end{proof}

\begin{lemma}\label{le:dominion_containing_topclass}
  Suppose that $\dominion \subset [n]$ is a dominion that contains the set
  of states of maximal value, i.e., $  \domMaxVal= \{j \in [n] \colon \gameval_j = \ucw(\shapley)\} \subset \dominion$. Then, $\ucw(\shapley) = \ucw(\shapley^\dominion)$ and $\domMaxVal
  = \{j \in \dominion \colon \gameval_j^{\dominion} = \ucw(\shapley^\dominion)\}$.
\end{lemma}
\begin{proof}
  We have $\gameval_j^{\dominion} \le \gameval_j$ for all $j \in \dominion$ by \cref{le:smaller_val_domin}. %
The set $\domMaxVal$ is a dominion of the operator $\shapley$ by \cref{as:good_operators}. Moreover, since $\domMaxVal \subset \dominion$, we have $\gameval^{\domMaxVal}_j \le \gameval_j^{\dominion} \le \gameval_j = \ucw(\shapley)$ for all $j \in \domMaxVal$ by \cref{le:smaller_val_domin}. As $\gameval^{\domMaxVal}_j = \ucw(\shapley)$ for all $j \in \domMaxVal$ by \cref{as:good_operators}, it follows that $\gameval_j^{\dominion} = \ucw(\shapley)$ for all $j \in \domMaxVal$. Furthermore, if $j \notin \domMaxVal$, then $\gameval_j^{\dominion} \le \gameval_j < \ucw(\shapley)$, which finishes the proof. 
\end{proof}

From now on, we denote $\gameval \coloneqq \gameval^{[n]}$.
The following theorem applies to Shapley operators for which an a priori
separation bound is known: if $\ucw(F)>\lcw(F)$, 
it requires an apriori bound $\delta>0$ such that
$\ucw(F)-\lcw(F)>\delta$.
We note that the existence
of the approximate sub and super-eigenvectors
$v$ and $w$ used in this theorem
follows from \cref{as:good_operators}.%
\begin{figure}[t]
\begin{footnotesize}
  \begin{algorithmic}[1]
    \Procedure{DecideConstantValue}{${\shapley},\delta,R$}
    \State $u \coloneqq \zerovector \in \R^n$, $\ell\coloneqq 0\in \N$.

    \State $\tilde{\shapley}\coloneqq$ any map such that $\tilde{\shapley}(u)$ is at most at distance $\epsilon \coloneqq \delta/8$ in the sup-norm from $\shapley(u)$.
\Repeat
\label{state-iternew3} $u \coloneqq \tilde{\shapley}(u)$; $\ell\coloneqq \ell+1$  
\Until{$\mytop (u)- \mybot (u) \le (3/4)\delta \ell$ or $\ell = 1 + \lceil 8R/\delta\rceil$}
\If{$\ell = 1 + \lceil 8R/\delta\rceil$}
\State
$\setAlgo \coloneqq \{i\colon u_i = \mybot(u)\}$
\State\Return{$\setAlgo$} \Lcomment{The value of the game depends on the initial state. We have $\gameval_{i} < \ucw(\shapley)$ for all $i\in \setAlgo$.}

\Else
\State\Return{$\varnothing$}
\Lcomment{The value of the game is independent of the initial state}.
\EndIf
\EndProcedure
\end{algorithmic}
\end{footnotesize}
\caption{Algorithm that decides if the value is constant.} %
\label{algo-vi4}
\end{figure}

\begin{theorem}\label{Th-algo-const-value}
  Suppose that $\shapley$ is such that either
  $\ucw(\shapley)= \lcw(\shapley)$ or
  $\ucw(\shapley) - \lcw(\shapley) > \delta$
  for some $\delta>0$.
Recall that $\domMaxVal$ denotes the set of states of maximal value and let $R \coloneqq \max\{\HNorm{v}, \HNorm{w}\}$, where $v,w \in \R^{\domMaxVal}$ are any two vectors that satisfy $\ucw(\shapley) - \delta/8 + v \le \shapley^{\domMaxVal}(v)$ and $\ucw(\shapley) + \delta/8 + w \ge \shapley^{\domMaxVal}(w)$. Then,  Procedure~\textsc{DecideConstantValue} (\cref{algo-vi4}) is correct.
\end{theorem}

The proof relies on the following lemma. 

\begin{lemma}\label{le:value_iter_maxstates}
Let $\apshapley \colon \R^{n} \to \R^{n}$ be such that $\supnorm{\apshapley(x) - \shapley(x)} \le \epsilon$ for all $x \in \R^n$. %
Suppose that $\set \subset [n]$, $v \in \R^{\set}$ and $\gamma > 0$ are such that $\ucw(\shapley) - \gamma + v \le \shapley^{\set}(v)$. Then, for all $\ell \in \N$ and $j \in \set$, we have $\apshapley^{\ell}_{j}(\zerovector) \ge \ell(\ucw(\shapley) - \gamma - \epsilon) - \HNorm{v}$.
\end{lemma}
\begin{proof}
In the first place, note that $\ucw(\shapley) - \gamma + \revproj^{\set}(v) \le {\shapley}(\revproj^{\set}(v))$. Indeed, if $j \in \set$ we have $\ucw(\shapley) - \gamma + \revproj^{\set}_j(v) \le {\shapley}_j(\revproj^{\set}(v))$ because $\revproj^{\set}_j(v) = v_j$ and ${\shapley}_j(\revproj^{\set}(v)) = \shapley^{\set}_j(v)$, and if $j \notin \set$ we also have $\ucw(\shapley) - \gamma + \revproj^{\set}_j(v) \le {\shapley}_j(\revproj^{\set}(v))$ because $\revproj^{\set}_j(v) = \zero$. 

Now, since $\zerovector \ge -\mytop(v) + \revproj^{\set}(v)$ (note that $\mytop(v) \neq \zero$ because $v \in \R^{\set}$), we get ${\shapley}^{\ell}
(\zerovector) \ge -\mytop(v) + {\shapley}^{\ell}(\revproj^{\set}(v)) \ge -\mytop(v) + \ell (\ucw(\shapley) - \gamma) + \revproj^{\set}(v)$. Hence, for every $j \in \set$ we have ${\shapley}^{\ell}_j(\zerovector) \ge -\mytop(v) + \ell (\ucw(\shapley) - \gamma) + \revproj^{\set}_j(v) = -\mytop(v) + \ell (\ucw(\shapley) - \gamma) + v_j \ge \ell (\ucw(\shapley) - \gamma) - \HNorm{v}$. Therefore, by \cref{le:approx_value_iter}, we conclude that $\apshapley^{\ell}_{j}(\zerovector) \ge \ell (\ucw(\shapley) - \gamma - \epsilon) - \HNorm{v}$. 
\end{proof}

\begin{proof}[Proof of \cref{Th-algo-const-value}.]
Denote $\tilde{u}^{\ell} \coloneqq \tilde{\shapley}^{\ell}(\zerovector)$ for $\ell \in \N$. As in the proof of \cref{pr:approx_constant_value}, we note that \cref{top_and_bot,le:approx_value_iter} imply that we have the inequality
\begin{equation}\label{eq:approx_bot_nonconst}
\mybot(\tilde{u}^{\ell}) - \ell (\delta/8)  \le \ell \lcw(\shapley) \le \ell \ucw(\shapley) \le \mytop(\tilde{u}^{\ell}) + \ell(\delta/8)  \, .
\end{equation}
Hence, if the condition $\mytop(\tilde{u}^{\ell}) - \mybot(\tilde{u}^{\ell}) \le (3/4)\delta \ell$ is satisfied for some $\ell$, then $\ucw(\shapley) - \lcw(\shapley) \le \delta$ and so $\ucw(\shapley) = \lcw(\shapley)$ by our assumption on $\delta$. Conversely, if $\ucw(\shapley) = \lcw(\shapley)$, then $\domMaxVal = [n]$ and \cref{pr:approx_constant_value} shows that the condition $\mytop(\tilde{u}^{\ell}) - \mybot(\tilde{u}^{\ell}) \le (3/4)\delta \ell$ is satisfied for some $\ell \le \lceil 8R/\delta \rceil$. Hence, the algorithm correctly decides if the value of the game is constant. 

To finish the proof, suppose the value of the game is not constant, so $\ucw(\shapley) - \lcw(\shapley) > \delta$ by our assumption on $\delta$, and let $\ell \coloneqq \lceil 8R/\delta \rceil + 1$. Then, by \cref{le:value_iter_maxstates} for any $j \in \domMaxVal$ we have $\tilde{u}^{\ell}_{j} \ge \ell(\ucw(\shapley) - \delta/8 - \epsilon) - \HNorm{v} \ge \ell(\ucw(\shapley) - \delta/8 - \epsilon) - R > \ell\ucw(\shapley) - 3\ell(\delta/8) $ since $R < \ell (\delta/8)$ and $\epsilon = \delta/8$ in Procedure~\textsc{DecideConstantValue}. On the other hand, \cref{eq:approx_bot_nonconst} implies $\mybot(\tilde{u}^{\ell}) - \ell (\delta/8) \le \ell \lcw(\shapley) < \ell(\ucw(\shapley) - \delta)$. Hence, we have $\mybot(\tilde{u}^{\ell}) < \ell\ucw(\shapley) - 7 \ell (\delta/8)  < \tilde{u}^{\ell}_{j}$. We conclude that, if we take any $k \in [n]$ such that $\tilde{u}_{k} = \mybot(\tilde{u}^{\ell})$, then $k \notin \domMaxVal$. 
\end{proof}

\subsection{Finding the States of Maximal Value}\label{sec:abstract_topclass}

We now refine \cref{Th-algo-const-value}, showing in \cref{topclass_oracle} below that we can in fact extract
all the initial states with maximal value. To this end, we introduce and analyze two procedures, \Extend{} and \TopClass{}, which are used to obtain \cref{topclass_oracle}.

\begin{figure}[t]
\begin{footnotesize}
\begin{algorithmic}[1]
\Procedure {\Extend}{$\dominionAlgo,\setAlgo$}
\\ \Lcomment{$\setAlgo \subset \dominionAlgo$ is a nonempty set}
\While{\texttt{True}}
\State $\setpAlgo \coloneqq \emptyset$
\State define $x \in \trop^{\dominionAlgo}$ as $x_j \coloneqq -\infty$ for $j \in \setAlgo$ and $x_j = 0$ otherwise
\For{$j \in \dominionAlgo$}
\If{${\shapley}^{\dominionAlgo}_j(x)  = -\infty$ }
\State $\setpAlgo \coloneqq \setpAlgo \cup \{j\}$
\EndIf
\EndFor
\If{$\setAlgo \cup \setpAlgo \neq \setAlgo$}
\State $\setAlgo \coloneqq \setAlgo \cup \setpAlgo$
\Else
\State \Return $\setAlgo$
\EndIf
\EndWhile
\EndProcedure
\end{algorithmic}
\end{footnotesize}
\caption{Procedure that extends the set laying outside some dominion.} \label{fig:algorithm_extend}
\end{figure}

\begin{lemma}\label{extension_of_outsiders}
Procedure \Extend{} (\cref{fig:algorithm_extend}) has the following properties. Let $\dominion\subset [n]$ be a dominion and $\set \subset \dominion$ be a non-empty set. If $\dominionp$ is a dominion such that $\dominionp \subset \dominion$ and $\set \cap \dominionp = \emptyset$, then $\Extend(\dominion,\set) \cap \dominionp = \emptyset$. Besides, $\dominion \setminus \Extend(\dominion,\set)$ is a dominion.
\end{lemma}
\begin{proof}%
At each execution of the while loop, the procedure either stops or strictly increases the cardinality of $\set \subset \dominion$. Therefore, the procedure stops after at most $\card{\dominion}-1$ executions of the while loop. 

To prove the first property of the procedure, let $\dominionp \subset \dominion$ be a dominion such that $\set \cap \dominionp = \emptyset$. %
Since $\set \cap \dominionp = \emptyset$ and $\dominionp \subset \dominion$, we have $\revproj^{\dominionp}(\zerovector) \le \revproj^{\dominion}(x)$, where $x$ is the vector defined in the while loop of the procedure (i.e., $x \coloneqq \revproj^{\dominion \setminus \set}(\zerovector)$). Hence, ${\shapley}\bigl(\revproj^{\dominionp}(\zerovector)\bigr) \le {\shapley}\bigl(\revproj^{\dominion}(x)\bigr)$ because ${\shapley}$ is order-preserving. Since $\dominionp$ is a dominion, we get $-\infty < \shapley^{\dominionp}_j(\zerovector) = {\shapley}_j(\revproj^{\dominionp}(\zerovector)) \le {\shapley}_j\bigl(\revproj^{\dominion}(x)\bigr) = \shapley^{\dominion}_j(x)$ for all $j \in \dominionp$. Therefore, $\setp \cap \dominionp = \emptyset$ and so $(\set \cup \setp) \cap \dominionp = \emptyset$. In other words, the set $\set \cup \setp$ obtained by a single execution of the while loop in the procedure is disjoint from $\dominionp$. Therefore, we have $\Extend(\dominion,\set) \cap \dominionp = \emptyset$. 

To prove the second property, let $\dominionp \coloneqq \dominion \setminus  \Extend(\dominion,\set)$, and let $x \in \trop^n$ be the vector defined as $x_j = -\infty$ for $j \in \Extend(\dominion,\set)$ and $x_j = 0$ otherwise (i.e., let $x \coloneqq \revproj^{\dominionp}(\zerovector)$). The stopping criterion of the procedure implies that $\{j \in \dominion \colon \shapley_j^\dominion\bigl(\proj^\dominion(x)\bigr) = \zero\} \subset \Extend(\dominion,\set)$. Since $x = \revproj^\dominion\bigl(\proj^\dominion(x)\bigr)$, we have $\shapley_j^\dominion\bigl(\proj^\dominion(x)\bigr) = \shapley_j(x)$ for every $j \in \dominion$. Therefore, $\{j \in \dominion \colon \shapley_j(x) = \zero\} \subset \Extend(\dominion,\set)$ and so $\shapley_j(x) > \zero$ for all $j \in \dominionp$. In particular, we have $\shapley^{\dominionp}_j(\zerovector) = {\shapley}_j(\revproj^{\dominionp}(\zerovector)) = {\shapley}_j(x) > -\infty$ for all $j \in \dominionp$. Hence, $\shapley^{\dominionp}(\zerovector) \in \R^{\dominionp}$ and $\dominionp$ is a dominion by \cref{domin_at_zero}. 
\end{proof}

Let
\[
\operatorname{sep}(F)\coloneqq \inf_{\dominion} \big(\ucw(F^\dominion)-\lcw(F^\dominion)\big)
\]
where the infimum is taken over all the dominions $\dominion$ of $F$ which contain all the states of maximal value and satisfy $\ucw(F^\dominion)-\lcw(F^\dominion)>0$.

 \begin{figure}[t]
\begin{footnotesize}
\begin{algorithmic}[1]
\Procedure {\TopClass}{$\shapley,\delta,R$}
\State $\dominionAlgo \coloneqq[n]$
\While{\texttt{True}}
\State $\setAlgo
\coloneqq \Call{DecideConstantValue}{\shapley^{\dominionAlgo}, \delta, R}$ %
\If{$\setAlgo = \varnothing$}
\State \Return $\dominionAlgo$
\Lcomment{$\dominionAlgo$ is the set of states that have the maximal value}
\EndIf
\State $\dominionAlgo \coloneqq \dominionAlgo \setminus \Extend(\dominionAlgo,\setAlgo)$
\EndWhile
\EndProcedure
\end{algorithmic}
\end{footnotesize}
\caption{Procedure that finds the set of states with maximal value.} \label{fig:top_class_finder-abstract}
\end{figure}
  
\begin{theorem}\label{th:topclass}
Let $\delta > 0$ be such that $\delta<\operatorname{sep}(F)$, $\domMaxVal$ be the set of states of maximal value and $R \coloneqq \max\{\HNorm{v}, \HNorm{w}\}$, where $v,w \in \R^{\domMaxVal}$ are any two vectors that satisfy $\ucw(\shapley) - \delta/8 + v \le \shapley^{\domMaxVal}(v)$ and $\ucw(\shapley) + \delta/8 + w \ge \shapley^{\domMaxVal}(w)$.
  Then, Procedure \textsc{TopClass}($\shapley,\delta,R$) (\cref{fig:top_class_finder-abstract}) halts after at most $n$
  iterations of the while loop, and correctly computes
  the set of initial states with maximal value.
  \end{theorem}

 \begin{proof}%
  Fix any dominion $\dominion \subset [n]$ such that $\domMaxVal \subset \dominion$ and consider two cases. If $\dominion = \domMaxVal$, then we have $\ucw(\shapley) = \ucw(\shapley^\dominion) = \lcw(\shapley^\dominion)$ by \cref{as:good_operators}. By applying \cref{Th-algo-const-value} to $\shapley^\dominion$, Procedure $\textsc{DecideConstantValue}(\shapley^{\dominion},\delta,R)$ outputs $\varnothing$. If $\dominion \neq \domMaxVal$, then we have $\ucw(\shapley) = \ucw(\shapley^\dominion) \neq \lcw(\shapley^\dominion)$ by \cref{le:dominion_containing_topclass}. Thus, $\ucw(\shapley^\dominion) - \lcw(\shapley^\dominion) \ge \separ(\shapley) > \delta$ by the definition of $\separ(\shapley)$. Hence, by applying \cref{Th-algo-const-value} to $\shapley^\dominion$, Procedure $\textsc{DecideConstantValue}(\shapley^{\dominion},\delta,R)$ outputs some nonempty set $\set \subset \dominion$ such that $\gameval^\dominion_i < \ucw(\shapley)$ for all $i \in \set$. In particular, $\set \cap \domMaxVal = \emptyset$ by \cref{le:dominion_containing_topclass}. Therefore, by \cref{extension_of_outsiders}, the set $\dominion \setminus \textsc{Extend}(\dominion,\set)$ is a dominion that contains $\domMaxVal$ and is strictly smaller than $\dominion$. By induction, if $\dominion_\ell$ denotes the set $\dominion$ at the $\ell$th iteration of the while loop in $\TopClass(\shapley,\delta,R)$, then we have $[n] = \dominion_1 \supsetneq \dominion_2 \supsetneq \dots \supset \domMaxVal$ and all of the sets $\dominion_\ell$ are dominions. Furthermore, Procedure $\TopClass(\shapley,\delta,R)$ stops only when it finds a set $\dominion_p$ such that $\dominion_p = \domMaxVal$ and outputs $\dominion_p$ as the set of states with maximal value. 
 \end{proof}

To state the final result of this section, we will suppose that we have access to an oracle that approximates (the canonical extension~\cref{e-extended} of) $\shapley$ to a given precision $\epsilon > 0$. More precisely, given a point $x \in \trop^n$, the oracle is supposed to output a point $y \in \trop^n$ that satisfies $y_j = \zero \iff \shapley_j(x) = \zero$ and $|\shapley_j(x) - y_j| \le \epsilon$ for all $j$ such that $\shapley_j(x) \neq \zero$. 

\begin{theorem}\label{topclass_oracle}
  Let $\delta > 0$ be such that $\delta<\operatorname{sep}(F)$, $\domMaxVal$ be the set of states of maximal value and $R \coloneqq \max\{\HNorm{v}, \HNorm{w}\}$, where $v,w \in \R^{\domMaxVal}$ are any two vectors that satisfy $\ucw(\shapley) - \delta/8 + v \le \shapley^{\domMaxVal}(v)$ and $\ucw(\shapley) + \delta/8 + w \ge \shapley^{\domMaxVal}(w)$.
Then, the set of initial states of maximal value can be found by making at most $n^2 + n\lceil 8R/\delta\rceil$ calls to an oracle that approximates $\shapley$ to precision $\epsilon \coloneqq \delta/8$.%
\end{theorem}

\begin{proof}%
Denote the oracle by $\apshapley \colon \trop^n \to \trop^n$ and let $\dominion_1, \dots, \dominion_p$ be as in the proof of \cref{th:topclass}. If $\dominion \coloneqq \dominion_{\ell}$ for some $\ell \in [p]$, then $|\shapley_j^\dominion(x) - \apshapley_j\bigl(\revproj^\dominion(x)\bigr)| \le \delta/8$ for every $x \in \R^\dominion$ and $j \in \dominion$. Indeed, since $\dominion$ is a dominion, we have $\shapley_j^\dominion(x) = \shapley_j\bigl(\revproj^\dominion(x)\bigr) > \zero$ for all $j \in \dominion$ and therefore $|\shapley_j^\dominion(x) - \apshapley_j\bigl(\revproj^\dominion(x)\bigr)| \le \delta/8$ by the definition of $\apshapley$. Hence, the map that approximates $\shapley^\dominion \colon \R^\dominion \to \R^\dominion$ as required by Procedure $\textsc{DecideConstantValue}(\shapley^{\dominion},\delta,R)$ is obtained by calling $\apshapley$. Furthermore, Procedure $\textsc{DecideConstantValue}(\shapley^{\dominion},\delta,R)$ finishes after making at most $1 + \lceil 8R/\delta \rceil$ calls to the oracle. Moreover, for every $x \in \trop^\dominion$ and $j \in \dominion$ we have $\shapley^\dominion_j(x) = \zero \iff \apshapley_j\bigl(\revproj^\dominion(x)\bigr) = \zero$ by the definition of $\apshapley$. Therefore, we can replace every call to $\shapley^\dominion(x)$ in $\textsc{Extend}(\dominion,\set)$ by a call to $\apshapley\bigl(\revproj^\dominion(x)\bigr)$. We also note that $\textsc{Extend}(\dominion,\set)$ makes at most $n-1$ calls to the oracle. Therefore, a single iteration of the while loop in $\TopClass(\shapley,\delta,R)$ can be done by making at most $n + \lceil 8R/\delta\rceil$ calls to the oracle. The claim follows from the fact that $\TopClass$ finishes after at most $n$ iterations of the while loop. 
\end{proof}

\section{Application to Stochastic Mean-Payoff Games}\label{sec:appl_smpg}

In this section, we apply our results to stochastic mean-payoff games. We follow the notation concerning stochastic mean-payoff games as introduced in \cref{subsec-introducing}. In particular, we suppose that $\shapley$ is of the form given in \cref{e-elemshapley}. Before starting, we note that the continuous extension of $\shapley$ to $\trop^\vertexsetMin$ is still given by the same formula \cref{e-elemshapley}, with the convention that $0 \cdot (\zero) = 0$. Indeed, under this convention, the operations $\min$, $\max$, and $x \to \alpha \cdot x$ for $\alpha \ge 0$ are continuous in $\trop$.\todo{MS: Is this good enough? SG: It seems to me ok}

\subsection{Dominions of Stochastic Mean-Payoff Games}\label{subsec:DomMPG}
We first describe the dominions in the special case of stochastic mean-payoff games, and deduce that this class of games satisfies \cref{as:good_operators}.

Intuitively speaking, $\dominion \subset \vertexsetMin$ is a dominion if player Max can force that the only states of $\vertexsetMin$ visited by the game are those in $\dominion$ provided that the initial state of the game belongs to $\dominion$.\todo{MS: We don't stay in $\dominion$ to be precise, but this is maybe not important. RK: I modified this sentence, I hope it is correct now.} To make this more precise, we use the following notation. If $\set \subset \vertexsetMin$ is any subset, then we denote by $\vertexsubset[Nat]{\set} \subset \vertexsetNat$ the set of vertices controlled by Nature whose all outgoing edges go to $\set$, i.e., $\vertexsubset[Nat]{\set} \coloneqq \{k \in \vertexsetNat \colon \sum_{l \in \set} P_{kl} = 1\}$. Moreover, we denote by $\vertexsubset[Max]{\set} \subset \vertexsetMax$ the set of vertices controlled by Max which have at least one outgoing edge that goes to $\vertexsubset[Nat]{\set}$, i.e., $\vertexsubset[Max]{\set} \coloneqq \{i \in \vertexsetMax \colon \exists k \in \vertexsubset[Nat]{\set}, (i,k) \in \edges\}$. Furthermore, we denote by $(\vertexset^\set, \edges^\set)$ the subgraph of $(\vertexset, \edges)$ induced by $\set \dunion \vertexsubset[Max]{\set} \dunion \vertexsubset[Nat]{\set}$.\todo{MS: Is this clear? SG: yes I think}

\begin{lemma}\label{le:dominions_SMPG}
A set $\dominion \subset \vertexsetMin$ is a dominion of a stochastic mean-payoff game if and only if all edges leaving $\dominion$ go to $\vertexsubset[Max]{\dominion}$, i.e., $\{i \in \vertexsetMax \colon \exists j \in \dominion, (j,i) \in \edges\} \subset \vertexsubset[Max]{\dominion}$. In particular, if $\dominion$ is a dominion, then the sets $\vertexsubset[Nat]{\dominion}$ and $\vertexsubset[Max]{\dominion}$ are nonempty. Furthermore, if $\dominion$ is a dominion, then we have the equality 
\[
\shapley^{\dominion}_j(x) = \min_{(j,i) \in \edges^\set} \Bigl(-A_{ij}+ \max_{(i,k) \in \edges^\set}\bigl(B_{ik}+ \sum_{(k,l) \in \edges^\set} P_{kl} x_l\bigr)\Bigr) \text{ for all } j\in \dominion \, .
\]
\end{lemma}

\begin{proof}
Let $\dominion \subset \vertexsetMin$ be any set. In the first place, note that for every $x \in \trop^{\dominion}$, by the definition of $\vertexsubset[Nat]{\dominion}$, for every $k \in \vertexsetNat$ we have %
\[
\sum_{(k,l) \in \edges}P_{kl}\revproj^{\dominion}(x)_l = \begin{cases}
\sum_{(k,l) \in \edges^\dominion}P_{kl}x_l  &\text{if $k \in \vertexsubset[Nat]{\dominion}$,}\\
\zero &\text{otherwise.}
\end{cases}
\]
Thus, by the definition of $\vertexsubset[Max]{\dominion}$, it follows that
\begin{align*}
\max_{(i,k) \in \edges}\Bigl(B_{ik}+ & \sum_{(k,l) \in \edges} P_{kl} \revproj^{\dominion}(x)_l \Bigr)  = \\ 
&\begin{cases}
 \max_{(i,k) \in \edges^\dominion}\bigl(B_{ik}+ \sum_{(k,l) \in \edges^\dominion} P_{kl} x_l\bigr) &\text{if $i \in \vertexsubset[Max]{\dominion}$,}\\
\zero &\text{otherwise,}
\end{cases}
\end{align*}
for all $i \in \vertexsetMax$. We conclude that\todo{MS: I put it in displaystyle to avoid going over the margin.}
\begin{align}
&\shapley_j\bigl(\revproj^{\dominion}(x)\bigr) = \label{Eq} \\
&\begin{dcases}
\min_{(j,i) \in \edges^\dominion} \Bigl(-A_{ij}+ \max_{(i,k) \in \edges^\dominion}\bigl(B_{ik}+ \sum_{(k,l) \in \edges^\dominion} P_{kl} x_l\bigr)\Bigr) & \text{if $i \in \vertexsubset[Max]{\dominion}$ for all $(j,i) \in \edges$,}\\
\zero &\text{otherwise,} \nonumber
\end{dcases}
\end{align}
for all $j \in \vertexsetMin$.

If\todo{MS: Maybe this part of the proof can be shortened?} we assume that $\dominion$ is a dominion, we necessarily have that $i \in \vertexsubset[Max]{\dominion}$ for every $j \in \dominion$ and $i \in \vertexsetMax$ such that $(j,i) \in \edges$, because otherwise~\cref{Eq} would imply that $\shapley^{\dominion}_{j}(\zerovector)=-\infty$ for some $j \in \dominion$, contradicting \cref{domin_at_zero}. Then, when $\dominion$ is a dominion,~\cref{Eq} also shows that 
\[
\shapley^{\dominion}_j(x) = \min_{(j,i) \in \edges^\dominion} \Bigl(-A_{ij}+ \max_{(i,k) \in \edges^\dominion}\bigl(B_{ik}+ \sum_{(k,l) \in \edges^\dominion} P_{kl} x_l\bigr)\Bigr) \,
\]
for every $x \in \trop^\dominion$ and $j \in \dominion$.

Assume now that $\dominion$ has the property that for every $j \in \dominion$ and $i \in \vertexsetMax$ such that $(j,i) \in \edges$ we have $i \in \vertexsubset[Max]{\dominion}$. %
Then, by~\cref{Eq} we have  $\shapley^{\dominion}_{j}(\zerovector)=\min_{(j,i) \in \edges^\dominion} \bigl(-A_{ij}+ \max_{(i,k) \in \edges^\dominion}(B_{ik})\bigr)$ for every $j \in \dominion$. Besides, since every vertex in $\vertexsetMin$ has at least one outgoing edge by \cref{assum-fin}, for every $j \in \dominion$ there exists $i\in \vertexsubset[Max]{\dominion}$ such that $(j,i) \in \edges^\dominion$. It follows that $\shapley^{\dominion}_{j}(\zerovector) > -\infty $ for every $j \in \dominion$ because $\max_{k \in \vertexsubset[Nat]{\dominion}}(B_{ik} ) > -\infty $ for every $i\in \vertexsubset[Max]{\dominion}$ by the definition of $\vertexsubset[Max]{\dominion}$. Thus, by \cref{domin_at_zero} we conclude that $\dominion$ is a dominion. 
\end{proof}

\Cref{le:dominions_SMPG} implies that if $\dominion$ is a dominion, then $\shapley^\dominion$ is the Shapley operator of a stochastic mean-payoff game on the subgraph $(\vertexset^\dominion, \edges^\dominion)$. We refer to this game as the \emph{subgame induced by $\dominion$}. We note that this subgame satisfies \cref{assum-fin}.

\begin{lemma}
If $\dominion \subset \vertexsetMin$ is a dominion, then the subgame induced by $\dominion$ satisfies \cref{assum-fin}.
\end{lemma}
\begin{proof}
\Cref{le:dominions_SMPG} implies that every vertex $j \in \dominion$ has an outgoing edge to a vertex $i \in \vertexsubset[Max]{\dominion}$. The vertices in $\vertexsubset[Nat]{\dominion}$ and $\vertexsubset[Max]{\dominion}$ have at least one outgoing edge in $\edges^\dominion$ by definition. 
\end{proof}

The following lemma characterizes the recession operators of Shapley
operators of stochastic mean-payoff games. Recall that
if $F$ is a self-map of $\R^n$, the recession operator $\hat{F}$ of $F$
is defined by~\eqref{e-def-recession}.

\begin{lemma}\label{le:recession_smpg}
For every $j \in \vertexsetMin$ we have\todo{MS: I put QED symbol in the equation for better spacing.}
\[
\recshapley_j(x) = \min_{(j,i) \in \edges}\max_{(i,k) \in \edges}\sum_{(k,l) \in \edges} P_{kl} x_l \, . \hfill
\]
\end{lemma}

\begin{lemma}\label{le:maxval_domin_SMPG}
The Shapley operators of stochastic mean-payoff games satisfy \cref{as:good_operators}.
\end{lemma}

\begin{proof}%
The fact that $\shapley$ satisfies the first part of \cref{as:good_operators} follows from \cref{rem:o-minimal-dominion}. Alternatively, to prove this part it is enough to note that $\shapley$ has an escape rate by \cref{Theo:Kohlberg} and, by \cref{le:dominions_SMPG}, the same is true for the operator $\shapley^{\dominion}$ when $\dominion$ is a dominion. 

In order to prove the second part of \cref{as:good_operators}, let $\tau$ be an optimal strategy of Max in the stochastic mean-payoff game described by $\shapley$. Consider the Shapley operator
\[
\shapley^{\tau}_j(x) \coloneqq \min_{(j,i) \in \edges} \bigl(-A_{ij}+ B_{i\tau(i)}+ \sum_{(\tau(i),l) \in \edges} P_{\tau(i)l} x_l\bigr), 
\]
of the game in which Max plays according to $\tau$. Since $\tau$ is optimal, we have $\gameval = \lim_{\ell \to \infty}\frac{(\shapley^{\tau})^{\ell}(\zerovector)}{\ell}$. Let $\recshapley^{\tau}$ be the recession operator of $\shapley^{\tau}$. By  \cref{le:recession_smpg} we have 
\[
\recshapley^{\tau}_j(x) = \min_{(j,i) \in \edges}\sum_{(\tau(i),l) \in \edges} P_{\tau(i)l} x_l
\]
Furthermore, $\recshapley^{\tau}(\gameval) = \gameval$ by \cref{prop-rec}.\todo{RK: I added this proposition to refer to this property}
Note that $\sum_{(k,l) \in \edges}P_{kl}\gameval_l = \ucw(\shapley)$ for every $k \in \vertexsubset[Nat]{\domMaxVal}$ and $\sum_{(k,l) \in \edges}P_{kl}\gameval_l < \ucw(\shapley)$ for all $k \in \vertexsetNat \setminus \vertexsubset[Nat]{\domMaxVal}$.\todo{MS: Should we recall somewhere that $\mytop(\gameval) = \ucw(\shapley)$?} Therefore, the equality $\recshapley^{\tau}(\gameval) = \gameval$ implies that for every $j \in \domMaxVal$ and every $i \in \vertexsetMax$ such that $(j,i) \in \edges$ we have $\tau(i) \in \vertexsubset[Nat]{\domMaxVal}$. In particular, every such $i$ belongs to $\vertexsubset[Max]{\domMaxVal}$ and $\domMaxVal$ is a dominion by \cref{le:dominions_SMPG}.

To complete the proof, recall that for all $j \in \domMaxVal$ we have $\gameval_j^{\domMaxVal} \le \ucw(\shapley)$ by \cref{le:smaller_val_domin}. To prove that $\ucw(\shapley) \le \gameval_j^{\domMaxVal}$ for $j \in \domMaxVal$, let $\sigma^{\domMaxVal} \colon \domMaxVal \to \vertexsubset[Max]{\domMaxVal}$ be an optimal strategy of Min in the subgame induced by $\domMaxVal$ and let $\sigma \colon \vertexsetMin \to \vertexsetMax$ be any strategy of Min that agrees with $\sigma^{\domMaxVal}$ on $\domMaxVal$. Likewise, let $\tau^{\domMaxVal} \colon \vertexsubset[Max]{\domMaxVal} \to \vertexsubset[Nat]{\domMaxVal}$ be any strategy of Max in the subgame induced by $\domMaxVal$ that agrees with $\tau$ on the set $\{i \in \vertexsubset[Max]{\domMaxVal}\colon \exists j \in \domMaxVal, (j,i) \in \edges \}$. Note that there is at least one such strategy, because we have shown above that $\tau(i) \in \vertexsubset[Nat]{\domMaxVal}$ for any $i$ that belongs to this set. Furthermore, if the initial state $j$ of the game belongs to $\domMaxVal$ and the players use the strategies $(\sigma, \tau)$, then the game never leaves the set of states $\domMaxVal \uplus \vertexsubset[Max]{\domMaxVal} \uplus \vertexsubset[Nat]{\domMaxVal}$ by the definition of $\vertexsubset[Nat]{\domMaxVal}$. Even more, $(\sigma, \tau)$ and $(\sigma^{\domMaxVal},\tau^{\domMaxVal})$ generate the same probability measures on the possible trajectories starting at $j$. In particular, we have $g_j(\sigma, \tau) = g_{j}(\sigma^{\domMaxVal},\tau^{\domMaxVal})$ for all $j \in \domMaxVal$. Hence, the optimality of $\sigma^{\domMaxVal}$ and $\tau$ implies $\gameval^{\domMaxVal}_j \ge g_{j}(\sigma^{\domMaxVal},\tau^{\domMaxVal}) = g_j(\sigma, \tau) \ge \gameval_j = \ucw(\shapley)$. 
\end{proof}

\subsection{Bit-Complexity Estimates for Stochastic Mean-Payoff Games}
We start by bounding the separation $\separ$
and the metric estimate
$R(\shapley)$, when $\shapley$ is the Shapley operator of a stochastic turn-based zero-sum game as in~\cref{e-elemshapley}.
We recall that the payoffs $A_{ij}$ and $B_{ik}$ are integers. This is not more special than assuming that they
are rational numbers (we may always rescale rational payments so that they become integers).\todo{MS: I changed these sentences slightly to go away from matrix notation.} We set
\begin{equation}\label{e-defw}
W\defi \max \left\{| A_{ij} - B_{ik} | \colon i \in \vertexsetMax, \, j\in \vertexsetMin, \, k \in \vertexsetNat \right\}\, .
\end{equation}
We also assume that the probabilities $\transition_{k j}$ are rational, and that they
have a common denominator $\comd \in \N_{> 0}$, i.e., $\transition_{k j} \in \N/\comd$,
  for all $k \in \vertexsetNat$ and $j \in \vertexsetMin$.
We say that a state $k \in \vertexsetNat$ is a \emph{significant random state} if there are at least two indices $j,j'\in \vertexsetMin$ such that
$P_{kj}>0$ and $P_{k j'}>0$. 
We denote by $\rstates(P)$ (or simply $\rstates$, when $P$ is clear from the context) the number of significant random states. %
The following separation bound improves an estimate in~\cite{boros_gurvich_makino}.

\begin{theorem}\label{cor:sep_smpg}
We have $\separ(\shapley) > 1/(\nstates\comd^{\min\{\rstates,\nstates-1\}})^2$.
\end{theorem}
We will deduce this theorem from an optimal bit-complexity result for Markov
chains, established in~\cite{skomra_bounds}.
\begin{lemma}[\cite{skomra_bounds}]\label{est_invariant_measure}
  Suppose that a Markov chain with transition matrix $Q$ and $\nstates$ states is irreducible, and that the transition probabilities are rational numbers whose denominators divide the integer $M$. Let $\rstates\coloneqq\rstates(Q)$,
  and let $\invmes \in (0,1]^{\minstates \times 1}$ be the invariant measure of the chain. Then, the least common denominator of the rational numbers $\invmes_j$, $j \in \oneto{\states}$, is not greater than $\nstates\comd^{\min\{\rstates,\nstates-1\}}$.
\end{lemma}

\begin{lemma}\label{le:maxval_denom_smpg}
Both $\ucw(\shapley)$ and $\lcw(\shapley)$ are rational numbers whose denominators are not greater than $\nstates\comd^{\min\{\rstates,\nstates-1\}}$.
\end{lemma}
\begin{proof}
Let $\sigma$ and $\tau$ be optimal strategies of players Min and Max respectively. Consider the Shapley operator
\[
\shapley^{\sigma,\tau}_j(x) \coloneqq -A_{\sigma(j)j} + B_{\sigma(j)\tau(\sigma(j))} + \sum_{(\tau(\sigma(j)),l) \in \edges}P_{\tau(\sigma(j))l}x_l
\]
of the game in which players play according to $(\sigma,\tau)$. Note that $\ucw(\shapley) = \ucw(\shapley^{\sigma,\tau})$ and $\lcw(\shapley) = \lcw(\shapley^{\sigma,\tau})$ because the strategies $\sigma, \tau$ are optimal. Furthermore, let the vector $r \in \R^n$ and the row-stochastic matrix $Q \in \R^{n \times n}$ be defined as $r_j \coloneqq -A_{\sigma(j)j} + B_{\sigma(j)\tau(\sigma(j))}$ for all $j \in \vertexsetMin$ and 
\[
Q_{jl} = \begin{cases}
P_{\tau(\sigma(j))l} &\text{if $(\tau(\sigma(j)),l) \in \edges$},\\
0 &\text{otherwise}
\end{cases}
\]
for all $j,l \in \vertexsetMin$. In this way, we get $\shapley^{\sigma,\tau}(x) = r + Qx$ for all $x \in \trop^n$. In other words, $\shapley^{\sigma,\tau}$ is an operator describing a Markov chain with rewards, in which the Markov chain has transition probabilities given by $Q$ and rewards given by $r$. The ergodic theorem of finite Markov chains~\cite[Appendix~A.4]{puterman} implies that this Markov chain has two (possibly identical) recurrent classes $U_1, U_2$ such that $\ucw(\shapley^{\sigma,\tau}) = \sum_{j \in U_1}r_j\pi^{(1)}_j$, $\lcw(\shapley^{\sigma,\tau}) = \sum_{j \in U_2}r_j\pi^{(2)}_j$, where $\pi^{(1)},\pi^{(2)}$ are the invariant measures of the restrictions of the Markov chain to these classes.\todo{MS: How much details do we want here?} Hence, the claim follows from \cref{est_invariant_measure}. 
\end{proof}

\begin{proof}[Proof of \cref{cor:sep_smpg}.]
By \cref{le:maxval_denom_smpg}, both $\ucw(\shapley)$ and $\lcw(\shapley)$ are rational numbers with denominators not greater than $\nstates\comd^{\min\{\rstates,\nstates-1\}}$. Therefore, they are either equal or satisfy $\ucw(\shapley) - \lcw(\shapley) > 1/(\nstates\comd^{\min\{\rstates,\nstates-1\}})^2$. By \cref{le:dominions_SMPG}, the same is true for any operator $\shapley^{\dominion}$ where $\dominion \subset [n]$ is a dominion.
\end{proof}

We now provide a bit-complexity estimate for the bias vector
of the Shapley operator of a stochastic mean-payoff game.
\begin{theorem}\label{le:blackwell_bias_estimate}\label{th-bias}
  Suppose that $\lcw(\shapley) = \ucw(\shapley)$. Then, there exists a vector $\bias \in \R^\vertexsetMin$ such that $\shapley(\bias) = \ucw(\shapley) + \bias$ and
  \[ R(\shapley) \leq \HNorm{\bias} \le 8\nstates \Payoff \comd^{\min\{\rstates, \nstates - 1\}}
  \enspace .
  \]
\end{theorem}
Before detailing the proof, we sketch the main ideas. 
The existence of the bias vector follows from Kohlberg's theorem (\cref{Theo:Kohlberg}). The bias is generally not unique (even up to an additive constant)
and then the main difficulty is to find a ``short'' bias. The one which
will be constructed in the proof of this theorem relies on the notion
of {\em Blackwell optimality}. This notion requires to consider
the {\em discounted} version of the game, in which the payment~\cref{e-def-mp}
is replaced by $\mathbb{E}_{\sigma \tau} \sum_{p = 0}^{\infty} (1-\alpha)^p (-A_{i_p j_p} + B_{i_p k_p})$, where $0<\alpha<1$ and $1-\alpha$ is the discount factor.  The discounted game
with initial state $j$ has a value, $x_j(\alpha)$, and the value
vector, $x(\alpha)=(x_j(\alpha))\in \R^n$, is the unique solution
of the fixed point problem $x(\alpha)=\shapley((1-\alpha)x(\alpha))$.
Then, a strategy of a player is {\em Blackwell optimal}
if it is optimal in all the discounted games with a discount
factor sufficiently close to $1$. It can be obtained
by selecting minimizing or maximizing
actions when evaluating the expression $\shapley((1-\alpha)x(\alpha))$,
for $\alpha>0$ close enough to $0$.
Kohlberg proved that $x(\alpha)$ admits a Laurent series expanstion
with a pole of order at most $1$. In fact, the result of Kohlberg
applies more generally to piecewise-affine maps that are nonexpansive
in any norm (not just to Shapley operators of stochastic mean-payoff games).

\begin{theorem}[{\cite{kohlberg}}]\label{th:kohlberg}
Let $F \colon \R^n \to \R^n$ be a piecewise-affine function that is nonexpansive in any norm. Then, for every $j \in [n]$ there exists a Laurent series $x_j \in \R((\alpha))$ of the form
\[
x_j(\alpha) = c_{j,-1} \alpha^{-1} + c_{j,0} + c_{j,1}\alpha + c_{j,2}\alpha^{2} + \dots
\]
such that $x_j(\alpha)$ converges for all small $\alpha > 0$ and $x(\alpha) = \bigl(x_1(\alpha),\dots, x_{n}(\alpha)\bigr) \in \R((\alpha))^n$ satisfies $F\bigl((1 - \alpha)x(\alpha) \big) = x(\alpha)$ for all small $\alpha > 0$. Furthermore, if we denote $\eta = (c_{1,-1}, \dots, c_{n,-1}) \in \R^{n}$ and $\bias = (c_{1,0}, \dots, c_{n,0}) \in \R^{n}$, then the equality $F\bigl(t\eta + \bias\bigr) = (t+1)\eta + \bias$ is satisfied for all sufficiently large $t > 0$.
\end{theorem}

In fact, the vector $\eta$ collecting the coefficients of $\alpha^{-1}$ in the expansions of $x_j(\alpha)$ coincides with the escape rate vector of $F$, and the vector $u$ collecting the coefficients of the term of order $1$ yields a bias vector, in a special case.
\begin{proposition}\label{prop-chi=eta}
  Under the assumptions of \cref{th:kohlberg}, we have $\chi(F)=\eta$.
  Moreover, if $F$ is the Shapley operator of a stochastic mean-payoff game,
  and if $\ucw(F)=\lcw(F)$, then $F(u)=\eta +u$, i.e., $u$ is a bias vector.
\end{proposition}
\begin{proof}
  For all sufficiently large $t$, and for all $N\in \mathbb{N}$,
  we have $F^N(t\eta +u)= (t+N)\eta +u$, from which we deduce
  that $\chi(F)=\lim_{N\to\infty} (F^N(t\eta +u))/N=\eta$.

  Moreover, we have $F(t\eta +u)=\eta + t\eta +u$. When
  $\ucw(F)=\lcw(F)$ and $F$ is the Shapley operator of a stochastic mean-payoff game, using the fact that $F$ commutes with the addition of a constant vector, we deduce that $F(u)=\eta + u$. 
\end{proof}

We shall see that the
bias vector $u$ satisfies a Poisson-type 
equation $\chi(\shapley)+u = r+Qu$,
and that this special bias vector has the remarkable
property of having a zero expectation
with respect to all invariant measures of $Q$.
Hence, the following lemma will allow us to bound the
bit-complexity of this bias vector.
\begin{lemma}[\cite{skomra_bounds}]\label{lemma-key-orth}
  Take a (possibly reducible) Markov chain with $\nstates$ states and transition probabilities given by the row-stochastic matrix $Q \in \R^{\nstates \times \nstates}$. Suppose that the transition
  probabilities are rational numbers, and let $\rstates$ and $\comd$ be as above. Let $r \in \Z^\nstates$ be an integer vector and suppose that $(\eta, \bias) \in \R^{2\nstates}$ is a solution of the system
\begin{align}
\begin{cases}
Q \eta &= \eta \\
Q \bias &= -r + \eta + \bias \, .
\end{cases}\label{e-poisson}
\end{align}
Furthermore, suppose that $\bias$ is orthogonal to all the invariant measures of the Markov chain, i.e., $\bias^{\transpose} \pi = 0$ whenever $\pi^{\transpose}Q = \pi^{\transpose}$. Then,
\[ \| \bias \|_\infty \le 4\| r \|_\infty \nstates\comd^{\min\{\rstates,\nstates-1\}}
\enspace .
\]
\end{lemma}

\begin{proof}[Proof of \cref{le:blackwell_bias_estimate}.]
Since $\shapley$ (understood as a function from $\R^{n}$ to $\R^{n}$) is piecewise-affine and nonexpansive in the sup-norm, \cref{th:kohlberg} shows that there exists two vectors $\eta, \bias \in \R^n$ such that the equality $\shapley(t\eta + \bias) = (t+1)\eta + \bias$ holds for all sufficiently large $t > 0$. If we take $v = t_0\eta + \bias$ for one such $t_0$, then we get $\shapley(v) = \eta + v$, which (by the fact that  $\lcw(\shapley) = \ucw(\shapley)$) implies that $\eta = \gameval = \ucw(\shapley)\unitvector$.\todo{MS: Some reference to earlier discussion?} In particular, we have $\shapley\bigl(t_0\ucw(\shapley) + \bias\bigr) = (t_0+1)\ucw(\shapley) + \bias$, which gives $\shapley(\bias) = \ucw(\shapley) + \bias$. Furthermore, we can suppose that the vector $\bias$ comes from a vector $x(\alpha) \in \R((\alpha))^n$ as described in \cref{th:kohlberg}. In particular, for all small $\alpha > 0$ we have
\begin{equation}\label{eq:fixed_point_discount}
\shapley\bigl((1 - \alpha)x(\alpha) \big) = x(\alpha)
\end{equation}
and, for every $j \in [n]$, the series $x_j(\alpha)$ satisfies $x_j(\alpha) = \ucw(\shapley)\alpha^{-1} + \bias_j + o(1)$. Define $y(\alpha) \coloneqq (1 - \alpha)x(\alpha) \in \R((\alpha))^n$ and note that this series satisfies $y(\alpha) = \ucw(\shapley)\alpha^{-1} + \bias - \ucw(\shapley) + o(1)$. For every $i \in \vertexsetMax$ consider the expression 
\begin{equation}\label{eq:stabilize_min}
\max_{(i,k) \in \edges}\bigl(B_{ik}+ \sum_{(k,l) \in \edges} P_{kl} y(\alpha)_l\bigr) \, .
\end{equation}
Observe that for all sufficiently small $\alpha$, the maxima in~\cref{eq:stabilize_min} are achieved by the same indices.\todo{MS: More details?} In other words, for every $i \in \vertexsetMax$ there exists $\tau(i) \in \vertexsetNat$ such that the equality $\max_{(i,k) \in \edges}\bigl(B_{ik}+ \sum_{(k,l) \in \edges} P_{kl} y(\alpha)_l\bigr) = B_{i\tau(i)} + \sum_{(\tau(i),l) \in \edges}P_{\tau(i)l}y_l(\alpha)$ holds for all small $\alpha$. In particular, for every $j \in \vertexsetMin$ we have the equality
\[
\shapley_j\bigl(y(\alpha)\bigr)  = \min_{(j,i) \in \edges}\bigl(-A_{ij} + B_{i\tau(i)} + \sum_{(\tau(i),l) \in \edges}P_{\tau(i)l}y_l(\alpha) \bigr) \, .
\]
As before, the minima on the right-hand side of this equality are achieved by the same indices if $\alpha$ is small. Hence, for every $j \in \vertexsetMin$ there exists $\sigma(j) \in \vertexsetMax$ such that the equality
\begin{equation}\label{eq:equality_fixedpoint}
\begin{aligned}
& x_j(\alpha)= \shapley_j\bigl(y(\alpha)\bigr) = -A_{\sigma(j)j} + B_{\sigma(j)\tau(\sigma(j))} + \sum_{(\tau(\sigma(j)),l) \in \edges}P_{\tau(\sigma(j))l}y_l(\alpha) \\
&= \ucw(\shapley)\alpha^{-1} -A_{\sigma(j)j} + B_{\sigma(j)\tau(\sigma(j))} - \ucw(\shapley) + \sum_{(\tau(\sigma(j)),l) \in \edges}P_{\tau(\sigma(j))l}\bias_l + o(1)
\end{aligned}
\end{equation}
holds for all small $\alpha > 0$. Therefore, we get
\[
\bias_j = -A_{\sigma(j)j} + B_{\sigma(j)\tau(\sigma(j))} - \ucw(\shapley) + \sum_{(\tau(\sigma(j)),l) \in \edges}P_{\tau(\sigma(j))l}\bias_l \,
\]
for all $j \in \vertexsetMin$. Let $Q \in \R^{n \times n}$ be the row-stochastic matrix defined as $Q_{jl} \coloneqq P_{\tau(\sigma(j))l}$ if $(\tau(\sigma(j)),l) \in \edges$ and $Q_{jl} \coloneqq 0$ otherwise, and let $r \in \R^n$ be the vector given by $r_j \coloneqq -A_{\sigma(j)j} + B_{\sigma(j)\tau(\sigma(j))}$. If $\pi \in \R^n$ satisfies $\pi^{\transpose} Q = \pi$, then~\cref{eq:equality_fixedpoint} gives $\pi^{\transpose} x(\alpha) = \pi^{\transpose} r + \pi^{\transpose} y(\alpha)$ for all small $\alpha$ and therefore $\pi^{\transpose} r = \alpha \pi^{\transpose}x(\alpha) = \ucw(\shapley) + \alpha \pi^{\transpose}\bias + o(\alpha)$ for all small $\alpha$, which gives $\pi^{\transpose} r = \ucw(\shapley)$ and $\pi^{\transpose}\bias = 0$. By defining $\eta = \ucw(\shapley)\unitvector$, the pair $(\eta,\bias)$ is a solution of the linear system given in \cref{le:blackwell_bias_estimate} for the pair $(r,Q)$. Hence, we have $\|\bias\|_{\infty} \le 4 \nstates \Payoff \comd^{\min\{\rstates, \nstates - 1\}}$. The claim follows from the fact that $\HNorm{\bias} \le 2\|\bias\|_{\infty}$. 
\end{proof}

Thanks to these estimates, we arrive at the following corollaries.
\begin{corollary}\label{coro-totalbound}
Let $F$ be a Shapley operator as above, supposing that $F$ has a bias vector and that $\ec(F)$ is nonzero. Then,
Procedure
\textsc{ValueIteration} stops
after
\begin{align}
N_{\textrm{vi}}\leq 8\nstates^{2} \Payoff \comd^{2\min\{\rstates, \nstates - 1\}} 
\label{e-firstbound}
\end{align}
iterations and correctly decides which of the two players is winning.
\end{corollary}
\begin{proof}%
By \cref{le:blackwell_bias_estimate} we have $R(\shapley) \leq 8\nstates \Payoff \comd^{\min\{\rstates, \nstates - 1\}}$. Moreover, $\ec(\shapley)$ is a rational number with denominator at most $\nstates \comd^{\min\{\rstates, \nstates - 1\}}$ by \cref{le:maxval_denom_smpg}. In particular, $|\ec(\shapley)| \ge (\nstates \comd^{\min\{\rstates, \nstates - 1\}})^{-1}$. Hence, the claim follows from \cref{Th:NitsBound}. 
\end{proof}

\begin{remark}
  When specialized to deterministic mean-payoff games, i.e., when $\rstates=0$,
  \cref{coro-totalbound} yields $N_{\textrm{vi}}=O(\nstates^{2} \Payoff)$
  which is precisely the bound that follows from the analysis of value iteration
  by Zwick and Paterson~\cite{zwick_paterson}. 
\end{remark}

Procedure \textsc{ApproximateConstantMeanPayoff} returns approximate optimality certificates, i.e., sub and super-eigenvectors $x$ and $y$ that satisfy
$\kappa - \delta/8 + x \le \shapley(x)$ and $\lambda + \delta/8 + y \ge \shapley(y)$. Then, we synthesize strategies as follows.
For each state $j\in \vertexsetMin$, Min selects a minimizing state $i \in \vertexsetMax$ in the
expression
\begin{align*}
\shapley_j(y) = \min_{(j,i) \in \edges} \Bigl(-A_{ij}+ \max_{(i,k) \in \edges}\bigl(B_{ik}+ \sum_{(k,l) \in \edges} P_{kl} y_k\bigr)\Bigr) \enspace .
\end{align*}
In this way, one gets a positional strategy which guarantees to Min
a value at most $\lambda + \delta/8$. A similar method
is used to construct a positional strategy of Max.

\begin{corollary}\label{cor-perturb}
  Suppose that $\shapley$ has a bias vector and let $\mu \defi\nstates\comd^{\min\{\rstates,\nstates-1\}}$.
Then, Procedure
\textsc{ApproximateConstantMeanPayoff}, applied to $F$ with $\delta\coloneqq\mu^{-2}$, 
terminates in at most
\[
128 n^3 \Payoff \comd^{3\min\{\rstates, \nstates - 1\}} 
\] 
calls to the oracle. Moreover the interval returned by this procedure
contains a unique rational number of denominator at most $\mu$,
which coincides with the value, and the above synthesis
procedure, exploiting the approximate optimality certificats $x$ and $y$,
provides optimal strategies.
\end{corollary}

\begin{proof}%
By \cref{le:blackwell_bias_estimate}, the value $R \coloneqq 8\nstates \Payoff \comd^{\min\{\rstates, \nstates - 1\}}$ satisfies the conditions of \cref{pr:approx_constant_value}. Therefore, \textsc{ApproximateConstantMeanPayoff} stops after at most $\lceil 8R/\delta \rceil = 64 \nstates^3 \Payoff \comd^{3\min\{\rstates, \nstates - 1\}}$ iterations of its first loop. In particular, it makes no more than $128 \nstates^3 \Payoff \comd^{3\min\{\rstates, \nstates - 1\}}$ calls to the oracle during its entire execution. Furthermore, it outputs an interval $[a,b]$ that contains $\ec(\shapley)$ and is of width at most $\delta$. Since $\ec(\shapley)$ is a rational number of denominator at most $\mu$ by \cref{le:maxval_denom_smpg} and $\delta \coloneqq \mu^{-2}$, $\ec(\shapley)$ is the unique rational number in $[a,b]$ of denominator at most $\mu$. Furthermore, the procedure outputs two vectors $x,y \in \R^n$ such that $a + x \le \shapley(x)$ and $b + y \ge \shapley(y)$. A pair of optimal strategies $(\sigma, \tau)$ is then synthesized as follows. For every $i \in \vertexsetMax$ let $\tau(i) \in \vertexsetNat$ be a vertex such that
\[
\max_{(i,k) \in \edges}(B_{ik} + \sum_{(k,l) \in \edges}P_{kl}x_l) = B_{i\tau(i)} + \sum_{(\tau(i),l) \in \edges}P_{\tau(i)l}x_l \, .
\]
Analogously, for every $j \in \vertexsetMin$, let $\sigma(j) \in \vertexsetMax$ be a vertex such that
\[
\min_{(j,i) \in \edges}\bigl(-A_{ij} + \max_{(i,k) \in \edges}(B_{ik} + \sum_{(k,l) \in \edges}P_{kl}y_l) \bigr) = -A_{\sigma(j)j} + \max_{(i,k) \in \edges}(B_{\sigma(j)k} + \sum_{(k,l) \in \edges}P_{kl}y_l) \, .
\]
We claim that $\sigma$ is optimal for Min and $\tau$ is optimal for Max. Indeed, let $\shapley^{\sigma}$ be the Shapley operator of the game in which Min uses $\sigma$. Then, the definition of $\sigma$ gives the inequality $\shapley^{\sigma}(y) = \shapley(y) \le b + y$. Hence, $\ucw(\shapley^{\sigma}) \le b$. Even more, since \cref{le:maxval_denom_smpg} applies to $\shapley^{\sigma}$, we get $\ucw(\shapley^{\sigma}) \le \ec(\shapley)$. If $\tau'$ is any strategy of Max, and we denote by $\shapley^{\sigma, \tau'}$ the operator obtained by fixing both strategies, then we have $\shapley^{\sigma, \tau'}(z) \le \shapley^{\sigma}(z)$ for any $z \in \R^n$ and hence $g_j(\sigma, \tau') \le \ucw(\shapley^{\sigma,\tau'}) \le \ucw(\shapley^{\sigma}) \le \ec(\shapley)$ for every $j \in \vertexsetMin$.\todo{MS: Maybe a reference to the section in which we introduce SMPG.} Analogously, if $\tau$ is a strategy of Max defined above and $\sigma'$ is any strategy of Min, then $\ec(\shapley) \le g_j(\sigma',\tau)$ for all $j \in \vertexsetMin$. Hence, the strategies $(\sigma, \tau)$ are optimal. 
\end{proof}
\begin{remark}
  As explained in the proof of~\Cref{cor-perturb}, we obtain from
  Procedure~\textsc{ApproximateConstantMeanPayoff} an interval
  $[a,b]$ which contains the exact value $\ec(\shapley)$. Then, this value
  can be found by the rational search technique, see, e.g.,~\cite{kwek_mehlhorn,forisek}. Alternatively, since the synthesis method also returns a pair of optimal strategies, the exact value can be computed by solving the $0$-Player
  stochastic mean-payoff problem induced by these strategies.
  \end{remark}
\begin{corollary}\label{finding_top_class_complexity0}
The set of initial states with maximal value of a stochastic mean-payoff game can be found by performing at most $65\nstates^4 \Payoff \comd^{3\min\{\rstates, \nstates - 1\}}$
  \todo{SG: to be checked, may be an explicit constant. MS: Done.}
  calls to an oracle approximating its Shapley operator $\shapley$
  with precision $\epsilon \coloneqq 1/(8\mu^2)$, where $\mu \defi\nstates\comd^{\min\{\rstates,\nstates-1\}}$.
\end{corollary}
\begin{proof}%
Let $\delta \coloneqq 1/\mu^2$. By \cref{cor:sep_smpg} we have $\delta < \separ(\shapley)$. Moreover, if we denote by $\domMaxVal \subset \vertexsetMin$ the set of initial states with maximal value, then by applying \cref{le:blackwell_bias_estimate} to $\shapley^{\domMaxVal}$ we see that the number $R\coloneqq 8\Payoff\mu$ satisfies the conditions of \cref{th:topclass}, so Procedure \TopClass{} applied to $\shapley$, with $\delta \coloneqq 1/\mu^2$ and $R \coloneqq 8\Payoff\mu$, is correct, and finds the set of initial states with maximal value. Furthermore, by \cref{topclass_oracle}, Procedure \TopClass{} can be executed by performing at most $n^2 + n\lceil8R/\delta\rceil = n^2 + 64\nstates^4 \Payoff \comd^{3\min\{\rstates, \nstates - 1\}}$ calls to the oracle that approximates $\shapley$. 
\end{proof}

\todo{MS: I added some remarks about the approximation oracles and about comparison with the pumping algorithm. Do we want to keep them in some form? SG: very important remarks, I quoted the second one in the introduction.}

\begin{remark}\todo{SG: consistency with new matrix-free presentation. MS: I changes the sentence a bit.}
  If we are given explicitly the graph of the game together with the probabilities $P_{kj}$ and the payoffs $A_{ij}$ and $B_{ik}$, then the operator $\shapley$ can be evaluated exactly in $O(E)$ complexity, where $E$ is the number of edges of the graph representing the stochastic mean-payoff game. In particular, there is no need to construct an approximation oracle in order to apply the results of this section. Nevertheless, even in this case it may be beneficial to use an approximation oracle. Indeed, if we evaluate $\shapley$ exactly, then each value iteration $u \coloneqq \shapley(u)$ increases the number of bits needed to encode $u$. As a result,
  value iteration
  would require exponential memory. In order to avoid this problem, one can replace $\shapley$ with an approximation oracle $\apshapley$ obtained as follows. Let $\mu \coloneqq \nstates\comd^{\min\{\rstates,\nstates-1\}}$ and $\epsilon \coloneqq 1/(8\mu^2)$. Given $x \in \trop^\vertexsetMin$, we first compute $y \coloneqq \shapley(x)$ exactly and then round the finite coordinates of $y$ in such a way that the rounded vector $\tilde{y}$ satisfies $|y_j - \tilde{y}_j| \le \epsilon$ whenever $y_j \neq \zero$ and $\tilde{y}_j$ is a rational number with denominator at most $8\mu^2$. One can check that if we use such a procedure as an approximation oracle, then all the algorithms presented in this section require $O\bigl(nE\log(n\comd\Payoff)\bigr)$ memory, which is polynomial in the size of the input.
\end{remark}

\begin{remark}\label{rk-compar-gurvich}
  Since a single call to the oracle approximating $\shapley$ can be done in $O(E)$ arithmetic operations, by combining \cref{finding_top_class_complexity0} with \cref{cor-perturb} we see that the set of states with maximal value, and a pair of optimal strategies within this set can be found in $O(n^4E\Payoff \comd^{3\min\{\rstates, \nstates - 1\}})$ complexity. This should be compared with the algorithm \textsc{BWR-FindTop} from \cite{boros_gurvich_makino} which achieves the same aim using a pumping algorithm instead of value iteration. If we combine the estimate from~\cite{skomra_bounds}
  with the complexity bound presented in~\cite{boros_gurvich_makino} for the pumping algorithm, then we get that \textsc{BWR-FindTop} has $O(V^6 EW\rstates 2^{\rstates}\comd^{4\rstates} + V^3E\Payoff \log \Payoff)$ complexity, where $V$ is the number of vertices of the graph of the game. In particular, our result gives a better complexity bound. Furthermore, the authors of \cite{boros_gurvich_makino} show that, given an oracle access to \textsc{BWR-FindTop} and to another oracle that solves deterministic mean-payoff games, one can completely solve stochastic mean-payoff games with pseudopolynomial number of calls to these oracles, provided that $\rstates$ is fixed. Hence, we can speed-up this algorithm by replacing the oracle \textsc{BWR-FindTop} with our algorithms.
\end{remark}
\todo{SG: new remark added}
\begin{remark}\label{rk-concurrent}
  \Cref{pr:approx_constant_value} leads to a complexity bound for the more general class of {\em concurrent}
  stochastic mean-payoff games, originally considered
  by Gillette~\cite{gillette}, see~\cite{RS01,Ney03,sorin_repeated_games} for more information.
  In these games, the two players play simultaneously using randomized strategies. This leads
  to a dynamic programming operator $\shapley:\R^n\to\R^n$ of the form
  \[
  \shapley_i(x) =\min_{\alpha\in \Delta(A_i)}\max_{\beta\in\Delta(B_i)}
  \sum_{a\in A_i,b\in B_i} \alpha(a)\beta(b) (r_i^{ab}+ \sum_{j\in[n]}P_{ij}^{ab}x_j)\enspace, i\in [n],
  \]
  where $A_i,B_i$ are finite action sets, $\Delta(A)$ denotes the set of probability
  measures on a finite set $A$, so that an element $\alpha$ of $\Delta(A)$ can be identified
  to a map $\alpha: A\to \R_{\geq 0}$ such that $\sum_{a\in A}\alpha(a)=1$, $r_i^{ab}$ denotes
  an instanteneous reward received by Player Max from Player Min in state $i$, and $P^{ab}_{ij}$
  the transition probability from $i$ to $j$, when action $a$ is selected by Player Min and action $b$ is selected
  by Player Max. We make {\em Assumption $(\mathcal{I})$}: there is a randomized positional strategy $\sigma^*$ of Player Min, such that for all deterministic positional strategies $\tau$ of Player Max, the stochastic transition matrix $P^{\sigma^*,\tau}$ obtained
by applying these two strategies  is irreducible. We set
$  \bar{T}_{\max}\coloneqq \max_{\tau\in \mathcal{T}} T_{\max}(P^{\sigma^*,\tau})$,
    where the maximum is taken over the set $\mathcal{T}$ of deterministic positional strategies $\tau$ of Player Max,
    and $T_{\max}(P)$ denotes the maximal expected first passage time between any two states,
    for a Markov chain with transition matrix $P$. We also set $\|r\|_\infty\coloneqq\max_{i\in[n],a\in A_i,b\in B_i}|r_i^{a,b}|$.
    Then, under Assumption $(\mathcal{I})$, 
    \begin{enumerate}\renewcommand{\theenumi}{\roman{enumi}}
    \item\label{e-exists-it}
      There exists $u\in \R^n$ and $\lambda\in \R$ such that $\shapley(u)=\lambda + u$;
      \item \label{e-bound-it} Any such vector $u$ satisfies
        $\|u\|_H \leq 2\|r\|_\infty \bar{T}_{\max}$.
    \end{enumerate}
    To show~\eqref{e-exists-it}, by~\Cref{th-exists-rec}, it suffices
    to check that the recession operator
    $\hat{\shapley}(x)$
  has only fixed points in $\R \unitvector$. Here,
  $\hat{\shapley}(x)\leq \sup_{\tau\in \mathcal{T}} P^{\sigma^*,\tau}x$. Hence,  if $\hat{\shapley}(x)=x$,
  selecting a positional strategy $\tau$ that achieves the latter supremum, we get $x\leq P^{\sigma^*,\tau}x$,
  and since $P^{\sigma^*,\tau}$ is irreducible, this entails that $x\in \R \unitvector$, showing
  that property~\eqref{e-exists-it} holds. Now, if $\shapley(u)= \lambda + u$, denoting by $r^{\sigma^*,\tau}_i$
  the expected instantaneous payment in state $i$ under the strategies $\sigma^*$ and $\tau$,  and setting
  $r^{\sigma^*,\tau}\coloneqq (r^{\sigma^*,\tau}_i)_{i\in [n]}\in\R^n$, 
  we get $\lambda + u \leq \sup_{\tau\in\mathcal{T}} r^{\sigma^*,\tau} + P^{\sigma^*,\tau}u$,
  and selecting a deterministic positional strategy $\tau$ which achieves
  the supremum, we deduce that $\lambda + u \leq r^{\sigma^*,\tau} + P^{\sigma^*,\tau}u$.
  Observe that if we have an inequality of the form $u\leq f + Pu$ where $P$ is an irreducible
  stochastic matrix, and $f\in \R^n$, then, $u_i -u_j\leq (\max_k f_k)T_{\max}(P)$
  for all $i,j\in[n]$.
  Applying this to $f=r^{\sigma^*,\tau}-\lambda$, and using the fact that $|\lambda|\leq \|r\|_\infty$, we deduce
  property~\eqref{e-bound-it}. Therefore, under Assumption~$(\mathcal{I})$, \Cref{pr:approx_constant_value}
  yields a $\delta$ approximation of the value $\lambda$ of a concurrent stochastic mean-payoff game
  in $O(\|r\|_\infty \bar{T}_{\max}/\delta)$ iterations. This saves an $O(\log\delta)$ factor
  by comparison with the bound
$O(\|r\|_\infty T'_{\max}\log(\delta)/\delta)$ proved in~\cite[Th.~18]{chatterjee_ibsen-jensen}, under
  the stronger assumption that all pairs of strategies yield an irreducible matrix,
  where $T'_{\max}\geq \bar{T}_{\max}$ denotes the maximal expected first passage time, taken
  now over {\em all pairs} of strategies, and still between all pairs of states. We note
  however that Assumption~$(\mathcal{I})$ is still restrictive.
Boros, Gurvich, Elbassioni and Makino considered in~\cite{Boros2016} the broader class
  of concurrent games in which the mean payoff is independent of the initial state,
  without making any irreducibility assumption. Whether the present methods can be applied to this
  broader class of games is left for future work.
  We also note that for the restricted class of stochastic games such that
  {\em every pair} of positional strategies yields a unichain matrix,
  the Krasnoselskii--Mann variant of relative value iteration discussed
  in~\Cref{rk-comparewithkm}  leads
  to an $O(|\log\delta|$) bound, as shown in the recent work~\cite{akian2023}.\todo[color=red!30]{SG: I added a pointer to the new paper and a sentence explaining the diff.}
\end{remark}
\section{Solving Entropy Games With Bounded Rank}\label{sec-entropy}
Recall that the dynamic programming operator $T$ of an entropy game,
as well as its conjugate $\shapley$, which we call the Shapley operator
of an entropy game, were defined in~\cref{eq:entropy_operator} and~\cref{e-def-conj}. As in the last section, we denote $n \coloneqq \card{\vertexsetD}$ and we put $\Weights \coloneqq \max_{(p,k) \in \edges}\eweight_{pk}$.

\subsection{Dominions of Entropy Games}
We first describe the dominions of entropy games, and verify
that this class of games satisfies \cref{as:good_operators}.
Given a set $\set \subset \vertexsetD$, we denote by $\vertexsubsetit[P]{\set} \subset \vertexsetP$ the set of states of People that have at least one outgoing edge that goes to $\set$, i.e., $\vertexsubsetit[P]{\set} \coloneqq \{p \in \vertexsetP \colon \exists l \in \set, (p,l) \in \edges\}$. In the same way, we denote by $\vertexsubsetit[T]{\set} \subset \vertexsetT$ the set of states of Tribune that have at least one outgoing edge that goes to $\vertexsubsetit[P]{\set}$, i.e., $\vertexsubsetit[T]{\set} \coloneqq \{t \in \vertexsetT \colon \exists p \in \vertexsubsetit[P]{\set}, (t,p) \in \edges\}$. The following lemma characterizes the dominions of entropy games.

\begin{lemma}\label{le:entropy_domin}
A set $\dominion \subset \vertexsetD$ is a dominion of the operator $\shapley$ if and only if every outgoing edge of every state of $\dominion$ goes to $\vertexsubsetit[T]{\dominion}$, i.e., if for every pair $(k,t) \in \dominion \times \vertexsetT$ we have $(k,t) \in \edges \implies t \in \vertexsubsetit[T]{\dominion}$. Furthermore, if $\dominion$ is a dominion, then $\shapley^\dominion = \log \circ \mshapley^\dominion \circ \exp$, where for all $x \in \R_{>0}^\dominion$ and all $k \in \dominion$ we define 
\[
\mshapley^\dominion_k (x) \coloneqq  \min_{\substack{(k,t) \in \edges \\ t \in \vertexsubsetit[T]{\dominion}}} \max_{\substack{(t,p) \in \edges\\ p \in \vertexsubsetit[P]{\dominion}}} \sum_{\substack{(p,l) \in \edges \\ l \in \dominion}} \eweight_{pl} x_{l} \, .
\]
\end{lemma}
\begin{proof}
The proof is similar to the proof of \cref{le:dominions_SMPG}. By \cref{domin_at_zero}, $\dominion$ is a dominion if and only if for every $k \in \dominion$ and every $t \in \vertexsetT$ such that $(k,t) \in \edges$ we have
\begin{equation}\label{eq:domin_entropy}
\max_{(t,p) \in \edges}\sum_{\substack{(p,l) \in \edges \\ l \in \dominion}}\eweight_{pl} \ind_l > 0 \, .
\end{equation}
By definition, we have $\sum_{\substack{(p,l) \in \edges \\ l \in \dominion}}\eweight_{pl} \ind_l > 0$ if and only if $p \in \vertexsubsetit[P]{\dominion}$, and so \cref{eq:domin_entropy} holds if and only if $t \in \vertexsubsetit[T]{\dominion}$. The second part of the claim follows from the definition of $\shapley^{\dominion}$.\todo{MS: This proof has considerably less details than the proof of the analogous claim for stochastic mean payoff games. Could we make them more uniform?} 
\end{proof}
In particular, as in the case of stochastic mean-payoff games, $\shapley^{\dominion}$ is the Shapley operator of a smaller entropy game that takes place on the state space $\dominion \dunion \vertexsubsetit[T]{\dominion} \dunion \vertexsubsetit[P]{\dominion}$.

\begin{lemma}\label{le:entropy_domin_nondeg}
If $\dominion$ is a dominion, then the subgame induced by $\dominion$ satisfies \cref{as:entropy_nondeg}.
\end{lemma}
\begin{proof}
This is an immediate consequence of \cref{le:entropy_domin}. 
\end{proof}

The following immediate lemma characterizes the recession operators of Shapley
operators of entropy games.
\begin{lemma}\label{le:recession_entropy}
The recession operator of the conjugate operator $\shapley$ is given by 
\[
  \hat{\shapley}_k(x) =   \min_{(k,t)\in \edges} \max_{(t,p)\in \edges} \max_{(p,l)\in \edges} x_l,
  \qquad \forall k\in \vertexsetD \enspace .%
   \hfill\]
\end{lemma}

\begin{lemma}\label{le:good_entropy_aux}
Let $\domMaxVal$ be the set of states of maximal value. If $\tau$ is an optimal strategy of Tribune, then for any $k \in \domMaxVal$ and $t \in \vertexsetT$ such that $(k,t) \in \edges$ we have $\tau(t) \in \vertexsubsetit[P]{\domMaxVal}$. Moreover, if $\sigma$ is an optimal strategy of Despot and $k \notin \domMaxVal$, then $\sigma(k) \notin \vertexsubsetit[T]{\domMaxVal}$.
\end{lemma}
\begin{proof}
Let $\gameval$ be the escape rate of $\shapley$ and $\tau \colon \vertexsetT \to \vertexsetP$ an optimal strategy of Tribune. We consider the reduced operator
\(
\mshapley_k^{\tau}(x) \coloneqq \min_{(k,t) \in \edges}\sum_{(\tau(t),l) \in \edges}\eweight_{\tau(t)l} x_{l} \),
obtained from the game in which Tribune plays according to $\tau$. Let $\shapley^{\tau} = \log \circ \mshapley^{\tau} \circ \exp$. By the optimality of $\tau$ we get $\chi(\shapley^{\tau}) = \gameval$. Since $\hat{\shapley^{\tau}}(\gameval) = \gameval$ by \cref{prop-rec},\todo{MS: Do we repeat the reference here? RK: I added \Cref{prop-rec} with this aim} \cref{le:recession_entropy} implies that
\(
\mytop(\gameval) = \gameval_k = \min_{(k,t)\in \edges} \max_{(\tau(t),l)\in \edges} \gameval_l \),
for all $k \in \domMaxVal$. Note that $\mytop(\gameval)  = \max_{(p,l)\in \edges} \gameval_l$ if and only if $p \in \vertexsubsetit[P]{\domMaxVal}$. Hence, for every $k \in \domMaxVal$ and $t \in \vertexsetT$ we have
\(%
(k,t) \in \edges \implies \tau(t) \in  \vertexsubsetit[P]{\domMaxVal}\).
This proves the first claim.

To prove the second claim, let $\sigma \colon \vertexsetD \to \vertexsetT$ be an optimal strategy of Despot. Again, we consider the reduced operator
\(
\mshapley_k^{\sigma}(x) \coloneqq \max_{(\sigma(k),p) \in \edges} \sum_{(p,l) \in \edges} \eweight_{pl} x_{l}\),
and its conjugate $\shapley^{\sigma} = \log \circ \mshapley^{\sigma} \circ \exp$. Since $\hat{\shapley^{\sigma}}(\gameval) = \gameval$, \cref{le:recession_entropy} implies that
\(%
\gameval_k = \max_{(\sigma(k),p) \in \edges} \max_{(p,l) \in \edges} \gameval_{l}\),
for all $k \in \vertexsetD$. Hence, we get $k \in \domMaxVal \iff \sigma(k) \in \vertexsubsetit[T]{\domMaxVal}$, proving the second claim. 
\end{proof}

\begin{lemma}\label{le:good_entropy}
The Shapley operator $\shapley$ of an entropy game satisfies \cref{as:good_operators}.
\end{lemma}\todo{RK: I divided the previous lemma in which it was stated that the Shapley operator of an entropy game satisfies \cref{as:good_operators} into two lemmas (\cref{le:good_entropy_aux} and \cref{le:good_entropy}) since it seems to me that in this way it easier to follow. If you disagree with this change, the previous version is still in the file.}
\begin{proof}
The operator $\shapley$ satisfies the first part of \cref{as:good_operators} by \cref{rem:o-minimal-dominion}. To deal with the second part, we proceed as in the proof of \cref{le:maxval_domin_SMPG}. Let $\gameval$ be the escape rate of $\shapley$ and $\domMaxVal$ be the set of states of maximal value. 

In the first place, note that if $k \in \domMaxVal$ and $(k,t) \in \edges$, then the first claim of \cref{le:good_entropy_aux} implies $t$ belongs to $\vertexsubsetit[T]{\domMaxVal}$. Thus, by \cref{le:entropy_domin} we conclude that $\domMaxVal$ is a dominion.

It remains to prove that $\lcw(\shapley^{\domMaxVal}) = \mytop(\gameval)$. To do so, let $\tau \colon \vertexsetT \to \vertexsetP$ be an optimal strategy of Tribune, $\sigma \colon \vertexsetD \to \vertexsetT$ be an optimal strategy of Despot, and $\tilde{\sigma} \colon \domMaxVal \to \vertexsubsetit[T]{\domMaxVal}$ be an optimal strategy of Despot in the subgame induced by $\domMaxVal$. We extend $\tilde{\sigma}$ to a strategy $\tilde{\sigma} \colon \vertexsetD \to \vertexsetT$ by setting $\tilde{\sigma}(k) \coloneqq \sigma(k)$ for all $k \notin \domMaxVal$. Let $\mshapley^{\tilde{\sigma},\tau}$
be the operator of the game in which the players play according to the strategies $\tilde{\sigma}$ and $\tau$, and let $\shapley^{\tilde{\sigma},\tau} \coloneqq \log \circ \mshapley^{\tilde{\sigma},\tau} \circ \exp$ be its conjugate. Then, we have $\chi(\shapley^{\tilde{\sigma},\tau}) \ge \gameval$ by the optimality of $\tau$. Furthermore, if $k \notin \domMaxVal$, the second claim of \cref{le:good_entropy_aux} implies that the dipaths in $\mathscr{G}^{\tilde{\sigma},\tau}$ that start at $k$ are the same as in the graph $\mathscr{G}^{\sigma,\tau}$. Therefore, by \cref{prop-rothblum}, $\chi(\shapley^{\tilde{\sigma},\tau})_k = \gameval_k$ for all $k \notin \domMaxVal$ and so 
\begin{equation}\label{eq:smaller_val_dipath}
k \notin \domMaxVal \implies \chi(\shapley^{\tilde{\sigma},\tau})_k < \mytop(\gameval) \, .
\end{equation}
Moreover, if $k \in \domMaxVal$, then $\chi(\shapley^{\tilde{\sigma},\tau})_k  \ge \mytop(\gameval)$ and \cref{prop-rothblum} shows that there exists a strongly connected component $\mathscr{C}$ of $\mathscr{G}^{\tilde{\sigma},\tau}$ such that $\specrad(M^{\tilde{\sigma},\tau}[\mathscr{C}]) \ge \mytop(\gameval)$, where $M^{\tilde{\sigma},\tau} \in \R^{n \times n}$ is the ambiguity matrix associated with the strategies $\tilde{\sigma}$ and $\tau$, and such that there exists a dipath from $k$ to $\mathscr{C}$ in $\mathscr{G}^{\tilde{\sigma},\tau}$. By~\cref{eq:smaller_val_dipath}, this dipath does not go through any vertex in $\vertexsetD \setminus \domMaxVal$. Therefore, this dipath only goes through vertices in 
\[
\set \coloneqq \domMaxVal \dunion \{\tilde{\sigma}(k) \colon k \in \domMaxVal\} \dunion \{\tau(\tilde{\sigma}(k)) \colon k \in \domMaxVal\} \subset \domMaxVal \dunion \vertexsubsetit[T]{\domMaxVal} \dunion \vertexsubsetit[P]{\domMaxVal}
\]
and the component $\mathscr{C}$ is included in $\set$. Consider the subgame induced by $\domMaxVal$ and suppose that in this game Tribune uses a strategy $\tilde{\tau} \colon \vertexsubsetit[T]{\domMaxVal} \to \vertexsubsetit[P]{\domMaxVal}$ that agrees with $\tau$ on the set $\{t \in \vertexsubsetit[T]{\domMaxVal} \colon \exists k \in \domMaxVal, (k,t) \in \edges\}$. Note that there exists at least one such strategy by the first claim of \cref{le:good_entropy_aux}. Furthermore, since $\mathscr{C}$ is included in $\set$, \cref{prop-rothblum} implies that the value of state $k$ in the game obtained by fixing $(\tilde{\sigma}, \tilde{\tau})$ is not smaller than its value in the original game. By the optimality of $\tilde{\sigma}$ and since $k \in \domMaxVal$ was arbitrary, we have $\lcw(\shapley^\domMaxVal) \ge \mytop(\gameval)$. The other inequality follows from \cref{le:smaller_val_domin}.
\todo{MS: I improved some details. RK: I slightly revised this proof.} 
\end{proof}

\subsection{Bit-Complexity Bounds for Entropy Games}
We define the {\em rank} of the
entropy game to be the maximum of the ranks of the ambiguity matrices,
see \cref{def-ambiguity}. The following result will be established
by combining a separation bound of Rump~\cite{rump} for algebraic
numbers, with bounds on determinants of nonnegative matrices
with entries in an interval, building on the study
of Hadamard's maximal determinant problem for matrices
with entries in $\{0,1\}$~\cite{cohn}.

\begin{theorem}\label{th-sepentropy}
  Suppose two pairs of strategies yield distinct values in an entropy game of rank $r$, with $n$ Despot's states. Then, these values differ at least by $\nu_{n,r}^{-1}$ where
  \[
  \nu_{n,r}\coloneqq
    2^{r} (r+1)^{8r}r^{-2r^2+r+1} (ne)^{4r^2}\max( 1,  \Weights/ 2)^{4r^2} \enspace .
  \]
\end{theorem}
\begin{proof}%
  Step~1. First, we note that if $C$ is a $k\times k$ matrix with entries in $[0,\Weights]$,
  we have 
  \begin{align}
|\det C |\leq   (k+1) \sqrt{k}^k (\Weights/2)^k \enspace .\label{e-nice}
  \end{align}
  Indeed, this follows by writing $C=B+\Weights J/2$, where $J$ is the matrix whose entries are all equal to one and $B=C-\Weights J/2$, so that
  $-\Weights/2\leq B_{ij}\leq \Weights/2$, then by expanding
  $\det C$ by multilinearity as a function of columns, i.e.,
  $\det C= \det B + \sum_{i=1}^k (\pm 1) \det (\hat{B}_i, (\Weights/2) \unitvector)$ in which %
  $\hat{B}_i$ denotes the matrix obtained by deleting column $i$ of $B$,
  noting that the other terms of this multilinear expansion are zero
  owing to repeated appearances of columns proportional to $\unitvector$.
  By Hadamard's inequality, $|\det B|\leq k^{k/2} (\Weights/2)^k$, and similarly
  for each $|\det (\hat{B}_i,(\Weights/2) \unitvector)|$, which gives~\cref{e-nice}. (We note for information that when all the entries of $C$ belong to $\{0,1\}$, the bound can be refined to $|\det C|\leq \sqrt{k+1}^{k+1}/2^k$, see~\cite{cohn}.)

  Step~2. Let $A$ be the $n\times n$ ambiguity matrix
  associated with a pair of strategies. Since $A$ has rank at most $r$,
  it has at most $r$ non-zero eigenvalues, and so, the $n-r$ first
  coefficients of the characteristic polynomial $Q^A\coloneqq\det (XI-A)=\sum_{k=0}^n Q^A_k X^k$ are zero. 
  The coefficient $Q^A_{n-k}$ of $X^{n-k}$ in this polynomial is given by the
  sum of principal minors of order $k$ of $A$. Hence, using the inequality
  established in Step~1, together
  with the inequality ${n\choose k}\leq (ne/k)^k$, and using the fact that the map $k\mapsto (ne/k)^k \sqrt{k}^k$ is nondecreasing for $k\leq n$, we get
  \begin{align*}
  S_{n,r}\coloneqq \sum_{k=0}^n |Q^A_{n-k}|
& \leq \sum_{k=0}^r {{n}\choose{k}} (k+1) \sqrt{k}^k (\Weights /2)^k \\ & \leq (r+1)^2 \big(\frac{ne}{r}\big)^r  \sqrt{r}^r \max(1, {\Weights}/{2})^r \enspace .
  \end{align*}

  Step~3. Now, if $A,B$ are two $n\times n$ ambiguity matrices, consider
  the product $Q^AQ^B$ of the characteristic polynomials of $A$ and $B$.
  Let $S$ denote the sum of absolute values of the coefficients of $Q^AQ^B$.
Then,
  \[
  S \leq S_{n,r}^2 \leq a_{n,r}:=(r+1)^4 \big(\frac{ne}{r}\big)^{2r}r^{r}
  \max(1, \Weights/2)^{2r}  \enspace .
  \]
  Now, we use a theorem of Rump, ~\cite[Th.~3]{rump}, showing that for a polynomial
  of degree $d$ with integer coefficients whose sum of absolute
  values is bounded by $S$, the distance between any two distinct
  real roots is at least
  \begin{align}\label{e-rump-ineq}
  2\sqrt{2}   \big(   d^{\frac d 2 +1} (S+1)^d \big)^{-1}
  \enspace .
\end{align}
We apply this result to the polynomial $\bar{Q}=X^{-\ell}Q^AQ^B$, where $\ell$ is the multiplicity of the root $0$ in $Q^A Q^B$, and observe that $\bar{Q}$
is of degree $d\leq 2r$, showing that the roots of this
  polynomial are separated at least by the inverse
  of the following quantity
\begin{align*}
2^{-3/2}  (2r)^{r+1} \Big( a_{n,r} +1 \Big)^{2r}
& \leq
2^{-3/2} (2r)^{r+1}
a_{n,r}^{2r} \exp\big(2r\log(1+a_{n,r}^{-1})\big)\\
& 
\leq
2^{-3/2} (2r)^{r+1}
a_{n,r}^{2r} \exp(2ra_{n,r}^{-1})
\leq
2^r r^{r+1}
a_{n,r}^{2r}\\
& = 2^{r}  r^{r+1} (r+1)^{8r} \big(\frac{ne}{r}\big)^{4r^2}r^{2r^2} \max( 1, \Weights/ 2)^{4r^2} 
  \enspace,
\end{align*}
in which the second inequality follows
from the concavity of the logarithm,
and the third one follows from
the fact that for all $n\geq r\geq 1$,
\[ 2r a_{n,r}^{-1} \leq 2r \big(
(r+1)^4 \big(\frac{ne}{r}\big)^{2r}r^{r}\big)^{-1}
= \frac{2}{(r+1)^3} \frac{r}{r+1}
\frac{\exp(-2r)}{n^r} \big(\frac{r}{n}\big)^r 
\leq \exp(-2)/4\enspace,
\]
together with $2^{-3/2}\exp(\exp(-2)/4)\leq 1/2$.
 Since the value of an entropy game coincides
  with a non-zero eigenvalue of the ambiguity matrix associated
  with a pair of strategies, the result is established.
\end{proof}
\begin{remark}
  When $W>2$, the separation bound of $\nu_{n,r}^{-1}$
  may be written as $C_{n,r}^{-1} W^{-4r^2}$.
  The exponent of $W$
  may be improved at the price of increasing the combinatorial
  factor $C_{n,r}$, by using a theorem of Mahler~\cite{mahler}, showing that
  for a polynomial of degree $d$, with integer coefficients bounded by $M$, and
  without multiple roots, the inverse of the distance
  between two distinct roots is bounded by
  \(
  \vartheta_{d,M} = 3^{-1/2} (d+1)^{ (2d+1)/2} M^{d-1} \).
  As explained in~\cite{rump}, one can deduce from this
  a separation root for polynomials $P$ with possibly multiple roots,
  after replacing $P$ by $\tilde{P}:=P/\operatorname{gcd}(P,P')$. 
  Using Malher's bound instead of~\cref{e-rump-ineq}, we would arrive
  at a bound $(C'_{n,r})^{-1} W^{-4r^2+2r}$ instead of $C_{n,r}^{-1} W^{-4r^2}$.
\end{remark}

We next deduce from \cref{th-sepentropy}  a separation bound for values
of strategies.
\begin{lemma}\label{lemma_entropy_value}
The values of an entropy game lie in the interval $[1,nW]$.
\end{lemma}
\begin{proof}
By \eqref{eq:entropy_operator}, we have $\mshapley(\unitvector) \le n\Weights \unitvector$, and so $\mshapley^{N}(\unitvector) \le (n\Weights)^{N}\unitvector$ for all $N$. Then, by \cref{th-valexists}, we conclude that $\mgameval \le n\Weights\unitvector$. Likewise, since entropy games satisfy \cref{as:entropy_nondeg} and the multiplicities $\eweight_{pl}$ are natural numbers, we have $\mshapley(\unitvector) \ge \unitvector$, and so $\mgameval \ge \unitvector$. 
\end{proof}

\begin{corollary}\label{coro-additive}
Suppose two pairs of strategies yield distinct values in an entropy game of rank $r$, with $n$ Despot's states. Then, the logarithms of these values differ at least by $\hat{\nu}_{n,r}^{-1}$ where
  \(
  \hat{\nu}_{n,r}\coloneqq n\Weights \nu_{n,r}\).
\end{corollary}
\begin{proof}
Let $\lambda,\mu$ denote the distinct values determined by two pairs of strategies, and suppose, without loss of generality, that $\lambda>\mu$. %
Then, \cref{lemma_entropy_value} implies $1\leq \mu < \lambda \leq n\Weights$.\todo{MS: Maybe explain this inequality? It is proven in the proof of \cref{pr:cw_in_entropy}, but this arrives later on. SG. I added a proof. RK: I modified the proof following the addition of \cref{lemma_entropy_value}, I hope this is fine.}  By the intermediate value theorem, we conclude that $\log \lambda -\log \mu = (\lambda-\mu)/\xi$ for
  some $\xi \in ]1,n\Weights[$, and so $\log \lambda -\log \mu \geq (\lambda-\mu)/(n\Weights) \geq \hat{\nu}_{n,r}^{-1}$ by \cref{th-sepentropy}. 
\end{proof}

\todo{MS: How explicit constant do we want to have?}
\begin{proposition}\label{pr:cw_in_entropy}
Let $0 < \delta < 1$. Then, there exist vectors $w,z \in \R_{>0}^{\vertexsetD}$ such that $e^{-\delta}\mybot(\mgameval) w \le \mshapley(w)$, $e^{\delta}\mytop(\mgameval)z \ge \mshapley(z)$, and $R_{n,r}\coloneqq 1200(n^3\log\Weights + n^2 \log \delta^{-1}) \geq \max\{\HNorm{\log w},\HNorm{\log z}\} $.
\end{proposition}\todo[color=red!30]{RK: Since I could not find where the parameter $R_{n,r}$ was defined, I added a definition in \cref{pr:cw_in_entropy}.}
This proposition will be established by observing that for a given value of $\delta$,
$w$ and $z$ are defined by semi-linear constraints, and by using
bitlength estimates on the generators and vertices of polyhedra
defined by inequalities. %
To do so, 
we first state, for convenience, the following consequences
of \cref{Theo:UCW}, which follow by applying this theorem
to the Shapley operator $F=\log\circ T\circ \exp$,
\begin{align}
  \mybot({\mgameval}) &= \sup \{\lambda>0\colon \exists w\in \R_{>0}^{n}, \;\lambda w\leq T(w)\} \; ,\label{cw-mult}
  \\
  \mytop({\mgameval}) &= \inf \{\mu>0\colon \exists z\in \R_{>0}^{n}, \;\mu z\geq T(z)\}
\; .
\end{align}

We also recall that the \emph{encoding length} of an integer number $r$ is defined as $\enc{r} \coloneqq \lceil \log_2 (\abs{r} + 1) \rceil + 1$. Moreover, if $r = p/q$ is a rational number then its encoding length is defined as $\enc{r} \coloneqq \enc{p} + \enc{q}$. The encoding length of an affine inequality with rational coefficients $ax \le b$ is defined as $\enc{b} + \enc{a_1} + \dots + \enc{a_n}$. In the proof, we use the following 
inequalities, which are stated in~\cite[\S~1.3]{grotschel_lovasz_schrijver_geometric_algorithms}:

\begin{lemma}\label{le:enc_ineqs}
If $r \ge 1$ is a natural number, then $\log_2(r) \le \enc{r} \le 3 + \log_2(r)$. If $r_1, \dots, r_m$ are rational numbers, then $\enc{r_1 + \dots + r_m} \le 2(\enc{r_1} + \dots + \enc{r_m})$ and $\enc{r_1 r_2} \le \enc{r_1} + \enc{r_2}$.\hfill 
\end{lemma}

\begin{proof}[Proof of \cref{pr:cw_in_entropy}.]
By \cref{lemma_entropy_value}, we have $\mybot(\mgameval) \ge 1$, so the open interval $]e^{-\delta}\mybot(\mgameval),\mybot(\mgameval)[$ is of length at least $1 - e^{-\delta} \ge \frac{\delta}{1 + \delta}$. Thus, it contains a rational number $q \in \Q$ with denominator at most $\frac{2(1+\delta)}{\delta} = 2 + \frac{2}{\delta} \le \frac{4}{\delta}$. Since $\mybot(\mgameval) \le n\Weights$ by \cref{lemma_entropy_value}, the numerator of $q$ is not greater than $\frac{4n\Weights}{\delta}$, and so the encoding length of $q$ satisfies 
\[
\enc{q} \le 10 + \log_2 n + \log_2 \Weights + 2\log_2 \delta^{-1} \, .
\]
Let $w \in \R_{>0}^{n}$ be any vector that satisfies $q w \le \mshapley(w)$.
Since $q<\mybot(\mgameval)$, the existence of such a vector
follows from~\cref{cw-mult}.
For every $t\in\vertexsetT$, let $p_t$ be such that $(t,p_t) \in \edges$ and $\max_{(t,p) \in \edges} \sum_{(p,l) \in \edges} \eweight_{pl} w_{l} = \sum_{(p_t,l) \in \edges} \eweight_{p_tl} w_{l}$. Likewise, for $k\in\vertexsetD $ let $t_k$ be such that $(k,t_k) \in \edges$ and
$$\min_{(k,t) \in \edges} \max_{(t,p) \in \edges} \sum_{(p,l) \in \edges} \eweight_{pl} w_{l} = \max_{(t_k,p) \in \edges} \sum_{(p,l) \in \edges} \eweight_{pl} w_{l}\enspace.$$
In this way, we have
\(
\mshapley_k(w) =  \sum_{(p_{(t_k)},l) \in \edges} \eweight_{p_{(t_k)}l} w_{l} \),
for all $k$. Consider now the polyhedron $\polyh \subset \R^{n}$ defined by the inequalities
\begin{equation}\label{le:cw_polyhedron}
\begin{aligned}
\forall (t,p) \in \edges \cap (\vertexsetT\times \vertexsetP),  \, \sum_{(p,l) \in \edges} \eweight_{pl} x_{l} &\le \sum_{(p_t,l) \in \edges} \eweight_{p_tl} x_{l} \\
\forall (k,t) \in \edges\cap (\vertexsetD\times \vertexsetT),  \, \sum_{(p_t,l) \in \edges} \eweight_{p_tl} x_{l} &\ge \sum_{(p_{(t_k)},l) \in \edges} \eweight_{p_{(t_k)}l} x_{l} \\
\forall k\in\vertexsetD,  \,\qquad 0\leq qx_k &\le  \sum_{(p_{(t_k)},l) \in \edges} \eweight_{p_{(t_k)}l} x_{l}\enspace . %
\end{aligned}
\end{equation}
Then, we have $w \in \polyh \cap \R^{n}_{>0}$ so $\polyh$ is nonempty. Even more, any vector $x \in \polyh \cap \R^{n}_{>0}$ satisfies $q x \le \mshapley(x)$. Let 
\[
\Enc 
\coloneqq 1 + 4n(3 + \log_2 \Weights) + 2\langle q \rangle \le 50(n \log_2 \Weights + \log_2 \delta^{-1}) \, .
\]
Then, \cref{le:enc_ineqs} shows that every inequality that describes $\polyh$ has encoding length at most $\Enc$. 
Hence, by \cite[Lemma~6.2.4]{grotschel_lovasz_schrijver_geometric_algorithms}, there exists two finite sets $X,Y \subset \Q^{n}$ such that $\polyh = \conv(X) + \cone(Y)$ and every entry of every vector in $X \cup Y$ has encoding length at most $4n\Enc$. 
Furthermore, we have $X \subset \Q^{n}_{\ge 0}$ because $X \subset \polyh$ and $Y \subset \Q^{n}_{\ge 0}$ because $x + \lambda y \in \polyh$ for any $(x,y) \in X \times Y$ and $\lambda \ge 0$. By Carath{\'e}odory's theorem~\cite[Section~7.7]{schrijver}, there exist $x_0, \dots x_n \in X$, $y_1,\dots,y_n \in Y$ such that $w \in \conv(x_0, \dots, x_n) + \cone(y_1, \dots, y_n)$. Since $w \in \R^{n}_{>0}$ and $X,Y \subset \Q^{n}_{\ge 0}$, the point $\tilde{w} \in \polyh$ defined as
\[
\tilde{w} \coloneqq \frac{1}{n+1}x_0 + \dots + \frac{1}{n+1}x_{n+1} + y_1 + \dots + y_n  
\]
also satisfies $\tilde{w} \in \R_{>0}^{n}$. Moreover, the inequalities from \cref{le:enc_ineqs} show that every entry of $\tilde{w}$ has encoding length at most $60n^2 \Enc$. 
Therefore, for all $k$ we can write $\tilde{w}_k = p_k/q_k$, where the numbers $p_k,q_k$ are natural and their encoding length is bounded by $60n^2\Enc$. 
In particular, we have $\HNorm{\log_2 \tilde{w}} \le 2\|\log_2 \tilde{w}\|_{\infty} \le 2\max_k (\log_2 p_k + \log_2 q_k)\le 240n^2\Enc \le 1200(n^3 \log_2 \Weights + n^2\log_2 \delta^{-1})$
 and $\tilde{w}$ satisfies the claim. The proof of the other part is analogous, using the fact that the open interval $]\mytop(\mgameval),e^{\delta}\mytop(\mgameval)[$ has length at least $e^{\delta} - 1 \ge \delta$, and so it contains a rational number with denominator at most $2/\delta$. 
\end{proof}

We deduce the following parameterized complexity result.
\begin{theorem}\label{th-entropy}
In an entropy game of rank at most $r$, we can find
the set of initial states with maximal value
by performing $O(nR_{n,r}\hat{\nu}_{n,r})$ calls to an oracle
approximating $F$ with precision $\delta/8$
where $\delta=(\hat{\nu}_{n,r})^{-1}$.
\end{theorem}
\begin{proof}%
The claim follows by combining \cref{topclass_oracle} with the estimates from \cref{coro-additive} and \cref{pr:cw_in_entropy}. 
  \end{proof}
  \begin{remark}\label{rk:borwein}
    We note that the oracle used in \cref{th-entropy}, approximating the Shapley operator of an entropy game
    up to a given precision, can be implemented in polynomial time, this follows by using
a result of Borwein and Borwein~\cite{borwein}
on the approximation of the $\log$ and $\exp$ maps,
together with a scaling argument, see~\cite[Lemma~27]{entropygamejournal}.
  \end{remark}
We now solve strategically the special case of entropy games whose value
is independent of the initial state. To do so, we suppose that we have access to the graph of the game as well as to an oracle that approximates the operator $\shapley$.

  \begin{proposition}\label{prop-ent-cstvalue}
    In an entropy game of rank at most $r$, such that the value is independent of the initial state,
    we can find a pair of optimal strategies for both players by performing
    $O(nR_{n,r}\hat{\nu}_{n,r})$ calls to an oracle
approximating $F$ with precision $\delta/16$
where $\delta=(\hat{\nu}_{n,r})^{-1}$.
  \end{proposition}
  \begin{proof}
  The proof is similar to the proof of \cref{cor-perturb}. We use  \textsc{ApproximateConstantMeanPayoff} to obtain two vectors $x,y \in \R^{\vertexsetD}$
  such that $a + x \le \shapley(x)$ and $b + y \ge \shapley(y)$,
and $[a,b]$ is an interval of width at most $\delta/2$ containing $\ec(F)$.
By \cref{pr:approx_constant_value} combined with the estimates from \cref{coro-additive} and \cref{pr:cw_in_entropy}, this requires $O(nR_{n,r}\hat{\nu}_{n,r})$ calls to an oracle approximating $F$.
  For every $k \in \vertexsetD$ we have
 \[
 b + y_k \ge \min_{(k,t) \in \edges}\max_{(t,p) \in \edges} \log\Bigl(\sum_{(p,l) \in \edges} \eweight_{pl}e^{y_l}\Bigr) \, .
 \]
 For every $p \in \vertexsetP$ we approximate the expression $\log\Bigl(\sum_{(p,l) \in \edges} e^{y_l}\Bigr)$ to precision $\delta/8$ using the procedure from~\cite[Lemma~27]{entropygamejournal}. Let $Q_p \in \R$ denote this approximation and let $\sigma$ be a strategy of Despot that satisfies 
 \(
 \min_{(k,t) \in \edges}\max_{(t,p) \in \edges} Q_p = \max_{(\sigma(k),p) \in \edges}Q_p
 \)
 for all $k \in \vertexsetD$. We have
 \begin{align*}
  b + y_k &\ge \min_{(k,t) \in \edges}\max_{(t,p) \in \edges} \log\Bigl(\sum_{(p,l) \in \edges} \eweight_{pl}e^{y_l}\Bigr) \ge -\delta/8 + \min_{(k,t) \in \edges}\max_{(t,p) \in \edges} Q_p \\
  &= -\delta/8 + \max_{(\sigma(k),p) \in \edges}Q_p \ge -\delta/4 + \max_{(\sigma(k),p) \in \edges}\log\Bigl(\sum_{(p,l) \in \edges} \eweight_{pl}e^{y_l}\Bigr) \, .
 \end{align*}
 Hence, if we denote by $\shapley^{\sigma}$ the Shapley operator obtained by fixing $\sigma$, then $b + \delta/4  + y \ge \shapley^{\sigma}(y)$. In particular, $\ucw(\shapley^{\sigma}) \le b + \delta/4$. Since the interval $[a, b +\delta/4]$ is of length smaller than $\delta$ and contains $\ec(\shapley)$, \cref{coro-additive} implies that $\ucw(\shapley^{\sigma}) \le \ec(\shapley)$ and $\sigma$ is optimal. We can construct an optimal strategy of Tribune in an analogous way. 
    \end{proof}

The following decomposition property for entropy games
extends a classical property of deterministic
mean-payoff games. Once the set of Despot's states
with maximal value is known, it allows one to
determine the value of the other states
by reduction to an entropy game induced
by the other states of Despot. %
\begin{lemma}[Decomposition property]\label{le:decomposition}
Let $\set_1 \coloneqq \domMaxVal \dunion \vertexsubsetit[T]{\domMaxVal} \dunion \vertexsubsetit[P]{\domMaxVal}$ and $\set_2 \coloneqq \vertexset \setminus \set_1$. Furthermore, suppose that $\set_2$ is nonempty. Consider the induced digraphs $\dgraph[\set_1]$ and $\dgraph[\set_2]$ of the original graph $\dgraph = (\vertexset, \edges)$. Then, the entropy games arising by restricting the graph to $\dgraph[\set_1]$ and $\dgraph[\set_2]$ satisfy \cref{as:entropy_nondeg}. Furthermore, if $(\sigma_1,\tau_1)$ are optimal strategies of Despot and Tribune in the induced entropy game on $\dgraph[\set_1]$ and $(\sigma_2,\tau_2)$ are optimal strategies of Despot and Tribune in the induced entropy game on $\dgraph[\set_2]$, then the joint strategies
\begin{equation}\label{eq:decomposition}
\forall k \in \vertexsetD, \ \sigma(k) = \begin{cases}
\sigma_1(k) &\text{if $k \in \domMaxVal$,} \\
\sigma_2(k) &\text{otherwise,}
\end{cases}\\
\; \;  \forall t \in \vertexsetT, \ \tau(t) = \begin{cases}
\tau_1(t) &\text{if $t \in \vertexsubsetit[T]{\domMaxVal}$,} \\
\tau_2(t) &\text{otherwise},
\end{cases}
\end{equation}
are optimal in the original game.
\end{lemma}
\begin{proof}%
The game on $\dgraph[\set_1]$ satisfies \cref{as:entropy_nondeg} by \cref{le:good_entropy,le:entropy_domin_nondeg}. The other game satisfies this assumption by the definition of the set $\set_2$. Indeed, if $p \in \vertexsetP \cap \set_2$, then all the outgoing edges of $p$ go to $\vertexsetD \setminus \domMaxVal$, so they are in $\dgraph[\set_2]$. Likewise, if $t \in \vertexsetT \cap \set_2$, then all the outgoing edges of $t$ go to $\vertexsetP \setminus \vertexsubsetit[P]{\domMaxVal}$, so they are in $\dgraph[\set_2]$. If $k \in \vertexsetD \cap \set_2$, then \cref{le:good_entropy_aux} shows that $k$ has an outgoing edge that goes to $\vertexsetT \setminus \vertexsubsetit[T]{\domMaxVal}$, so this edge is in $\dgraph[\set_2]$.

To prove the second part of the claim, let $\hat{\sigma}, \hat{\tau}$ be any pair of optimal strategies of Despot and Tribune in the original game and let $\gameval$ be the escape rate of the operator $\shapley$. Let $\gameval^{\sigma,\tau}$ be the escape rate of the operator obtained by fixing $(\sigma,\tau)$. We use analogous notation for other pairs of strategies. By \cref{le:good_entropy_aux} for every $k \in \vertexsetD \setminus \domMaxVal$ we have $\hat{\sigma}(k) \notin \vertexsubsetit[T]{\domMaxVal}$. Hence, \cref{prop-rothblum} combined with the optimality of $\hat{\sigma},\tau_2,\sigma_2,\hat{\tau}$ gives the inequality $\gameval_k \ge \gameval^{\hat{\sigma},\tau}_k \ge \gameval^{\sigma,\tau}_k \ge \gameval^{\sigma,\hat{\tau}}_k \ge \gameval_k$. Therefore, $\gameval^{\sigma,\tau}_k = \gameval_k$ for all such $k$.

Let $\bar{\sigma}$ be an optimal response to $\tau$, i.e., an optimal strategy of Despot in the game in which Tribune plays according to $\tau$. To prove the optimality of $\tau$ it is enough to show that $\gameval^{\bar{\sigma},\tau} \ge \gameval$. Consider the game obtained by fixing $(\bar{\sigma},\tau)$. Note that if we remove from the graph $\dgraph^{\bar{\sigma},\tau}$ the edges that go from $\vertexsubsetit[P]{\domMaxVal}$ to $\vertexsetD \setminus \domMaxVal$, then the value of this game can only decrease.
Moreover, by \cref{le:entropy_domin,le:good_entropy} we have $\bar{\sigma}(k) \in \vertexsubsetit[T]{\domMaxVal}$ for all $k \in \domMaxVal$ and $\lcw(\shapley^{\domMaxVal}) = \ucw(\shapley^{\domMaxVal}) = \mytop(\gameval)$. Hence, the optimality of $\tau_1$ combined with \cref{prop-rothblum} give $\gameval^{\bar{\sigma},\tau}_k \ge \mytop(\gameval) = \gameval_k$ for all $k \in \domMaxVal$. Furthermore, we have $\bar{\sigma}(k) \notin \vertexsubsetit[T]{\domMaxVal}$ for all $k \in \vertexsetD \setminus \domMaxVal$. Indeed, suppose that $\bar{\sigma}(k) \in \vertexsubsetit[T]{\domMaxVal}$. Then, there is a dipath in $\dgraph^{\bar{\sigma},\tau}$ that goes from $k$ to $\domMaxVal$ and so $\gameval^{\bar{\sigma},\tau}_k \ge \mytop(\gameval)$ by the previous observation. However, $\bar{\sigma}$ is an optimal response to $\tau$ and therefore we have $\gameval^{\bar{\sigma},\tau}_k \le \gameval_k < \mytop(\gameval)$, which gives a contradiction. Thus,  $\bar{\sigma}(k) \notin \vertexsubsetit[T]{\domMaxVal}$ for all $k \in \vertexsetD \setminus \domMaxVal$. In particular, the optimality of $\tau_2$ gives $\gameval^{\bar{\sigma},\tau}_k \ge \gameval^{\sigma,\tau}_k = \gameval_k$.

Analogously, let $\bar{\tau}$ be an optimal response to $\sigma$ and consider the game obtained by fixing $(\sigma, \bar{\tau})$. The optimality of $\sigma_2$ gives $\gameval^{\sigma,\bar{\tau}}_k \le \gameval^{\sigma,\tau}_k = \gameval_k$ for all $k \in \vertexsetD \setminus \domMaxVal$. Fix $k \in \domMaxVal$ and let $\mathscr{C}$ be any strongly connected component of $\dgraph^{\sigma,\bar{\tau}}$ that can be reached form $k$ in this graph. If $\mathscr{C}$ contains a vertex from $\vertexsetD \setminus \domMaxVal$, then the previous observation combined with \cref{prop-rothblum} gives $\specrad(M^{\sigma,\bar{\tau}}[\mathscr{C}]) < \mytop(\gameval)$. If $\mathscr{C}$ does not contain any such vertex, then it is included in $\set_1$. Hence, the optimality of $\sigma_1$ combined with \cref{prop-rothblum} and the equality $\lcw(\shapley^{\domMaxVal}) = \ucw(\shapley^{\domMaxVal}) = \mytop(\gameval)$ give $\specrad(M^{\sigma,\bar{\tau}}[\mathscr{C}]) \le \mytop(\gameval)$. Hence, $\gameval^{\sigma,\bar{\tau}}_k \le \mytop(\gameval) = \gameval_k$ for all $k \in \domMaxVal$ and $\sigma$ is optimal. 
\end{proof}

Then, by combining \cref{th-entropy} and \cref{le:decomposition}, we get:
\begin{theorem}\label{cor-fp}
  A pair of optimal strategies of an entropy game of rank $r$
  can be found in $O(n^2R_{n,r}\hat{\nu}_{n,r})$ calls to an oracle
  that returns $\shapley$ with a precision of $1/(16\hat{\nu}_{n,r})$.
  Then, entropy games in the original model of Asarin et al.~\cite{asarin_entropy} and with a fixed rank are polynomial-time solvable, whereas entropy games with weights, in the model of Akian et al.~\cite{entropygamejournal}, and with a fixed rank, are pseudo-polynomial time solvable.
\end{theorem}
\begin{proof}%
  We solve entropy games by the following algorithm. We find the states with maximal value using \TopClass{}, find optimal strategies on this dominion using \cref{prop-ent-cstvalue}, we split the game in two as in the decomposition property (\cref{le:decomposition}) and continue recursively on the smaller game. The bound on the number of calls to the oracle follows
from \cref{th-entropy}, \cref{le:decomposition,prop-ent-cstvalue}. Since the approximating oracle can be constructed in polynomial time as discussed in \cref{rk:borwein}, this algorithm can be implemented to work in polynomial memory. Furthermore, as in \cref{rk-compar-gurvich}, the time complexity of this algorithm is dominated by the number of calls to the oracle multiplied by the time cost of a single call. 
  \end{proof}

\begin{corollary}\label{cor-constpeople}
  Entropy games with weights and with a fixed number of People's positions are pseudo-polynomial
  time solvable.
\end{corollary}
\begin{proof}%
  The claim follows from \cref{cor-fp} by noting that the rank of an entropy game is bounded by the number People's states. 
  \end{proof}

\subsection{A Lower Complexity Bound for Value Iteration for Entropy Games}
  We say that a state of a game is {\em significant} if there are several
  options in this state, in particular, a state $p$ of People is significant
  if there are at least two distinct arcs $(p,k)$ and $(p,l)$ in $\edges$. We say that an entropy game is {\em Despot-free} (resp. {\em Tribune-free}) when Despot (resp. Tribune) has no significant states.\todo{RK: I added the last definition.} 
  
  We may ask whether the statement of \cref{cor-constpeople} carries
  over to entropy games with a fixed number of significant People's states.
  The following result shows that this can not be derived
  from the universal value iteration bounds, since value
  iteration needs $\Omega(\Weights^{n-1})$ iterations
 to recognize the optimal strategy. 

\begin{theorem}\label{th-cex}
  There is a family of Despot-free entropy games $G_n(\Weights)$,
  and a constant $C>0$, with the following
  properties:
  \begin{enumerate}
\item $G_n(\Weights)$
  has arc weights $\leq \Weights$, only one significant Tribune's state,
  with two actions,
and $2n+1$ People's states among which
there are only $4$ significant states;
\item The action of Tribune that is optimal
  in the mean-payoff entropy game is never played,
  if Tribune plays optimally
  in the entropy game of finite horizon $k$, for all $k\leq C\Weights^{n-1}$.\todo{MS: The index ``$k$'' is used above for states of Despot. SG. Not easy to find an unused letter. Perhaps leave it as it like this (or use $h$?). RK: In \cref{sec-preliminaries} we used $N$ for the horizon}
  \end{enumerate}
\end{theorem}
To construct this game, we need an estimate of
the positive root of a special polynomial $p_n$.
\begin{proposition}\label{prop-puiseux}
  Consider the polynomial $p_n(x)=x^n-\Weights(x^{n-1}+\dots +x + 1)$, where $\Weights \geq 1$.
  Then, for all $n\geq 1$, $p_n$ has a unique positive
  root, $x_n(\Weights)$. Moreover, $x_{n-1}(\Weights)<x_n(\Weights)<W+1$ and
  \[
  x_n(\Weights)=\Weights+1-1/\Weights^{n-1}+ o(1/\Weights^{n-1}) \enspace ,
  \qquad \text{as }\Weights\to \infty \enspace .
  \]
\end{proposition}
\begin{proof}%
  Recall that Descartes' rule of sign states
  that the number of positive roots of a polynomial
  is bounded by the number of variations of signs
  of the sequence of its coefficients, and that
  it is equal modulo $2$ to this number.
  It follows that for all $n\geq 1$, and
  for all $\Weights\geq 1$,
  $p_n$ has a unique positive root, $x_n(W)$;
moreover, $p_n(W+1)=1>0$, which 
  entails that $x_n(W)<W+1$. Furthermore,
  $p_n(x_{n-1}(W))=x_{n-1}(W)^{n-1}(x_{n-1}(W)-W-1)<0$
  entails that $x_{n}(W)>x_{n-1}(W)$.
Let us define the Newton polygon $\Delta$ as
  the upper boundary of the convex hull of the points $(i,j)$ such that
  $p_n$ has a monomial of degree $i$ in $x$ and of degree $j$ in $\Weights$.
  Then, the Newton--Puiseux theorem~\cite{walker_algebraic_curves}
  shows that all the roots of $p_n$ have Puiseux series expansions,
  with a leading term $a\Weights^\alpha + o(\Weights^\alpha)$ where
  $\alpha$ is the opposite of the slope of an edge
  of the Newton polygon $\Delta$ of $p_n$, and $a \in \mathbb{C}\setminus \{0\}$. Moreover, the horizontal width of an edge determines
  the number of roots of order $\Weights^\alpha$, counted with multiplicities.
  Here, the monomials $-\Weights,\dots,-\Weights x^{n-1}$ determine the edge $[(0,1),(n-1,1)]$,
which is of slope $0$, and has horizontal width $n-1$.
  The monomials $-\Weights x^{n-1}$ and $x^n$ determine the edge
  $[(n-1,1),(n,0)]$, which is of slope $-1$ and has horizontal width $1$. 
So there are $n-1$ roots $x_i(\Weights)= a_i \Weights^0+o(\Weights^0)=a_i+o(1)$ as $\Weights\to \infty$,
with $1\leq i\leq n-1$, and there is one root $x_n(\Weights)=a_n \Weights^1 + o(\Weights)$.
Substituting $x_i(\Weights)$ in $p_n(x_i(\Weights))=0$, for $1\leq i\leq n-1$,
we find that $a_i$ must be a root of $1+\dots+x^{n-1}$.
Thus, $a_i$ cannot be positive for $i<n$. So, the unique positive root of $p_n$
is $x_n(\Weights)=a_n\Weights +o(\Weights)$. Substituting
  $x_n(\Weights)$ in $p_n(x_n(\Weights))=0$, and cancelling negligible terms,
  we get $a_n=1$. Hence, we can write $x_n(\Weights)=\Weights + y(\Weights)$, where
  $y(\Weights)=o(\Weights)$. 
   Moreover, using
 $(x^{n-1}+\dots + 1)=(x^n-1)/(x-1)$,
 $p_n(x)=0$ can be rewritten
 as $x^n(x-\Weights-1)+\Weights=0$. Substituting $x_n(\Weights)=\Weights+y(\Weights)$,
 we end up with $y(\Weights)=1- \Weights/x_n(\Weights)^n =1 -1/\Weights^{n-1}+ o(1/\Weights^n)$,
 hence,
 $x_n(\Weights)=1+\Weights-1/\Weights^{n-1}+o(1/\Weights^{n-1})$. 
\end{proof}

We denote by
\[ A_n(\Weights)
= \left(\begin{array}{ccccc}\Weights & \dots & \dots & &\Weights \\ 1 & 0 & \dots & \dots &0\\
  0& \ddots & \ddots & &\vdots\\
\vdots & \ddots & \ddots & \ddots&\vdots \\
0   &\dots & 0 & 1 &0
  \end{array}
\right)
\]
the $n\times n$ companion matrix of the polynomial
$p_n$ considered in \cref{prop-puiseux}. 
\begin{lemma}\label{lemma-perroneig}
  The left Perron eigenvector $u^n$ of the matrix $A_n(\Weights)$
  satisfies
  \[
  \frac{1}{2} u_1^n \leq u_i^n \leq u_1^n \enspace \text{for all }2\leq i\leq n
  \enspace .
  \]
\end{lemma}
\begin{proof}
Noting that $p_n(\Weights+1)>0>p_n(\Weights)$,
and using the intermediate value theorem,
we deduce that the unique positive root $x_n(\Weights)$ of $p_n$ satisfies
$\Weights<x_n(\Weights)<\Weights+1$. Since $A_n(W)$ is a companion matrix,
its characteristic polynomial is precisely $p_n$, and so,
the Perron root $\lambda$ of $A_n(\Weights)$ coincides with $x_n(\Weights)$.
Since $A_n(W)$ is irreducible, it has a (left) positive eigenvector $u$, the Perron eigenvector, which is unique up to a multiplicative constant.
  We have $\lambda u_1=\Weights u_1+u_2$,
  $\lambda u_2=\Weights u_1+u_3$,\dots,
  $\lambda u_{n-1}=\Weights u_1+u_n$,
  $\lambda u_n = \Weights u_1$.
  Hence,
  $u_i\geq u_1\Weights/\lambda\geq u_1 \Weights/(\Weights+1)\geq u_1/2$,
  for all $2\leq i\leq n$.
  From $\lambda u_1 = \Weights u_1+u_2$, we deduce
  that $u_2 = (\lambda -\Weights)u_1 \leq u_1$.
  Then, from $\lambda u_2 = \Weights u_1 +u_3$,
  we deduce that $u_3 = \lambda u_2 -\Weights u_1 \leq \lambda u_1 -\Weights u_1 \leq u_1$.
  Continuing in this way, we get that $u_i \leq u_1$ for all $2\leq i\leq n$. 
\end{proof}

The proof of \cref{th-cex} also relies on the following lemma.
We set
$\lambda_i\coloneqq x_i(\Weights)$ for all $i\geq 1$.
\begin{lemma}
  Let $\alpha>1$. Then, in the expression
  \[
z(k) = \max( \unitvector^\top_n A_n^k \unitvector_n, \alpha \unitvector^\top_{n-1}  A^k_{n-1}\unitvector_{n-1} )
\enspace,
\]
the maximum is achieved by the rightmost term, 
for all $k$ such that
\begin{align}
  k \leq  k^*\coloneqq \frac{ \log(\alpha (n-1)/(4n))}{\log(\lambda_n/\lambda_{n-1})} \enspace.
\label{e-stop}
\end{align}
\end{lemma}
\begin{proof}
  The Perron eigenvector $u^n$ is defined only up
  to a positive multiplicative constant, so 
  by \cref{lemma-perroneig}, we may assume
  that $\unitvector_n^\top  \leq u^n \leq 2 \unitvector_n^\top $.
  Hence,
  \[\frac{n}{2}\lambda _n^k \leq  \frac{1}{2} \lambda_n^k u^n\unitvector_n=\frac{1}{2} u^n A_n^k \unitvector_n \leq \unitvector^\top_n A_n^k\unitvector_n \leq u^n A_n^k \unitvector_n = \lambda_n^k u^n\unitvector_n \leq 2n \lambda_n ^k
  \enspace .
  \]
  Hence, $\unitvector^\top_n A_n^k \unitvector_n \leq \alpha \unitvector^\top_{n-1}  A_{n-1}^k \unitvector_{n-1} $
  holds as soon as
  \( 2n\lambda_n^k \leq \frac{n-1}{2} \alpha \lambda_{n-1}^k \),
which is the case when~\cref{e-stop} holds. 
\end{proof}

\begin{figure}
\centering
\includegraphics{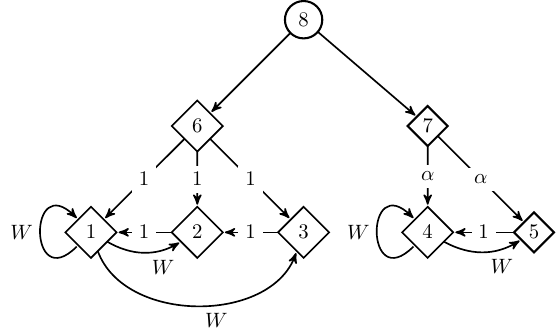}
\caption{The entropy game constructed in the proof of \cref{th-cex}, here for $n=3$. The diamonds represent People's states, whereas the circle represents the unique significant Tribune's state. The weights are indicated on the arcs. For $\alpha\gg 1$, when the horizon is small, it is optimal for Tribune to play ``bottom right'', whereas when the horizon is large, the optimal move of Tribune is ``bottom left''. \label{fig-entropy}}
\end{figure}

\begin{proof}[Proof of \cref{th-cex}.]
  We set $\alpha\coloneqq 12$, so that $\log(\alpha (n-1)/(4n))\geq \log 2$ for all $n\geq 3$.
Using the asymptotics of $\lambda_n=x_n(\Weights)$ given by \cref{prop-puiseux},
we find that
\begin{align}
   k^*=  (\log 2)\Weights^{n-1} +o(\Weights^{n-1}) \, .\label{e-stop2-appendix}
\end{align}
We now claim that $z(k)$ can be interpreted as the value in horizon
$k+1$ of an entropy game satisfying the conditions
of the theorem. \Cref{fig-entropy} illustrates
the proof when $n=3$.

  First, we note that the term $\unitvector^\top_n A_n^k \unitvector_n$ can be interpreted
  as the value in horizon $k+1$ of a Despot-free and Tribune-free
  entropy game, with $n+1$ People's states,
  among which there are only two significant states,
  one encoding the first row of $A_n$, and another one
  encoding the row vector $\unitvector^\top_n$. Recall every turn
  of an entropy game involves a succession of three stages,
  with moves made by Despot, Tribune, and People.
  So, this interpretation
  requires to insert dummy states of Despot and Tribune for each
  transition between People states. E.g., in the figure, in which dummy states
  are not represented,
  the value of the entropy game in horizon $k+1$ with initial state
  $6$ is precisely $ \unitvector^\top_3 A_3^k \unitvector_3$.

  The term $\alpha \unitvector^\top_{n-1} A_{n-1}^k \unitvector_{n-1}$ admits
  a similar interpretation, with $n$ People's states
  instead of $n+1$. One significant state encodes
  the first row of $A_{n-1}$, whereas the other
  significant state encodes the row vector $\alpha \unitvector^\top_{n-1}$.
  E.g., in the figure,
  the value of the entropy game in horizon $k+1$ with initial state
  $7$ is precisely $\alpha \unitvector^\top_2 A_2^k \unitvector_2$.\todo[color=red!30]{RK: I added $\alpha$ here.}

  We complete the construction of the entropy game $G_n(\Weights)$ by adding a 
  significant state of Tribune, with only two options:
  moving to the state encoding $\unitvector^\top_n$,
  or moving to the state encoding $\alpha \unitvector^\top_{n-1}$.
In the figure, this significant state of Tribune is labeled $8$.

Then, using the dynamic programming equation~\cref{e-dp-entropygame},
we see that the value of the corresponding entropy game in horizon $k$,
starting from the significant state of Tribune,
is precisely $z(k-1)$. Since $\lambda_{n}>\lambda_{n-1}$
(by~\Cref{prop-puiseux}),
  in the mean-payoff entropy game, the optimal action
  for Tribune is to move to the state encoding $\unitvector^\top_n$
  (move ``bottom left'' in the figure)
  which guarantees a geometric growth of $\lambda_n$. However,
  for $k\leq k^*+1$, the optimal action, for the initial move,
  is to select the term achieving the maximum in the expression
  of $z(k)$, and so, to move to the state encoding
  $\alpha \unitvector^\top_{n-1}$ (move ``bottom right''). 
\end{proof}

\section{Concluding Remarks}
We developed generic value iteration algorithms,
which apply to various classes of zero-sum games with mean
payoffs. These algorithms admit universal complexity
bounds, in an approximate oracle model -- we only need
an oracle evaluating approximately the Shapley operator.
These bounds involve
three fundamental ingredients: the number of states,
a separation bound between the values induced by different strategies,
and a bound on the norms of Collatz--Wielandt vectors.
We showed that entropy games with a fixed rank
(and in particular, entropy games with a fixed number
of People's states) are pseudo-polynomial time
solvable.
This should be compared with the result
of~\cite{entropygamejournal}, showing that entropy games
with a fixed number of Despot positions are polynomial-time
solvable. Since fixing the number of states of Despot or People
leads to improved complexity bounds, one may ask
whether entropy games with a fixed number of significant Tribune states
are polynomial or at least pseudo-polynomial, this is still an open question.

\section*{Acknowledgement}
Mateusz Skomra was partially supported by a PGMO young researcher grant of Fondation Mathématique Jacques Hadamard. The four authors thank the reviewers of the ICALP version of this work, as well as the reviewer of the present paper,
for helpful comments.
\bibliographystyle{elsarticle-num}
\bibliography{phd_bibliography}
\end{document}